\documentclass{vldb}

\usepackage{times,amsmath}
\usepackage{float}
\usepackage{graphicx}
\usepackage{enumerate}
\usepackage{subfigure}
\usepackage{url}
\usepackage{multirow}
\usepackage{makecell}
\usepackage{color}
\usepackage{balance}
\usepackage{epstopdf}

\usepackage[linesnumbered,vlined,ruled]{algorithm2e}

\newcommand{\cut}[1]{}
\newcommand{\shrink}{\vskip -2ex}

\newtheorem{theorem}{Theorem}

\newtheorem{example}{Example}
\newtheorem{lemma}{Lemma}

\newtheorem{assumption}{Assumption}

\def\inprobHIGH{\,{\buildrel p \over \longrightarrow}\,}
\def\inprob{\,{\inprobHIGH}\,}
\def\as{\,{\buildrel as \over \longrightarrow}\,}
\def\indistr{\,{\buildrel d \over \longrightarrow}\,}

\DeclareMathOperator{\E}{E}
\DeclareMathOperator{\Var}{Var}
\DeclareMathOperator{\Cov}{Cov}

\begin{document}

\title{Uncertainty Aware Query Execution Time Prediction}

\author{
\alignauthor{
Wentao Wu$^{\dagger}$
\hspace{.8cm}
Xi Wu$^\dagger$
\hspace{.8cm}
Hakan Hac{\i}g\"um\"u\c{s}$^\ddagger$
\hspace{.8cm}
Jeffrey F. Naughton$^\dagger$}\\
\vspace{.2cm}
       $^{\dagger}$Department of Computer Sciences, University of Wisconsin-Madison\\
       $^{\ddagger}$NEC Laboratories America\\
       \vspace{.2cm}
       \{wentaowu, xiwu, naughton\}@cs.wisc.edu, hakan@nec-labs.com
}

\maketitle

\begin{abstract}
Predicting query execution time is a fundamental issue underlying many database management tasks. Existing predictors rely on information such as cardinality estimates and system performance constants that are difficult to know exactly. As a result, accurate prediction still remains elusive for many queries. However, existing predictors provide a single, point estimate of the true execution time, but fail to characterize the uncertainty in the prediction. In this paper, we take a first step towards providing uncertainty information along with query execution time predictions. We use the query optimizer's cost model to represent the query execution time as a function of the selectivities of operators in the query plan as well as the constants that describe the cost of CPU and I/O operations in the system. By treating these quantities as random variables rather than constants, we show that with low overhead we can infer the distribution of likely prediction errors. We further show that the estimated prediction errors by our proposed techniques are strongly correlated with the actual prediction errors.
\end{abstract}

\section{Introduction}

The problem of predicting query execution time has received a great deal of recent research attention (e.g.,~\cite{AhmadDAB11-edbt,AkdereCRUZ12-brown-icde,DugganCPU11,Ganapathi-berkeley09,WuCHN13,WuCZTHN13}). Knowledge about query execution time is essential to many important database management issues, including query optimization, admission control~\cite{Tozer-QCop,Xiong-ActiveSLA}, query scheduling~\cite{ChiHHN13}, and system sizing~\cite{Wasserman-dbSizing}. Existing predictors rely on information such as cardinality estimates and system performance constants that are difficult to know exactly. As a result, accurate prediction remains elusive for many queries. However, existing predictors provide a single, point estimate of the true execution time, but fail to characterize the \emph{uncertainty} in the prediction.

It is a general principle that if there is uncertainty in the estimate of a quantity, systems or individuals using the estimate can benefit from information about this uncertainty. (As a simple but ubiquitous example, opinion polls cannot be reliably interpreted without considering the uncertainty bounds on their results.) In view of this, it is somewhat surprising that something as foundational as query running time estimation typically does not provide any information about the uncertainty embedded in the estimates.

There is already some early work indicating that providing this uncertainty information could be useful.  For example, in approximate
query answering~\cite{Hellerstein-online97,Jermaine-dbo07}, approximate query results are accompanied by error bars to indicate the confidence in the estimates. It stands to reason that other user-facing running time estimation tasks, for example, query progress indicators~\cite{ChaudhuriNR04,LuoNEW04}, could also benefit from similar mechanisms regarding uncertainty.  Other examples include robust query processing and optimization techniques (e.g.,~\cite{BabcockC05,ChuHS99,Graefe11,GraefeW89,MarklRSLP04,UnterbrunnerGAFK09}) and distribution-based query schedulers~\cite{ChiHHN13}. We suspect that if uncertainty information were widely available many more applications would emerge.

In this paper, we take a first step towards providing uncertainty information along with query execution time predictions. In particular, rather than just reporting a point estimate, we provide a distribution of likely running times. There is a subtlety in semantics involved here --- the issue is not ``if we run this query 100 times what do we think the distribution of running times will be?'' Rather, we are reporting
``what are the likelihoods that the actual running time of this query would fall into certain confidence intervals?'' As a concrete example, the distribution conveys information such as ``I believe, with probability 70\%, the running time of this query should be between 10s and 20s.''

Building on top of our previous work~\cite{WuCZTHN13}, we use query optimizers' cost models to represent the query execution time as a function of selectivities of operators in the query plan as well as basic system performance parameters such as the unit cost of a single CPU or I/O operation. However, our approach here is different from that in~\cite{WuCZTHN13} --- we treat these quantities as random variables rather than fixed constants. We then use sampling based approaches to estimate the distributions of these random variables. Based on that, we further develop analytic techniques to infer the distribution of likely running times.

In more detail, for specificity consider the cost model used by the query optimizer of PostgreSQL:

\begin{example}[PostgreSQL's Cost Model]
PostgreSQL\\ estimates the execution runtime overhead $t_O$ of an operator $O$ (e.g., scan, sort, join, etc.) as follows:
\begin{equation}\label{eq:cost-model}
t_O=n_s\cdot c_s + n_r\cdot c_r + n_t\cdot c_t + n_i\cdot c_i + n_o\cdot c_o.
\end{equation}
\end{example}
Here the $c$'s are \emph{cost units} described in Table~\ref{tab:cost-units}. Accordingly, the $n$'s are then the number of pages sequentially scanned, the number of pages randomly accessed, and so on, during the execution of $O$.
The total estimated overhead $t_q$ of a query $q$ is simply the sum of the costs of the individual operators in its query plan. Moreover, as illustrated in~\cite{WuCZTHN13}, the $n$'s are actually functions of the input/output cardinalities (or equivalently, selectivities) of the operators. As a result, we can further represent $t_q$ as a function of the cost units $\mathbf{c}$ and the selectivities $\mathbf{X}$, namely,
\begin{equation}\label{eq:tquery}
t_q=\sum_{O\in Plan(q)}t_O=g(\mathbf{c},\mathbf{X}).
\end{equation}

\begin{table}
\centering
\begin{tabular}{|l|l|}
\hline
$c$ & Description \\
\hline
$c_s$ & The I/O cost to \emph{sequentially} access a page\\
$c_r$ & The I/O cost to \emph{randomly} access a page\\
$c_t$ & The CPU cost to process a \emph{tuple} \\
$c_i$ & The CPU cost to process a tuple via \emph{index} access\\
$c_o$ & The CPU cost to perform an \emph{operation} (e.g., hash)\\
\hline
\end{tabular}
\caption{Cost units in PostgreSQL's cost model}
\label{tab:cost-units}
\shrink
\end{table}

Perfect predictions therefore rely on three assumptions: (i) the $c$'s are accurate; (ii) the $X$'s are accurate; and (iii) $g$ is itself accurate. Unfortunately, none of these holds in practice. First, the $c$'s are inherently random. For example, the value of $c_r$ may vary for different disk pages accessed by a query, depending on where the pages are located on disk.
Second, accurate selectivity estimation is often challenging, though significant progress has been made.
Third, the equations and functions modeling query execution make approximations and simplifications so they could make errors. For instance, Equation~(\ref{eq:cost-model}) does not consider the possible interleaving of CPU and I/O operations during runtime.

To quantify the uncertainty in the prediction, we therefore need to consider potential errors in all three parts of the running time estimation formula. It turns out that the errors in the $c$'s, the $X$'s, and $g$ are inherently different. The errors in the $c$'s result from fluctuations in the system state and/or variances in the way the system performs for different parts of different queries. (That is, for example, the cost of a random I/O may differ substantially from operator to operator and from query to query.) We therefore model the $c$'s as random variables and extend our previous calibration framework~\cite{WuCZTHN13} to obtain their distributions. The errors in the $X$'s arise from selectivity estimation errors. We therefore also model these as random variables and consider sampling-based approaches to estimate their variance. The errors in $g$, however, result from simplifications or errors made by the designer of the cost model and are out of the scope of this work. We show in our experiments that even imperfect cost model functions are useful for estimating uncertainty in predictions.

Based on the idea of treating the $c$'s and the $X$'s as random variables rather than constants, the predicted execution time $t_q$ is then also a random variable so that we can estimate its distribution. A couple of challenges arise immediately. First, unlike the case of providing a point estimate of $t_q$, knowing that $t_q$ is ``some'' function of the $c$'s and the $X$'s is insufficient if we want to infer the distribution of $t_q$ --- we need to know the \emph{explicit} form of $g$. By Equation~(\ref{eq:tquery}), $g$ relies on cost functions that map the $X$'s to the $n$'s. As a result, for concreteness we have to choose some specific cost model. Here, for simplicity and generality, we leverage the notion of \emph{logical} cost functions~\cite{Du-calib92} rather than the cost functions of any specific optimizer. The observation is that the costs of an operator can be specified according to its logical execution. For instance, the number of CPU operations of the in-memory \emph{sort} operator could be specified as $n_o=aN\log N$, where $N$ is the input cardinality. Second, while we can show that the distribution of $t_q$ is asymptotically normal based on our current ways of modeling the $c$'s and the $X$'s, determining the parameters of the normal distribution (i.e., the mean and variance) is difficult for non-trivial queries with deep query trees. The challenge arises from correlations between selectivity estimates derived by using shared samples. We present a detailed analysis of the correlations and develop techniques to either directly compute or provide upper bounds for the covariances with respect to the presence of correlations. Finally, providing estimates to distributions of likely running times is desirable only if it can be achieved with low overhead. We show that it is the case for our proposed techniques --- the overhead is almost the same as that of the predictor in~\cite{WuCZTHN13} which only provides point estimates.

Since our approach makes a number of approximations when computing the distribution of running time estimates, an important question is how
accurate the estimated distribution is. An intuitively appealing experiment is the following: run the query multiple times, measure the
distribution of its running times, and see if this matches the estimated distribution.  But this is not a reasonable approach due to
the subtlety we mentioned earlier. The estimated distribution we calculate is not the expected distribution of the actual query running
time, it is the distribution of running times our estimator expects due to uncertainties in its estimation process. To see this another
way, note that cardinality estimation error is a major source of running time estimation error. But when the query is actually run, it does not appear at all --- the query execution of course observes the true cardinalities, which are identical every time it is run.

Speaking informally, what our predicted running time distribution captures is the ``self-awareness'' of our estimator. Suppose that embedded in the estimate is a dependence on what our estimator knows is a very inaccurate estimate. Then the estimator knows that while it gives a specific point estimate for the running time (the mean of a distribution), it is likely that the true running time will be far away from the estimate, and it captures this by indicating a distribution with a large variance.

So our task in evaluating our approach is to answer the following question: how closely does the variance of our estimated distribution
of running times correspond to the observed errors in our estimates (when compared with true running times)? To answer this question, we
estimate the running times for and run a large number of different queries and test the agreement between the observed errors and the
predicted distribution of running times, where ``agreement'' means that larger variations correspond to more inaccurate estimates.

In more detail, we report two metrics over a large number of queries: (M1) the correlation between the standard deviations of the estimated
distributions and the actual prediction errors; and (M2) the proximity between the inferred and observed distributions of prediction
errors. We show that (R1) the correlation is strong; and (R2) the two distributions are close. Intuitively, (R1) is \emph{qualitative}; it
suggests that one can judge if the prediction errors will be small or large based on the standard deviations of the estimated
distributions. (R2) is more \emph{quantitative}; it further suggests that the likelihoods of prediction errors are specified by the distributions as well. We therefore conclude that the estimated distributions do a reasonable job as indicators of prediction errors.

We start by presenting terminology and notation used throughout the paper in Section~\ref{sec:preliminary}. We then present the details of how to estimate the distributions of the $c$'s and the $X$'s (Section~\ref{sec:distribution}), the explicit form of $g$ (Section~\ref{sec:costfunc}), and the distribution of $t_q$ (Section~\ref{sec:uncertainty}). We further present experimental evaluation results in Section~\ref{sec:experiment}, discuss related work in Section~\ref{sec:relatedwork}, and conclude the paper in Section~\ref{sec:conclusion}.

\section{Preliminaries}\label{sec:preliminary}

In most current DBMS implementations, the operators are either unary or binary. Therefore, we can model a query plan with a rooted \emph{binary tree}. Consider an operator $O$ in the query plan. We use $O_l$ and $O_r$ to represent its \emph{left} and \emph{right} child operator, and use $N_l$ and $N_r$ to denote its left and right input cardinality. If $O$ is unary, then $O_r$ does not exist and thus $N_r=0$. We further use $M$ to denote $O$'s output cardinality.

Let $\mathcal{T}$ be the subtree rooted at the operator $O$, and let $\mathcal{R}$ be the (multi)set of relations accessed by the leaf nodes of $\mathcal{T}$. Note that the leaf nodes in a query plan must be \emph{scan} operators that access the underlying tables.\footnote{We use ``\emph{relation}'' and ``\emph{table}'' interchangeably in this paper since our discussion does not depend on the \emph{set}/\emph{bag} semantics.} We call $\mathcal{R}$ the \emph{leaf tables} of $O$. Let $|\mathcal{R}|=\prod_{R\in\mathcal{R}}|R|$. We define the \emph{selectivity} $X$ of $O$ to be:
\begin{equation}\label{eq:selectivity}
X=\frac{M}{|\mathcal{R}|}=\frac{M}{\prod_{R\in\mathcal{R}}|R|}.
\end{equation}

\begin{example}[Selectivity]
Consider the query plan in Figure~\ref{fig:query-plan}. $O_1$, $O_2$, and $O_3$ are scan operators that access three underlying tables $R_1$, $R_2$, and $R_3$, and $O_4$ and $O_5$ are join operators. The selectivity of $O_1$, for instance, is $X_1=\frac{M_1}{|R_1|}$, whereas the selectivity of $O_4$ is $X_4=\frac{M_4}{|R_1|\cdot|R_2|}$.
\end{example}

\begin{figure}
\centering
\includegraphics[width=0.9\columnwidth]{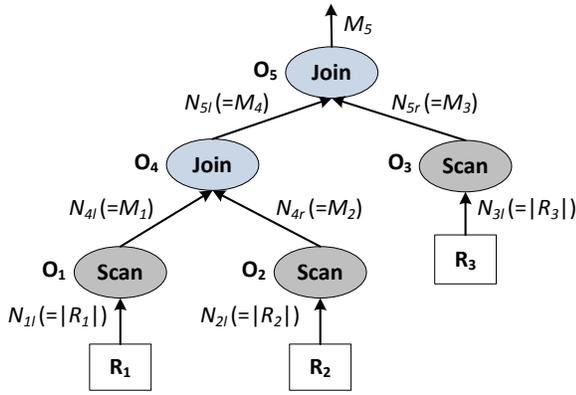}
\caption{Example query plan}
\label{fig:query-plan}
\shrink
\end{figure}

We summarize the above notation in Table~\ref{tab:notations} for convenience of reference. Since the $n$'s in Equation~(\ref{eq:cost-model}) are functions of input/output cardinalities of the operators (we discuss different types of cost functions in Section~\ref{sec:costfunc:types}), it is clear that the $n$'s are also functions of the selectivities (i.e., the $X$'s) defined here. Based on Equation~(\ref{eq:tquery}), $t_q$ is therefore a function of the $c$'s and the $X$'s. We next discuss how to measure the uncertainties in these parameters.

\begin{table}
\centering
\begin{tabular}{|l|l|}
\hline
Notation & Description \\
\hline
$O$ & An operator in the query plan\\
$O_l$ ($O_r$) & The left (right) child operator of $O$\\
$N_l$ ($N_r$) & The left (right) input cardinality of $O$\\
$M$ & The output cardinality of $O$\\
$\mathcal{R}$ & The leaf tables of $O$\\
$X$ & The selectivity of $O$\\
$\mathcal{T}$ & The subtree rooted at $O$\\
$Desc(O)$ & The descendant operators of $O$ in $\mathcal{T}$\\
\hline
\end{tabular}
\caption{Terminology and notation}
\label{tab:notations}
\shrink
\end{table}

\section{Input Distributions}\label{sec:distribution}

To learn the distribution of $t_q$, we first need to know the distributions of the $c$'s and the $X$'s. We do this by extending the framework in our previous work~\cite{WuCZTHN13}.

\subsection{Distributions of the $c$'s}

In~\cite{WuCZTHN13}, we designed dedicated calibration queries for each $c$. Consider the following example:

\begin{example}[Calibration Query]
Suppose that we want to know the value of $c_t$, namely, the CPU cost of processing one tuple. We can use the calibration query \verb|SELECT * FROM R|, where \verb|R| is some table whose size is known and is loaded into memory. Since this query only involves $c_t$, its execution time $\tau$ can be expressed as $\tau=|\verb|R||\cdot c_t$. We can then run the query, record $\tau$, and compute $c_t$ from this equation.
\end{example}

Note that we can use different \verb|R|'s here, and different \verb|R|'s may give us different $c_t$'s. We can think of these observed values as i.i.d. samples from the distribution of $c_t$, and in~\cite{WuCZTHN13} we used the sample mean as our estimate of $c_t$. To quantify the uncertainty in $c_t$, it would make more sense to treat $c_t$ as a random variable rather than a constant. We assume that the distribution of $c_t$ is normal (i.e., Gaussian), for intuitively the CPU speed is likely to be stable and centered around its mean value. Now let $c_t\sim\mathcal{N}(\mu_t, \sigma_t^2)$. It is then a common practice to use the mean and variance of the observed $c_t$'s as estimates for $\mu_t$ and $\sigma_t^2$.

In general, we can apply similar arguments to all the five cost units. Due to space limitations, readers are referred to~\cite{WuCZTHN13} for more details on the calibration procedure. In~\cite{WuCZTHN13} we only calculated the mean, not the variance, but the extension to the variance is straightforward.

\subsection{Distributions of the $X$'s}

The uncertainties in the $X$'s are quite different from those in the $c$'s. The uncertainties in the $c$'s are due to unavoidable fluctuations in hardware execution speeds. In other words, the $c$'s are inherently random. However, the $X$'s are actually fixed numbers --- if we run the query we should always obtain the same ground truths for the $X$'s. The uncertainties in the $X$'s really come from the fact that so far we do not have a perfect selectivity estimator. How to quantify the uncertainties in the $X$'s therefore depends on the nature of the selectivity estimator used. Here we extend the sampling-based approach used in~\cite{WuCZTHN13}, which was first proposed by Haas et al.~\cite{Haas-sample96}. It provides a mathematically rigorous way to quantify potential errors in selectivity estimates. It remains interesting future work to investigate the possibility of extending other alternative estimators such as those based on histograms.

\subsubsection{A Sampling-Based Selectivity Estimator}\label{sec:distribution:sel:sn2}

Suppose that we have a database consisting of $K$ relations $R_1$, ..., $R_K$, where
$R_k$ is partitioned into $m_k$ blocks each with size $N_k$, namely, $|R_k|=m_kN_k$.
Without loss of generality, let $q$ be a selection-join query over $R_1$, ..., $R_K$,
and let $B(k, j)$ be the
$j$-th block of relation $k$ ($1\leq j\leq m_k$, and $1\leq k\leq
K$). Define $$\mathbf{B}(L_{1,i_1}, ..., L_{K,i_K})=B(1,L_{1,i_1})\times\cdots\times B(K,L_{K,i_K}),$$ where
$B(k,L_{k,i_k})$ is the block (with index $L_{k,i_k}$) randomly picked
from the relation $R_k$ in the $i_k$-th sampling step. After $n$ steps, we can obtain $n^K$ such samples (notice that these samples are not independent), and the estimator is defined as
\begin{equation}\label{eq:estimator-cp}
\rho_n=\frac{1}{n^K}\sum_{i_1=1}^{n}\cdots\sum_{i_K=1}^{n}\rho_{\mathbf{B}(L_{1,i_1},\cdots,L_{K,i_K})}.
\end{equation}
Here $\rho_n$ is the estimated selectivity of $q$ (after $n$ sampling steps), and $\rho_{\mathbf{B}}$ is the observed selectivity of $q$ over the sample $\mathbf{B}$. This estimator is shown to be both unbiased and strongly consistent for the actual selectivity $\rho$ of $q$~\cite{Haas-sample96,WuCZTHN13}.\footnote{Strong consistency is also called \emph{almost sure convergence} in probability theory (denoted as ``a.s.''). It means that the more samples we take, the closer $\rho_n$ is to $\rho$.}

By applying the Central Limit Theorem, we can show that
$$\frac{\sqrt{n}}{\sigma}\big(\rho_n-\rho\big)\indistr N(0,1).$$
That is, the distribution of $\rho_n$ is approximately normal after a large number of sampling steps~\cite{Haas-sample96}: $\rho_n\sim \mathcal{N}(\rho, \sigma_n^2)$, where $\sigma_n^2=\sigma^2 / n$ and $\sigma^2=\lim_{n\to\infty}n\Var[\rho_n]$. We present a more detailed study of $\Var[\rho_n]$ in Appendix~\ref{sec:proofs:theorem:variance}.

However, here $\sigma_n^2$ is unknown since $\sigma^2$ is unknown. In~\cite{Haas-sample96}, the authors further proposed the following estimator for $\sigma^2$:
\begin{equation}\label{eq:sn-square}
S_n^2=\sum_{k=1}^{K}\bigg(\frac{1}{n-1}\sum_{j=1}^n(Q_{k,j,n}/n^{K-1}-\rho_n)^2\bigg),
\end{equation}
for $n\geq 2$ (we set $S_1^2=0$). Here
\begin{equation}\label{eq:Q}
Q_{k,j,n}=\sum_{(i_1,...,i_K)\in\Omega_k^{(n)}(j)}\rho_{\mathbf{B}(L_{1,i_1},...,L_{K,i_K})},
\end{equation}
where $\Omega_k^{(n)}(j)=\{(i_1,...,i_K)\in\{1,...,n\}^K:i_k=j\}$. It can be shown that $\lim_{n\to\infty}S_n^2=\sigma^2$ a.s. As a result, it is reasonable to approximate $\sigma^2$ with $S_n^2$ when $n$ is large. So $\sigma_n^2\approx S_n^2 / n$.

\subsubsection{Efficient Computation of $S_n^2$}\label{sec:distribution:Sn2}

Efficiency is crucial for a predictor to be practically useful. We have discussed efficient implementation of $\rho_n$ in~\cite{WuCZTHN13}. Taking samples at runtime might not be acceptable since it will result in too many random disk I/O's. Therefore, we instead take samples off-line and store them as materialized views (i.e., sample tables). In the following presentation, we use $R^s$ to denote the sample table of a relation $R$. In~\cite{WuCZTHN13}, we further showed that, given a selection-join query, we can estimate the selectivities of all the selections and joins by running the original query plan over the sample tables once. The trick is that, since the block size is not specified when partitioning the relations, it could be arbitrary. We can then let a block be a single tuple so that the cross-product of sample blocks is reduced to the cross-product of sample tuples.

\begin{example}[Implementation of $\rho_n$]\label{ex:sel-est}
Let us consider the query plan in Figure~\ref{fig:query-plan} again. Based on the tuple-level partitioning scheme, by Equation~(\ref{eq:estimator-cp}) we can simply estimate $X_4$ and $X_5$ as
$$\widehat{X}_4=\frac{|R_1^s\bowtie R_2^s|}{|R_1^s|\cdot |R_2^s|}\quad \textrm{and} \quad \widehat{X}_5=\frac{|R_1^s\bowtie R_2^s\bowtie R_3^s|}{|R_1^s|\cdot |R_2^s|\cdot |R_3^s|}.$$
Also note that we can compute the two numerators by running the query plan over the sample relations $R_1^s$, $R_2^s$, and $R_3^s$ once. That is, to compute $R_1^s\bowtie R_2^s\bowtie R_3^s$, we reuse the join results from $R_1^s\bowtie R_2^s$ that has been computed when estimating $X_4$.
\end{example}

We now extend the above framework to further compute $S_n^2$. For this sake we need to know how to compute the $Q_{k,j,n}$'s in Equation~(\ref{eq:sn-square}). Let us consider the cases when an operator represents a selection (i.e., a scan), a two-way join, or a multi-way join query.

\paragraph*{Selection}
In this case, $K=1$ and by Equation~(\ref{eq:Q}) $Q_{k,j,n}$ is reduced to $Q_{1,j,n}=\rho_{B}(L_{1,j})$.
Therefore, $S_n^2$ can be simplified as
$$S_n^2=\frac{1}{n-1}\sum_{j=1}^n(\rho_{B}(L_{1,j})-\rho_n)^2.$$
Since a block here is just a tuple, $\rho_{B}(L_{1,j})=0$ or $\rho_{B}(L_{1,j})=1$. We thus have
\begin{eqnarray*}
S_n^2&=&\frac{1}{n-1}\big(\sum_{\rho_{B}(L_{1,j})=0}\rho_n^2+\sum_{\rho_{B}(L_{1,j})=1}(1-\rho_n)^2\big)\\
&=&\frac{1}{n-1}\big((n-M)\rho_n^2 + M(1-\rho_n)^2\big),
\end{eqnarray*}
where $M$ is the number of output tuples from the selection. When $n$ is large, $n\approx n-1$, so we have
$$S_n^2\approx(1-\frac{M}{n})\rho_n^2 + \frac{M}{n}(1-\rho_n)^2=\rho_n(1-\rho_n),$$
by noticing that $\rho_n=\frac{M}{n}$. Hence $S_n^2$ is directly computable for a scan operator once we know its estimated selectivity $\rho_n$.

\paragraph*{Two-way Join}
Consider a join $R_1\bowtie R_2$. In this case, $Q_{k,j,n}$ ($k=1,2$) can be reduced to
$$Q_{1,j,n}=\sum_{i_2=1}^{n}\rho_{\mathbf{B}}(L_{1,j},L_{2,i_2})\textrm{ and }
Q_{2,j,n}=\sum_{i_1=1}^{n}\rho_{\mathbf{B}}(L_{1,i_1},L_{2,j}).$$
Again, since a block here is just a tuple, $\rho_{\mathbf{B}}$ is either 0 or 1. It is then equivalent to computing the following two quantities:
\begin{itemize}
\item $Q_{1,j,n}=|\{t_{1j}\}\bowtie R_2^s|$, where $t_{1j}$ is the $j$th tuple of $R_1^s$;
\item $Q_{2,j,n}=|R_1^s\bowtie\{t_{2j}\}|$, where $t_{2j}$ is the $j$th tuple of $R_2^s$.
\end{itemize}
That is, to compute $Q_{k,j,n}$ ($k=1,2$), conceptually we need to join each sample tuple of one relation with all the sample tuples of the other relation. However, directly performing this is quite expensive, for we need to do $2n$ joins here.

We seek a more efficient solution. Recall that we need to join $R_1^s$ and $R_2^s$ to compute $\rho_n$. Let $R^s=R_1^s\bowtie R_2^s$. Consider any $t\in R^s$. $t$ must satisfy $t=t_{1i}\bowtie t_{2j}$, where $t_{1i}\in R_1^s$ and $t_{2j}\in R_2^s$. Then $t$ contributes 1 to $Q_{1,i,n}$ and 1 to $Q_{2,j,n}$. On the other hand, any $t$ in $R_1^s\times R_2^s$ but not in $R^s$ will contribute nothing to the $Q$'s. Based on this observation, we only need to scan the tuples in $R^s$ and increment the corresponding $Q$'s. The remaining problem is how to know the indexes $i$ and $j$ as in $t=t_{1i}\bowtie t_{2j}$. For this purpose, we assign an \emph{identifier} to each tuple in the sample tables when taking the samples. This is akin to the idea in data provenance research where tuples are annotated to help tracking the lineages of the query results~\cite{GreenKT07}.

\paragraph*{Multi-way Joins}
The approach of processing two-way joins can be easily generalized to handle multi-way joins. Now we have
$$Q_{k,j,n}=|R_1^s\bowtie\cdots\bowtie\{t_{kj}\}\bowtie\cdots\bowtie R_K^s|.$$
As a result, if we let $R^s=R_1^s\bowtie\cdots\bowtie R_K^s$, then any $t\in R^s$ satisfies $t=t_{1i_1}\bowtie\cdots\bowtie t_{Ki_K}$. $t\in R_1^s\times\cdots\times R_K^s$ will contribute 1 to each $Q_{k,i_k,n}$ ($1\leq k\leq K$) if and only if $t\in R^s$.
Therefore, as before, we can just simply scan $R^s$ and increment the corresponding $Q$'s when processing each tuple.

\paragraph*{Putting It Together} Algorithm~\ref{alg:sel-est} summarizes the procedure of computing $\rho_n$ and $S_n^2$ for a single operator $O$. It is straightforward to incorporate it into the previous framework where the selectivities of the operators are refined in a bottom-up fashion (Appendix~\ref{sec:framework}). We discuss some implementation details in the following.

\begin{algorithm}
  \SetAlgoLined
  \KwIn{$O$, an operator; $\mathcal{R}^s=\{R_1^s,...,R_K^s\}$, the sample tables; $Agg$, if some $O'\in Desc(O)$ is an aggregate}
  \KwOut{$\rho_n$, estimated selectivity of $O$; $S_n^2$, sample variance}
  \SetAlgoLined
  $R^s\leftarrow RunOperator(O,\mathcal{R}^s)$\;

  \uIf{$Agg$}{
    $M\leftarrow CardinalityByOptimizer(O)$\;
    $\rho_n\leftarrow\frac{M}{\prod_{k=1}^K|R_k|}$\;
    $S_n^2\leftarrow 0$\;
  }\uElseIf{$O$ is a scan} {
    $\rho_n\leftarrow\frac{|R^s|}{|R_1^s|}$\;
    $S_n^2\leftarrow\rho_n(1-\rho_n)$\;
  }\uElseIf{$O$ is a join} {
    $\rho_n\leftarrow\frac{|R^s|}{\prod_{k=1}^K|R_k^s|}$\;
    \ForEach{$t=t_{1i_1}\bowtie\cdots\bowtie t_{Ki_K}\in R^s$} {
        $Q_{k,i_k,n}\leftarrow Q_{k,i_k,n}+1$, for $1\leq k\leq K$\;
    }
    $S_n^2\leftarrow\sum_{k=1}^{K}\bigg(\frac{1}{n-1}\sum_{j=1}^n(Q_{k,j,n}/n^{K-1}-\rho_n)^2\bigg)$;
  }\Else {
    $\rho_n\leftarrow \hat{\mu}_l$, $S_n^2\leftarrow \hat{\sigma}_l^2$; // Let $X_l\sim \mathcal{N}(\hat{\mu}_l, \hat{\sigma}_l^2)$.
  }

  \Return{$\rho_n$ \emph{and} $S_n^2$\;}
  \caption{Computation of $\rho_n$ and $S_n^2$}
\label{alg:sel-est}
\end{algorithm}

First, the selectivity estimator cannot work for operators such as \emph{aggregates}. Our current strategy is to use the original cardinality estimates from the optimizer to compute $\rho_n$, and we simply set $S_n^2$ to be 0 for these operators (lines 3 to 5). This may cause inaccuracy in the prediction as well as our estimate of its uncertainty, if the optimizer does a poor job in estimating the cardinalities. However, we find that it works reasonably well in our experiments. Nonetheless, we are working to incorporate sampling-based estimators for aggregates (e.g., the GEE estimator~\cite{Charikar-sample00}) into our current framework.

Second, to compute the $Q_{k,i_k,n}$'s, we maintain a hash map $H_k$ for each $k$ with $i_k$'s the keys and $Q_{k,i_k,n}$'s the values. The size of $H_k$ is upper bounded by $|R_k^s|$ and usually is much smaller.

Third, for simplicity of exposition, in Algorithm~\ref{alg:sel-est} we first compute the whole $R^s$ and then scan it. In practice we actually do not need to do this. Typical join operators, such as \emph{merge join}, \emph{hash join}, and \emph{nested-loop join}, usually compute join results on the fly. Once a join tuple is computed, we can immediately postprocess it by increasing the corresponding $Q_{k,i_k,n}$'s. Therefore, we can avoid the additional memory overhead of caching intermediate join results, which might be large even if the sample tables are small.

\section{Cost Functions}\label{sec:costfunc}

By Equation~(\ref{eq:tquery}), to infer the distribution of $t_q$ for a query $q$, we also need to know the explicit form of $g$. According to Equation~(\ref{eq:cost-model}), $g$ relies on the cost functions of operators that map the selectivities to the $n$'s. As mentioned in the introduction, we use logical cost functions in our work. While different DBMS may differ in their implementations of a particular operator, e.g., nested-loop join, they follow the same execution logic and therefore have the same logical cost function. In the following, we first present a detailed study of representative cost functions. We then formulate the computation of cost functions as an optimization problem that seeks the best fit for the unknown coefficients, and we use standard quadratic programming techniques to solve this problem.

\subsection{Types of Functions}\label{sec:costfunc:types}

We consider the following types of cost functions in this paper:

\begin{enumerate}[(C1)]
\item $f=a_0$: The cost function is a constant. For instance, since a sequential scan has no random disk reads, $n_r=0$.
\item $f=a_0M + a_1$: The cost function is linear with respect to the \emph{output} cardinality. For example, the number of random reads of an index-based table scan falls into this category, which is proportional to the number of qualified tuples that pass the selection predicate.
\item $f=a_0N_{l} + a_1$: The cost function is linear with respect to the \emph{input} cardinality. This happens for unary operators that process each input tuple once. For example, \emph{materialization} is such an operator that creates a buffer to cache the intermediate results.
\item $f=a_0N_{l}^2 + a_1N_{l} + a_2$: The cost function is nonlinear with respect to the \emph{input} cardinality. For instance, the number of CPU operations (i.e., $c_o$) performed by a \emph{sort} operator is proportional to $N_{l}\log N_{l}$. While different nonlinear unary operators may have specific cost functions, we choose to only use quadratic polynomials based on the following observations:
    \begin{itemize}
    \item It is quite general to approximate the nonlinear cost functions used by current relational operators. First, as long as a function is smooth (i.e., it has continuous derivatives up to some desired order), it can be approximated by using the well-known Taylor series, which is basically a polynomial of the input variable. Second, for efficiency reasons, the overhead of an operator usually does not go beyond quadratic of its input cardinality --- we are not aware of any operator implementation whose time complexity is $\omega(N^2)$. Similar observations have been made in~\cite{DDH08}.
    \item Compared with functions such as logarithmic ones, polynomials are mathematically much easier to manipulate. Since we need to further infer the distribution of the predicted query execution time based on the cost functions, this greatly simplifies the derivations.
    \end{itemize}
\item $f=a_0N_{l} + a_1N_{r}+a_2$: This cost function is linear with respect to the \emph{input} cardinalities when the operator is binary. An interesting observation here is that the cost functions in the case of binary operators are not necessarily nonlinear. For example, the number of I/O's involved in a hash join is only proportional to the number of input tuples.
\item $f=a_0N_{l}N_{r} + a_1N_{l} + a_2N_{r}+a_3$: The cost function here also involves the product of the left and right input cardinalities of a binary operator. This happens typically in a nested-loop join, which iterates over the inner (i.e., the right) input table multiple times with respect to the number of rows in the outer (i.e., the left) input table.
\end{enumerate}

It is straightforward to translate these cost functions in terms of selectivities. Specifically, we have $N_{l}=|\mathcal{R}_l|X_l$, $N_{r}=|\mathcal{R}_r|X_r$, and $M=|\mathcal{R}|X$. The above six cost functions can be rewritten as
\begin{enumerate}[(C1')]
\item $f=b_0$, where $b_0=a_0$.
\item $f=b_0X + b_1$, where $b_0=a_0|\mathcal{R}|$ and $b_1=a_1$.
\item $f=b_0X_l + b_1$, where $b_0=a_0|\mathcal{R}_l|$ and $b_1=a_1$.
\item $f=b_0X_l^2 + b_1X_l + b_2$, where $b_0=a_0|\mathcal{R}_l|^2$, $b_1=a_1|\mathcal{R}_l|$, and $b_2=a_2$.
\item $f=b_0X_l + b_1X_r + b_2$, where $b_0=a_0|\mathcal{R}_l|$, $b_1=a_1|\mathcal{R}_r|$, and $b_2=a_2$.
\item $f=b_0X_lX_r + b_1X_l + b_2X_r + b_3$, where $b_0=a_0|\mathcal{R}_l|\cdot|\mathcal{R}_r|$, $b_1=a_1|\mathcal{R}_l|$, $b_2=a_2|\mathcal{R}_r|$, and $b_3=a_3$.
\end{enumerate}

\subsection{Computation of Cost Functions}

To compute the cost functions, we use an approach that is similar to the one proposed in~\cite{DDH08}. Regarding the types of cost functions we considered, the only unknowns given the selectivity estimates are the coefficients in the functions (i.e., the $b$'s). Moreover, notice that $f$ is a \emph{linear} function of the $b$'s once the selectivities are given. We can then collect a number of $f$ values by feeding in the cost model with different $X$'s and find the best fit for the $b$'s.

As an example, consider (C4'). Suppose that we invoke the cost model $m$ times and obtain $m$ points: $$\{(X_{l1},f_{1}),...,(X_{lm},f_{m})\}.$$
Let $\mathbf{y}=(f_{1},...,f_{m})$, $\mathbf{b}=(b_0, b_1, b_2)$, and
\begin{displaymath}
\mathbf{A} =
\left( \begin{array}{ccc}
X_{l1}^2 & X_{l1} & 1 \\
\vdots & \vdots & \vdots \\
X_{lm}^2 & X_{lm} & 1
\end{array} \right).
\end{displaymath}
The optimization problem we are concerned with is:

\begin{equation*}
\begin{aligned}
& \underset{\mathbf{b}}{\text{minimize}}
& & \|\mathbf{A}\mathbf{b}-\mathbf{y}\| \\
& \text{subject to}
& & b_i\geq 0, \; i = 0, 1.
\end{aligned}
\end{equation*}
Note that we require $b_0$ and $b_1$ be nonnegative since they have the natural semantics in the cost functions as the amount of work with respect to the corresponding terms. For example, $b_1X_l=a_1N_l$ is the amount of work that is proportional to the input cardinality. To solve this quadratic programming problem, we use the \verb|qpsolve| function of Scilab~\cite{scilab}. Other equivalent solvers could also be used.

The remaining problem is how to pick these $(X, f)$'s. In theory, one could arbitrarily pick the $X$'s from $[0,1]$ to obtain the corresponding $f$'s as long as we have more points than unknowns. Although more points usually mean we can have better fittings, in practice we cannot afford too many points due to the efficiency requirements when making the prediction. On the other hand, given that the $X$'s here follow normal distributions and the variances are usually small when the sample size is large, the likely selectivity estimates are usually concentrated in a much shorter interval than $[0,1]$. Intuitively, we should take more points within this interval, for we can then have a more accurate view of the shape of the cost function restricted to this interval. Therefore, in our current implementation, we adopt the following strategy.

Let $X\sim \mathcal{N}(\mu,\sigma^2)$. Consider the interval $\mathcal{I}=[\mu-3\sigma, \mu+3\sigma]$. It is well known that $\Pr(X\in \mathcal{I})\approx 0.997$, which means the probability that $X$ falls out of $\mathcal{I}$ is less than 0.3\%. We then proceed by partitioning $\mathcal{I}$ into $W$ subintervals of equal width, and pick the $W+1$ boundary $X$'s to invoke the cost model. Generalizing this idea to binary cost functions is straightforward. Suppose $X_l\sim \mathcal{N}(\mu_l,\sigma_l^2)$ and $X_r\sim \mathcal{N}(\mu_r,\sigma_r^2)$. Let $\mathcal{I}_l=[\mu_l-3\sigma_l, \mu_l+3\sigma_l]$ and $\mathcal{I}_r=[\mu_r-3\sigma_r, \mu_r+3\sigma_r]$. We then partition $\mathcal{I}_l\times\mathcal{I}_r$ into a $W\times W$ grid and obtain $(W+1)\times (W+1)$ points $(X_l, X_r)$ to invoke the cost model.

\section{Distribution of Running Times}\label{sec:uncertainty}

We have discussed how to estimate the distributions of input parameters (i.e., the $c$'s and the $X$'s) and how to estimate the cost functions of each operator. In this section, we discuss how to combine these two to further infer the distribution of $t_q$ for a query $q$.

Since $t_q=g(\mathbf{c},\mathbf{X})$, the distribution of $t_q$ relies on the \emph{joint} distribution of $(\mathbf{c},\mathbf{X})$.\footnote{Note that the distributions of the $c$'s and $X$'s that we obtained in Section~\ref{sec:distribution} are \emph{marginal} rather than joint.} We therefore first present a detailed analysis of the correlations between the $c$'s and the $X$'s. Based on that, we then show that the distribution of $t_q$ is asymptotically normal and thus reduce the problem to estimating the two parameters of normal distributions, i.e., the mean and variance of $t_q$. We further address the nontrivial problem of computing $\Var[t_q]$ due to correlations between selectivity estimates.

\subsection{Correlations of Input Variables}

In our current setting, it is reasonable to assume that the $c$'s and the $X$'s are independent. In the following, we analyze the correlations within the $c$'s and the $X$'s.

\subsubsection{Correlations Between Cost Units}

Since the randomness within the $c$'s comes from the variations in hardware execution speeds, by using our current framework we have no way to observe the true values of the $c$'s and thus it is impossible to obtain the exact joint distribution of the $c$'s. Nonetheless, it might be reasonable to assume the independence of the $c$'s. First, since the CPU and I/O cost units measure the speeds of different hardware devices, their values do not depend on each other. Second, within each group (i.e., CPU or I/O cost units), we used independent calibration queries for each individual cost unit.
\begin{assumption}
The $c$'s are independent of each other.
\end{assumption}

We further note here that the independence of the $c$'s depends on the cost model as well as the hardware configurations.
For instance, if certain devices are connected via the same infrastructure (e.g., a bus), then they might influence each other's communication patterns. Our current framework for calibrating the $c$'s cannot capture the correlations of the $c$'s. However, perhaps low-level tools for monitoring hardware execution status could be used for this purpose. We leave it as interesting future work to investigate such possibilities and study the effectiveness of incorporating correlation information of the $c$'s into our current framework.

\subsubsection{Correlations Between Selectivity Estimates}

The $X$'s are clearly not independent, because the same samples are used to estimate the selectivities of different operators. We next study the correlations between the $X$'s in detail.

Let $O$ and $O'$ be two operators, and $\mathcal{R}$ and $\mathcal{R'}$ be the corresponding leaf tables. Consider the two corresponding selectivity estimates $\rho_n$ and $\rho'_n$ as defined by Equation~(\ref{eq:estimator-cp}). Since the samples from each table are drawn independently, we first have:

\begin{lemma}\label{lemma:ind}
If $\mathcal{R}\cap\mathcal{R}'=\emptyset$, then $\rho_n\bot\rho'_n$.\footnote{We use $Y\bot Z$ to denote that $Y$ and $Z$ are independent.}
\end{lemma}

For binary operators, it follows from Lemma~\ref{lemma:ind} immediately that:

\begin{lemma}\label{lemma:ind-bin}
Let $O$ be binary. If $\mathcal{R}_l\cap\mathcal{R}_r=\emptyset$, then $X_l\bot X_r$.
\end{lemma}
That is, $X_l$ and $X_r$ will only be correlated if $\mathcal{R}_l$ and $\mathcal{R}_r$ share \emph{common} relations. However, in practice, we can maintain more than one sample table for each relation. When the database is large, this is affordable since the number of samples is very small compared to the database size~\cite{WuCZTHN13}. Since the samples from each relation are drawn independently, $X_l$ and $X_r$ are still independent if we use a different sample table for each appearance of a shared relation. We thus assume $X_l\bot X_r$ in the rest of the paper.

More generally, $X$ and $X'$ are independent as long as neither $O\in Desc(O')$ nor $O'\in Desc(O)$. However, the above discussion cannot be applied if $O\in Desc(O')$ (or vice versa). This is because we pass the join results from downstream joins to upstream joins when estimating the selectivities (recall Example~\ref{ex:sel-est}). So $\mathcal{R}$ and $\mathcal{R'}$ are naturally not disjoint. In fact, $\mathcal{R}\subseteq\mathcal{R}'$. To make $\rho_n$ and $\rho'_n$ independent, we need to replace each of the sample tables used in computing $\rho'_n$ with another sample table from the same relation, which basically is the same as run the query plan again on a different set of sample tables. The number of runs is then in proportion to the number of selective operators (i.e., selections and joins) in the query plan, and the runtime overhead might be prohibitive in practice. We summarize this observation as follows:

\begin{lemma}\label{lemma:correlated}
Given that multiple sample tables of the same relation can be used, $\rho_n$ and $\rho'_n$ are correlated if and only if either $O\in Desc(O')$ or vice versa.
\end{lemma}

\subsection{Asymptotic Distributions}

Now for specificity suppose that the query plan of $q$ contains $m$ operators $O_1$, ..., $O_m$. Since $t_q$ is the sum of the predicted execution time spent on each operator, it can be expressed as $t_q=\sum_{k=1}^m t_k$, where $t_k$ is the predicted execution time of $O_k$ and is itself a random variable.

We next show that $t_k$ is asymptotically normal, and then by using very similar arguments, we can show that $t_q$ is asymptotically normal as well. Since $t_k$ can be further expressed in terms of Equation~(\ref{eq:cost-model}), to learn its distribution we need to know the distributions of cost functions that map the selectivities to the $n$'s. We therefore start by discussing the distributions of the typical cost functions as presented in Section~\ref{sec:costfunc:types}.

\subsubsection{Asymptotic Distributions of Cost Functions}

In the following discussion, we assume that $X\sim \mathcal{N}(\mu, \sigma^2)$, $X_l\sim \mathcal{N}(\mu_l, \sigma_l^2)$, and $X_r\sim \mathcal{N}(\mu_r, \sigma_r^2)$. The distributions of the six types of cost functions previously discussed are as follows:
\begin{enumerate}[(C1')]
\item $f=b_0$: $f\sim \mathcal{N}(b_0, 0)$.
\item $f=b_0X + b_1$: $f\sim \mathcal{N}(b_0\mu+b_1, b_0^2\sigma^2)$.
\item $f=b_0X_l + b_1$: $f\sim \mathcal{N}(b_0\mu_l+b_1, b_0^2\sigma_l^2)$.
\item $f=b_0X_l^2 + b_1X_l + b_2$: In this case $\Pr(f)$ is not normal. Although it is possible to derive the \emph{exact} distribution of $f$ based on the distribution of $X_l$, the derivation would be very messy. Instead, we consider $f^{\mathcal{N}}\sim \mathcal{N}(\E[f], \Var[f])$ and use this to approximate $\Pr(f)$. We present the formula of $\Var[f]$ in Lemma~\ref{lemma:c4var} (proof in Appendix~\ref{sec:proofs:lemma:c4var}). Obviously, $f^{\mathcal{N}}$ and $f$ have the same expected value and variance. Moreover, we can actually show that $f^{\mathcal{N}}$ and $f$ (and therefore their corresponding distributions) are very close to each other when the number of samples is large (see Theorem~\ref{theorem:c4approx} below; the proof is in Appendix~\ref{sec:proofs:theorem:c4approx}).
\item $f=b_0X_l + b_1X_r + b_2$: Since $X_l\bot X_r$ by Lemma~\ref{lemma:ind-bin}, $f\sim \mathcal{N}(b_0\mu_l+b_1\mu_r+b_2,b_0^2\sigma_l^2+b_1^2\sigma_r^2)$.
\item $f=b_0X_lX_r + b_1X_l + b_2X_r + b_3$: Again, $\Pr(f)$ is not normal. Since $X_l\bot X_r$, $X_lX_r$ follows the so called \emph{normal product distribution}~\cite{Aroian47}, whose exact form is again complicated. We thus use the same strategy as in (C4') (see Appendix~\ref{sec:proofs:theory:c6}).
\end{enumerate}

\begin{lemma}\label{lemma:c4var}
If $X_l\sim \mathcal{N}(\mu_l, \sigma_l^2)$ and $f=b_0X_l^2 + b_1X_l + b_2$, then
$$\Var[f]=\sigma_l^2 [(b_1 + 2 b_0 \mu_l)^2 + 2 b_0^2 \sigma_l^2].$$
\end{lemma}

\begin{theorem}\label{theorem:c4approx}
Suppose that $X_l\sim \mathcal{N}(\mu_l, \sigma_l^2)$ and $f=b_0X_l^2 + b_1X_l + b_2$. Let $f^{\mathcal{N}}\sim \mathcal{N}(\E[f], \Var[f])$, where $\Var[f]$ is shown in Lemma~\ref{lemma:c4var}. Then $f^\mathcal{N}\inprob f$.\footnote{$f^\mathcal{N}\inprob f$ means $f^\mathcal{N}$ converges in probability to $f$.}
\end{theorem}

\subsubsection{Asymptotic Distribution of $t_k$}\label{sec:uncertainty:distr:tk}

Based on the previous analysis, the cost functions (or equivalently, the $n$'s in Equation~(\ref{eq:cost-model})) are asymptotically normal. Since the $c$'s are normal and independent of the $X$'s (and hence the $n$'s as well), by Equation~(\ref{eq:cost-model}) again $t_k$ is asymptotically the sum of products of two independent normal random variables. Specifically, let $\mathcal{C}=\{c_s, c_r, c_t, c_i, c_o\}$, and for $c\in\mathcal{C}$, let $f_{kc}$ be the cost function indexed by $c$. Defining $t_{kc}=f_{kc}^{\mathcal{N}}c$, we have
$$t_k\approx\sum_{c\in\mathcal{C}}t_{kc}=\sum_{c\in\mathcal{C}} f_{kc}^{\mathcal{N}}c,$$

Again, each $t_{kc}$ is not normal. But we can apply techniques similar to that in Theorem~\ref{theorem:c4approx} here by using the normal random variable
$$t_{kc}^{\mathcal{N}}\sim \mathcal{N}(\E[f_{kc}^\mathcal{N}c], \Var[f_{kc}^\mathcal{N}c])=\mathcal{N}(\E[f_{kc}c], \Var[f_{kc}c])$$
as an approximation of $t_{kc}$. Defining $Z=\E[f_{kc}]c$, we have

\begin{theorem}\label{theorem:tc}
$t_{kc}\indistr Z$, and $t_{kc}^{\mathcal{N}}\indistr Z$.
\end{theorem}

\noindent Theorem~\ref{theorem:tc} (proof in Appendix~\ref{sec:proofs:theorem:tc}) implies that $t_{kc}$ and $t_{kc}^\mathcal{N}$ tend to follow the same distribution as the sample size grows. 
Since $c$ is normal, $Z$ is normal as well. Furthermore, the independence of the $c$'s also implies the independence of the $Z$'s. So $t_k$ is approximately the sum of the independent normal random variables $t_{kc}^{\mathcal{N}}$. Hence $t_k$ is itself approximately normal with large sample size.

\subsubsection{Asymptotic Distribution of $t_q$}\label{sec:uncertainty:distr:tq}

Finally, let us consider the distribution of $t_q$. Since $t_q$ is merely the sum of the $t_k$'s, we have exactly the same situation as when we analyze each $t_k$. Specifically, we can express $t_q$ as
$$t_q=\sum_{k=1}^m t_k \approx\sum_{c\in\mathcal{C}}g_c c,$$
where $g_c=\sum_{k=1}^m f_{kc}^\mathcal{N}$ is the sum of the cost functions of the operators with respect to the particular $c$. However, since the $f_{kc}^\mathcal{N}$'s are not independent, $g_c$ is not normal. We can again use the normal random variable
$$g_c^\mathcal{N}\sim\mathcal{N}(\E[g_c], \Var[g_c])$$
as an approximation of $g_c$. We show $g_c^\mathcal{N}\inprob g_c$ in Appendix~\ref{sec:proofs:theorem:gc}. With exactly the same argument used in Section~\ref{sec:uncertainty:distr:tk} we can then see that $t_q$ is approximately normal when the sample size is large.

\subsubsection{Discussion}

The analysis that $t_q$ is asymptotically normal relies on three facts: (1) the selectivity estimates are unbiased and strongly consistent; (2) the cost model is additive; and (3) the cost units are independently normally distributed. While the first fact is a property of the sampling-based selectivity estimator and thus always holds, the latter two are specific merits of the cost model of PostgreSQL, though we believe that cost models of other database systems share more or less similar features. (As far as we know, MySQL~\cite{mysql-qo}, IBM DB2~\cite{db2-qo}, Oracle~\cite{oracle-qo}, and Microsoft SQL Server~\cite{sqlserver-qo} use similar cost models.) Therefore, we need new techniques when either (2) or (3) does not hold. For instance, if the cost model is still additive and the $c$'s are independent but cannot be modeled as normal variables, then by the analysis in Section~\ref{sec:uncertainty:distr:tq} we can still see that $t_q$ is asymptotically a linear combination of the $c$'s and thus the distribution of $t_q$ can be expressed in terms of the \emph{convolution} of the distributions of the $c$'s. We may then find this distribution by using generating functions or characteristic functions~\cite{ross}. We leave the investigation of other types of cost models as future work.

\subsection{Computing Distribution Parameters}

As discussed, we can approximate the distribution of $t_q$ with a normal distribution $\mathcal{N}(\E[t_q], \Var[t_q])$. We are then left with the problem of estimating the two parameters $\E[t_q]$ and $\Var[t_q]$. While $\E[t_q]$ is trivial to compute --- it is merely the original prediction from our predictor, estimating $\Var[t_q]$ is a challenging problem due to the correlations presented in selectivity estimates.

In more detail, so far we have observed the additive nature of $t_q$, that is, $t_q=\sum_{k=1}^m t_k$ and $t_k=\sum_{c\in\mathcal{C}}t_{kc}$ (Section~\ref{sec:uncertainty:distr:tk}). Recall the fact that for sum of random variables $Y = \sum_{1\leq i\leq m}Y_i$,
$$\Var[Y] = \sum\nolimits_{1\leq i,j\leq m}\Cov(Y_i,Y_j).$$
Applying this to $t_q$, our task is then to compute each $\Cov(t_i,t_j)$. Note that $\Cov(t_i,t_i) = \Var[t_i]$ which is easy to compute, so it is left to compute $\Cov(t_i,t_j)$ for $i \neq j$. By linearity of covariance,
\begin{align*}
  \label{eq:var-tq}
  \Cov(t_i,t_j) = \Cov\Big(\sum_{c \in {\cal C}}t_{ic}, \sum_{c \in {\cal C}}t_{jc}\Big)
  = \sum_{c,c' \in {\cal C}}\Cov(t_{ic}, t_{jc'}).
\end{align*}

In the following, we first specify the cases where direct computation of $\Cov(t_{ic}, t_{jc'})$ can be done. We then develop upper bounds for those covariances that cannot be directly computed.

\subsubsection{Direct Computation of Covariances}

Any $\Cov(t_{ic}, t_{jc'})$ can fall into the following two cases:
\begin{itemize}
\item $i=j$, then it is the covariance between different cost functions from the same operator.
\item $i\neq j$, then it is the covariance between cost functions from different operators.
\end{itemize}

Consider the case $i=j$ first. If the operator is unary, regarding the cost functions we are concerned with, we only need to consider $\Cov(X, X)$, $\Cov(X, X^2)$, and $\Cov(X^2, X^2)$, where $X\sim \mathcal{N}(\mu,\sigma^2)$. Since $X$ is normal, the non-central moments of $X$ can be expressed in terms of $\mu$ and $\sigma^2$. Hence it is straightforward to compute these covariances~\cite{winkelbauer2012moments}. If the operator is binary, then we need to consider $\Cov(X_l, X_l)$, $\Cov(X_r, X_r)$, $\Cov(X_l, X_r)$, $\Cov(X_lX_r, X_l)$, $\Cov(X_lX_r, X_r)$, and $\Cov(X_lX_r, X_lX_r)$. By Lemma~\ref{lemma:ind-bin}, $X_l\bot X_r$. So we are able to directly compute these covariances as well.

When $i\neq j$, while the types of covariances that we need to consider are similar as before, it is more complicated since the selectivities are no longer independent. Without loss of generality, we consider two operators $O$ and $O'$ such that $O\in Desc(O')$. By Lemma~\ref{lemma:correlated}, this is the only case where the covariances might not be zero. Based on the cost functions considered in this paper, we need to consider the covariances $\Cov(Z,Z')$, where $Z\in\{X_l, X_l^2, X_r, X_lX_r\}$ and $Z'\in\{X'_l, (X'_l)^2, X'_r, X'_lX'_r\}$. Some of them can be directly computed by applying Lemma~\ref{lemma:correlated}, while the others can only be bounded as discussed in the next section.

\begin{example}[Covariances between selectivities]
To illustrate, consider the two join operators $O_4$ and $O_5$ in Figure~\ref{fig:query-plan}. Assume that the cost functions of $O_4$ and $O_5$ are all linear, i.e., they are of type (C5'). Based on Lemma~\ref{lemma:ind-bin}, $\Cov(X_1, X_2)=0$ and $\Cov(X_4, X_3)=0$. Also, based on Lemma~\ref{lemma:correlated}, $\Cov(X_1, X_3)=0$ and $\Cov(X_2, X_3)=0$. However, we are not able to compute $\Cov(X_1, X_4)$ and $\Cov(X_2, X_4)$. Instead, we provide upper bounds for them.
\end{example}

\subsubsection{Upper Bounds of Covariances}\label{sec:uncertainty:covar:bounds}

Based on the fact that the covariance between two random variables is bounded by the geometric mean of their variances~\cite{ross}, we can establish an upper bound for $Z$ and $Z'$ in the previous section:
$$|\Cov(Z,Z')|\leq\sqrt{\Var[Z]\Var[Z']}.$$
Note that the variances are directly computable based on the independence assumptions (Lemma~\ref{lemma:ind-bin} and~\ref{lemma:correlated}).

By analyzing the correlation of the samples used in selectivity estimation, we can develop tighter bounds (details in Appendix~\ref{sec:proofs:theorem:tighter-bound}). The key observation here is that the correlations are caused by the samples from the shared relations. Consider two operators $O$ and $O'$ such that $O\in Desc(O')$. Suppose that $|\mathcal{R}\cap\mathcal{R}'|=m$ ($m\geq 1$), namely, $O$ and $O'$ share $m$ common leaf tables. Let the estimators for $O$ and $O'$ be $\rho_n$ and $\rho'_n$, where $n$ is the number of sample steps. We define $S_{\rho}^2(m,n)$ to be the variance of samples restricted to the $m$ common relations. This is actually a generalization of $\Var[\rho_n]$. To see this, let $\mathcal{R}'=\mathcal{R}$. Then $\rho_n=\rho'_n$ and hence
$$\Var[\rho_n]=\Cov(\rho_n,\rho_n)=\Cov(\rho_n,\rho'_n)=S_{\rho}^2(K,n),$$
where $K=|\mathcal{R}|$. We can show that $S_{\rho}^2(m,n)$ is a monotonically increasing function of $m$ (see Appendix~\ref{sec:proofs:theorem:tighter-bound}). As a result, $S_{\rho}^2(m,n)\leq \Var[\rho_n]$ given that $m\leq K$. Hence, we have the following refined upper bound for $\Cov(\rho_n,\rho'_n)$:
$$|\Cov(\rho_n,\rho'_n)|\leq\sqrt{S_{\rho}^2(m,n)S_{\rho'}^2(m,n)}\leq \sqrt{\Var[\rho_n]\Var[\rho'_n]}.$$
To compute $S_{\rho}^2(m,n)$, we use an estimator akin to the estimator $\sigma_n^2=S_n^2 / n$ that we used to estimate $\Var[\rho_n]$. Specifically, define
$$S_{n,m}^2=\sum_{r=1}^{m}\bigg(\frac{1}{n-1}\sum_{j=1}^n(Q_{r,j,n}/n^{m-1}-\rho_n)^2\bigg),$$
for $n\geq 2$ (we set $S_{1,m}^2=0$). 
Very similarly, we can show that $\lim_{n\to\infty}S_{n,m}^2=n S_{\rho}^2(m,n)$. As a result, it is reasonable to approximate $S_{\rho}^2(m,n)$ with $S_{\rho}^2(m,n)\approx S_{n,m}^2 / n$. Moreover, by comparing the expressions of $S_{n,m}^2$ and $S_n^2$ (ref. Equation~(\ref{eq:sn-square})), we can see that $S_n^2=S_{n,K}^2$. Therefore it is straightforward to adapt the implementation framework in Section~\ref{sec:distribution:Sn2} to compute $S_{n,m}^2$. More discussions on bounding covariances are in Appendix~\ref{sec:more-cov-bounds}.

\section{Experimental Evaluation} \label{sec:experiment}

We present experimental evaluation results in this section. There are two key respects that could impact the utility of a predictor: its prediction accuracy and runtime overhead. However, for the particular purpose of this paper, we do not care much about the \emph{absolute} accuracy of the prediction. Rather, we care if the distribution of likely running times reflects the uncertainty in the prediction. Specifically, we measure if the estimated prediction errors are correlated with the actual errors. To measure the accuracy of the predicted distribution, we also compare the estimated likelihoods that the actual running times will fall into certain confidence intervals with the actual likelihoods. On the other hand, we measure the runtime overhead of the sampling-based approach in terms of its relative overhead with respect to the original query running time without sampling. We start by presenting the experimental settings and the benchmark queries we used.

\subsection{Experimental Settings}

We implemented our proposed framework in PostgreSQL 9.0.4. We ran PostgreSQL under Linux 3.2.0-26, and we evaluated our approaches with both the TPC-H 1GB and 10 GB databases. Since the original TPC-H database generator uses uniform distributions, to test the effectiveness of the approach under different data distributions, we used a skewed TPC-H database generator~\cite{skewed-gen}. It produces TPC-H databases with a Zipf distribution and uses a parameter $z$ to control the degree of skewness. $z$ = 0 generates a uniform distribution, and the data becomes more skewed as $z$ increases. We created skewed databases using $z$ = 1. All experiments were conducted on two machines with the following configurations:
\begin{itemize}
\item PC1: Dual Intel 1.86 GHz CPU and 4GB of memory;
\item PC2: 8-core 2.40GHz Intel CPU and 16GB of memory.
\end{itemize}

\subsection{Benchmark Queries}

We created three benchmarks \textbf{MICRO}, \textbf{SELJOIN}, and \textbf{TPCH}:
\begin{itemize}
\item \textbf{MICRO} consists of pure selection queries (i.e., scans) and two-way join queries. It is a micro-benchmark with the purpose of exploring the strength and weakness of our proposed approach at different points in the selectivity space. We generated the queries with the similar ideas used in the Picasso database query optimizer visualizer~\cite{ReddyH05}. Since the queries have either one (for scans) or two predicates (for joins), the selectivity space is either one or two dimensional. We generated SQL queries that were evenly across the selectivity space, by using the statistics information (e.g., histograms) stored in the database catalogs to compute the selectivities.

\item \textbf{SELJOIN} consists of selection-join queries with multi-way joins. We generated the queries in the following way. We analyzed each TPC-H query template, and identified the ``maximal'' sub-query without aggregates. We then randomly generated instance queries from these \emph{reduced} templates. The purpose is to test the particular type of queries to which our proposed approach is tailored --- the selection-join queries.

\item \textbf{TPCH} consists of instance queries from the TPC-H templates. These queries also contain aggregates, and our current strategy is simply ignoring the uncertainty there (recall Section~\ref{sec:distribution:Sn2}). The purpose of this benchmark is to see how this simple work-around works in practice. We used 14 TPC-H templates: 1, 3, 4, 5, 6, 7, 8, 9, 10, 12, 13, 14, 18, and 19. We did not use the other templates since their query plans contain structures that cannot be handled by our current framework (e.g., sub-query plans or views).
\end{itemize}

We ran each query 5 times and took the average as the actual running time of a query. We cleared both the filesystem cache and the database buffer pool between each run of each query.

\subsection{Usefulness of Predicted Distributions}\label{sec:experiment:accuracy}

Since our goal is to quantify the uncertainty in the prediction and our output is a distribution of likely running times, the question is then how we can know that we have something useful. A reasonable metric here could be the correlation between the standard deviation of the predicted (normal) distribution and the actual prediction error. Intuitively, the standard deviation indicates the confidence of the prediction. A larger standard deviation indicates lower confidence and hence larger potential prediction error. With this in mind, if our approach is effective, we would expect to see positive correlations between the standard deviations and the real prediction errors when a large number of queries are tested.

A common metric used to measure the correlation between two random variables is the Pearson correlation coefficient $r_p$. Suppose that we have $n$ queries $q_1$, ..., $q_n$. Let $\sigma_i$ be the standard deviation of the distribution predicted for $q_i$, $\mu_i$ and $t_i$ be the predicted (mean) and actual running time of $q_i$, and $e_i=|\mu_i-t_i|$ be the prediction error. $r_p$ is then defined as
\begin{equation}\label{eq:rp}
r_p=\frac{\sum_{i=1}^n(\sigma_i-\bar{\sigma})(e_i-\bar{e})}{\sqrt{\sum_{i=1}^n(\sigma_i-\bar{\sigma})^2}{\sqrt{\sum_{i=1}^n(e_i-\bar{e})^2}}},
\end{equation}
where $\bar{\sigma}=\frac{1}{n}\sum_{i=1}^n \sigma_i$ and $\bar{e}=\frac{1}{n}\sum_{i=1}^n e_i$.

\begin{figure*}
\centering
\subfigure[\textbf{MICRO}, Uniform 1GB, PC2]{ \label{fig:cc:micro}
\includegraphics[width=0.68\columnwidth]{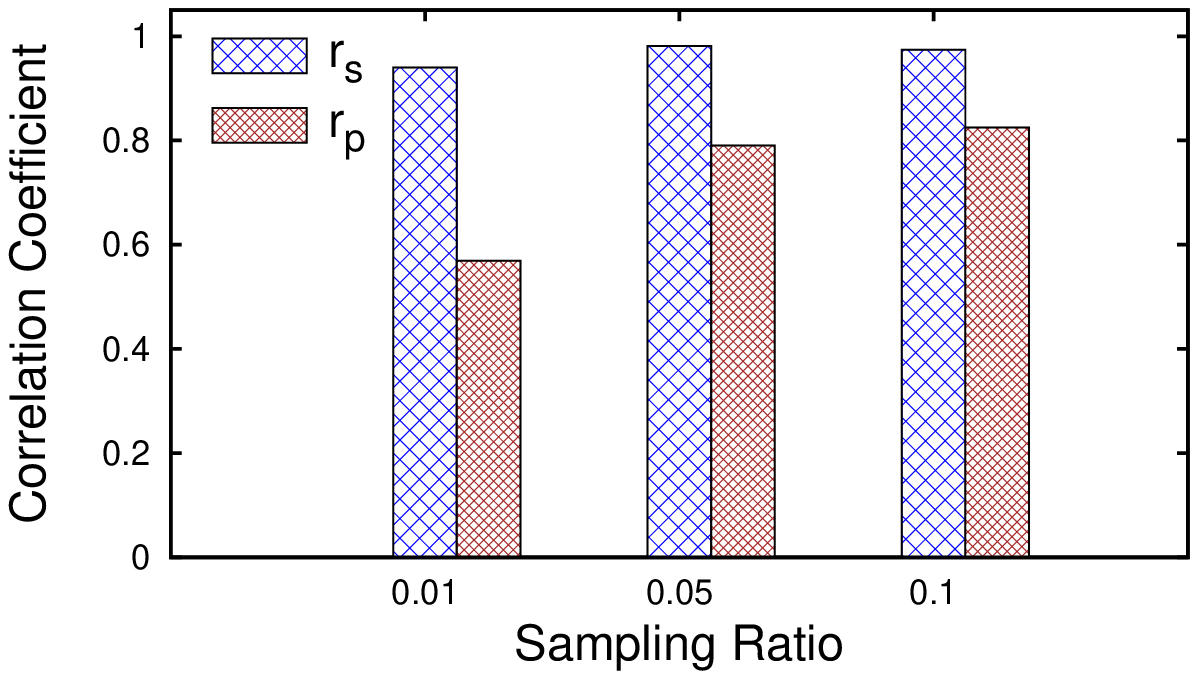}}
\subfigure[\textbf{SELJOIN}, Uniform 1GB, PC1]{ \label{fig:cc:sj}
\includegraphics[width=0.68\columnwidth]{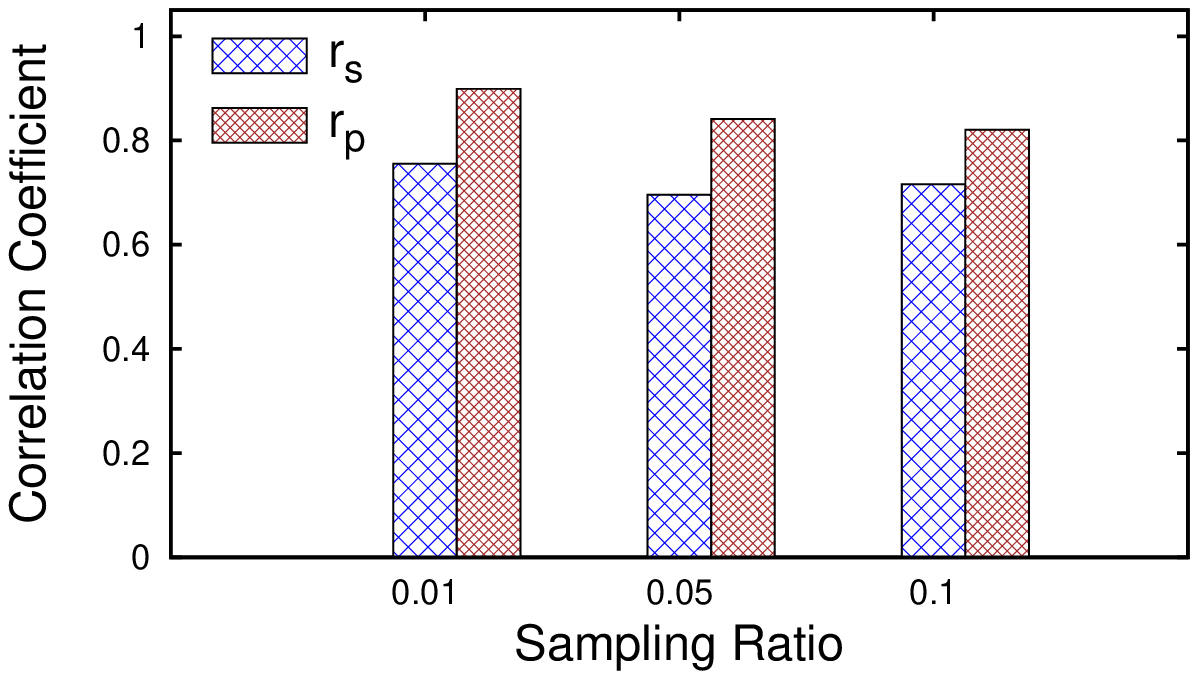}}
\subfigure[\textbf{TPCH}, Skewed 10GB, PC1]{ \label{fig:cc:tpch}
\includegraphics[width=0.68\columnwidth]{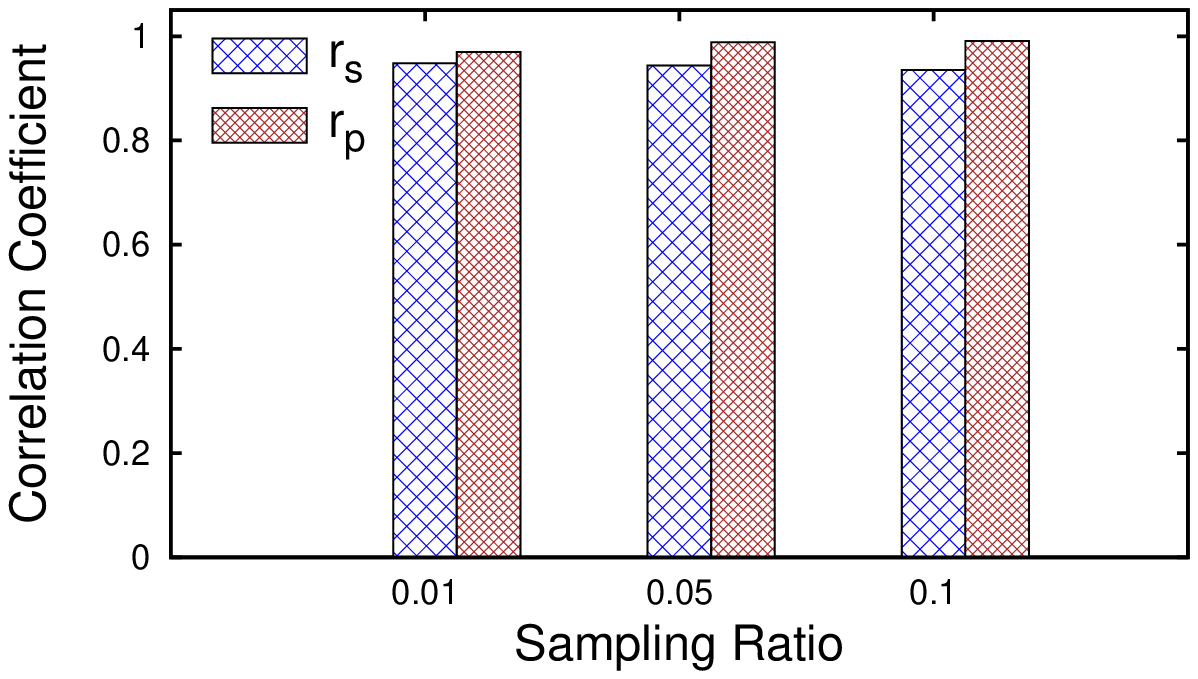}}
\caption{$r_s$ and $r_p$ of the benchmark queries over different hardware and database settings}
\label{fig:cc}
\shrink
\end{figure*}

\begin{figure*}
\centering
\subfigure[Case (1)]{ \label{fig:robust:case1}
\includegraphics[width=0.68\columnwidth]{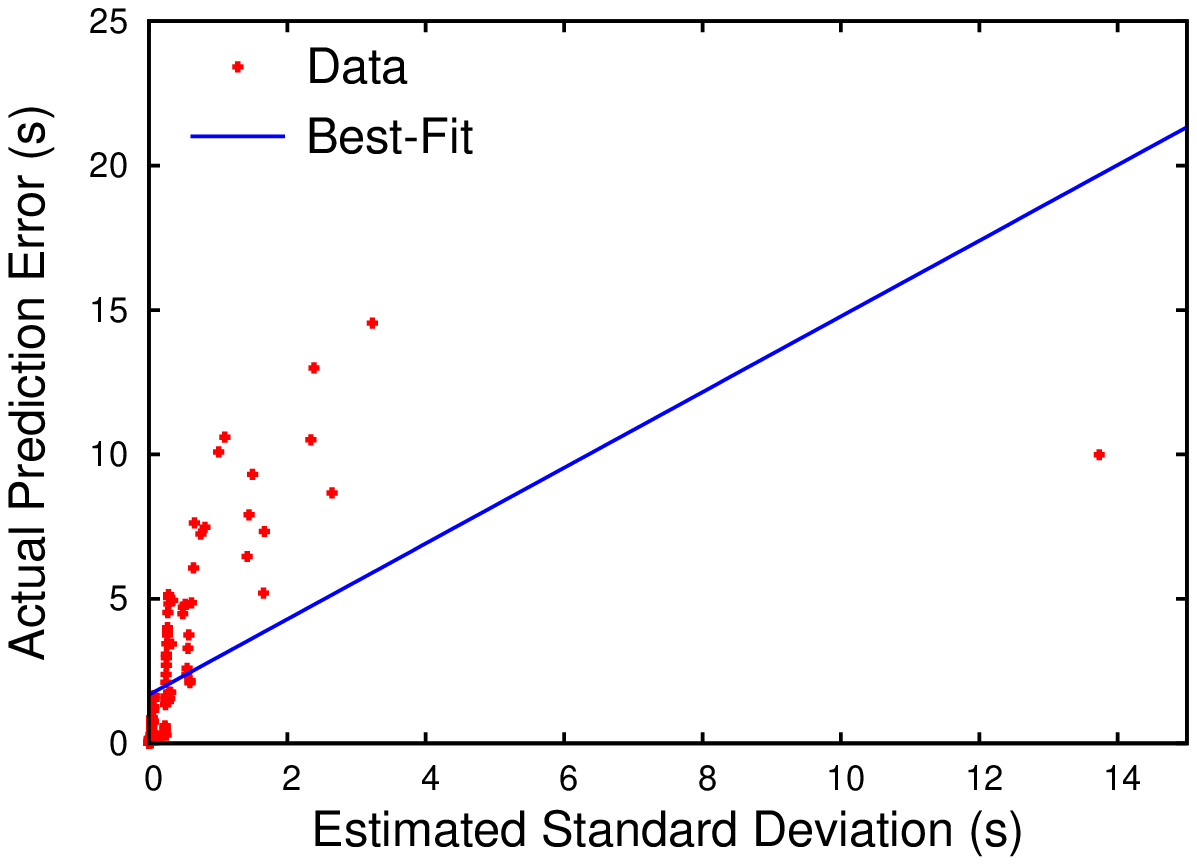}}
\subfigure[Case (1) after one outlier is removed]{ \label{fig:robust:case1-outlier-removed}
\includegraphics[width=0.68\columnwidth]{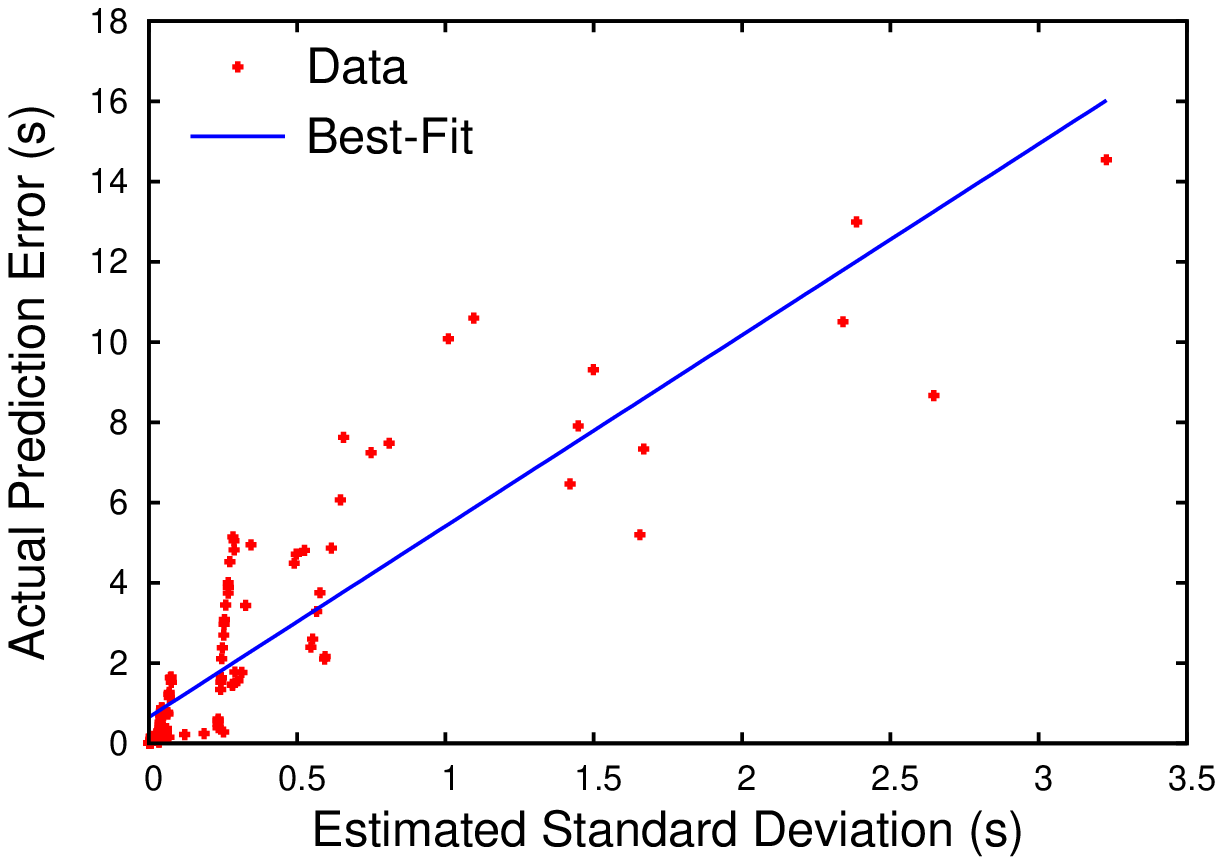}}
\subfigure[Case (2)]{ \label{fig:robust:case2}
\includegraphics[width=0.68\columnwidth]{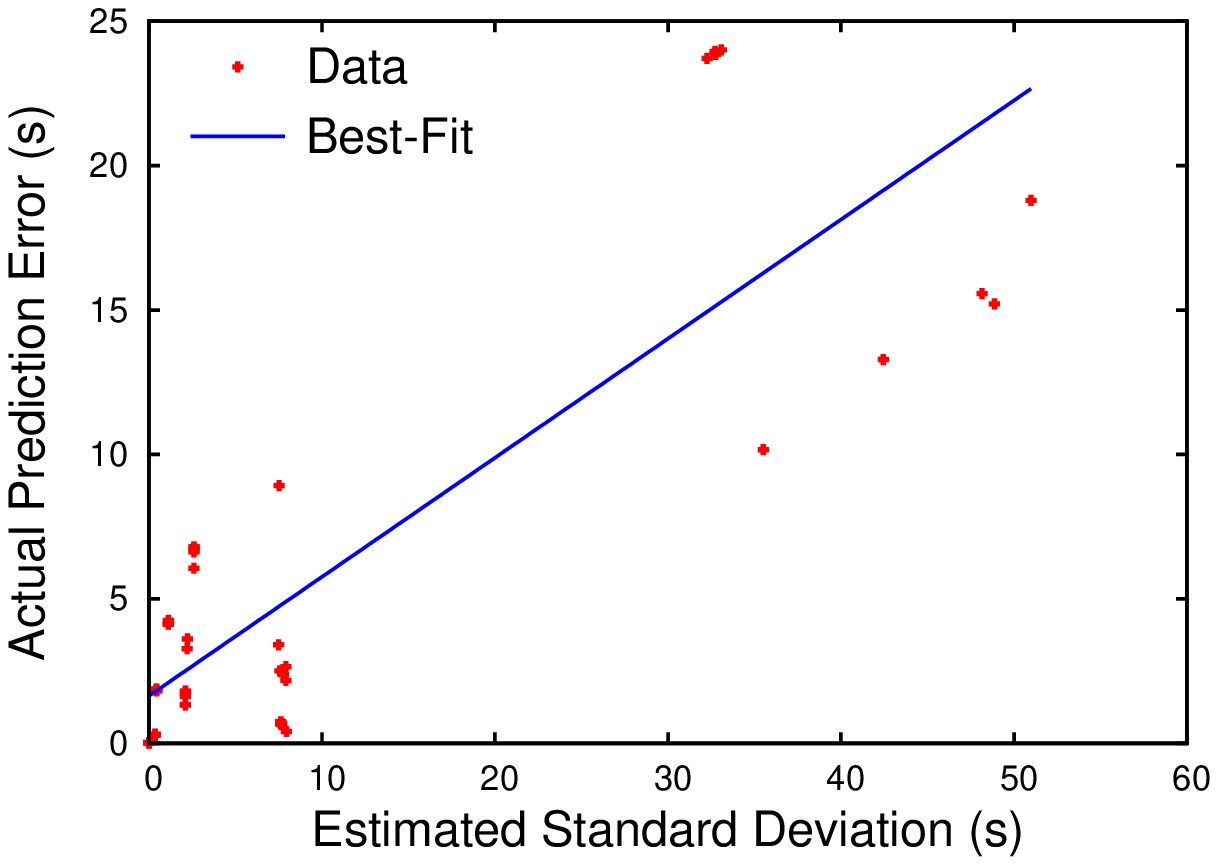}}
\caption{Robustness of $r_s$ and $r_p$ with respect to outliers}
\label{fig:robust}
\shrink
\end{figure*}

Basically, $r_p$ measures the \emph{linear} correlation between the $\sigma$'s and the $e$'s.
The closer $r_p$ is to $1$, the better the correlation is. However, there are two issues here. First, even if the $\sigma$'s and the $e$'s are positively correlated, the correlation may not be linear. Second, $r_p$ is not robust and its value can be misleading if outliers are present~\cite{Dev75-robustness}.
Therefore, we also measure the correlations by using another well known metric called the Spearman's rank correlation coefficient $r_s$~\cite{myers2003-spearman}. The formula of $r_s$ is the same as Equation~(\ref{eq:rp}) except for that the $\sigma$'s and $e$'s are replaced with their \emph{ranks} in the ascending order of the values. For instance, given three $\sigma$'s $\sigma_1=4$, $\sigma_2=7$, and $\sigma_3=5$, their ranks are $1$, $3$, and $2$ respectively. Intuitively, $r_s$ indicates the linear correlation between the ranks of the values, which is more robust than $r_p$ since the mapping from the values to their ranks can be thought of as some \emph{normalization} procedure that reduces the impact of outliers. In fact, $r_s$ assesses how well the correlation can be characterized by using a \emph{monotonic} function and $r_s=1$ means the correlation is perfect.

In Figure~\ref{fig:cc}, we report the $r_s$'s (and the corresponding $r_p$'s) for the benchmark queries over different hardware and database settings (see Table~\ref{tab:cc:pred-time} of Appendix~\ref{sec:more-exp-results:correlations} for the complete results). Here, sampling ratio (SR) stands for the fraction of the sample size with respect to the database size. For instance, SR = 0.01 means that 1\% of the data is taken as samples. We have several observations.

First, for most of the cases we tested, both $r_s$ and $r_p$ are above 0.7 (in fact above 0.9), which implies strong positive (linear) correlation between the standard deviations of the predicted distributions and the actual prediction errors.\footnote{It is generally believed that two variables are strongly correlated if their correlation coefficient is above 0.7.}
Second, in~\cite{WuCZTHN13} we showed that as expected, prediction errors can be reduced by using larger number of samples. Interestingly, it is not necessarily the case that more samples improves the correlation between the predicted and actual errors. This is because taking more samples
simultaneously reduces the errors in selectivity estimates and the uncertainty in the predicted running times. So it might improve
the estimate but not the correlation with the true errors. Third, reporting both $r_s$ and $r_p$ is necessary since they sometimes disagree with each other. For instance, consider the following two cases in Figure~\ref{fig:cc:micro} and~\ref{fig:cc:sj}:
\begin{enumerate}[(1)]
\item On PC2, the \textbf{MICRO} queries over the uniform TPC-H 1GB database give $r_s$ = 0.9400 but $r_p$ = 0.5691 when SR = 0.01;
\item On PC1, the \textbf{SELJOIN} queries over the uniform TPC-H 1GB database give $r_s$ = 0.6958 but $r_p$ = 0.8414 when SR = 0.05.
\end{enumerate}
In Figure~\ref{fig:robust:case1} and~\ref{fig:robust:case2}, we present the scatter plots of these two cases. Figure~\ref{fig:robust:case1-outlier-removed} further shows the scatter plot after the rightmost point is removed from Figure~\ref{fig:robust:case1}. We find that now $r_s=0.9386$ but $r_p=0.8868$. So $r_p$ is much more sensitive to outliers in the population. Since in our context there is no good criterion to remove outliers, $r_s$ is thus more trustworthy. On the other hand, although the $r_p$ of (2) is better than that of (1), by comparing Figure~\ref{fig:robust:case1-outlier-removed} with Figure~\ref{fig:robust:case2} we would instead conclude that the correlation of (2) seems to be worse. This is again implied by the worse $r_s$ of (2). More results and analysis can be found in Appendix~\ref{sec:more-exp-results:selectivity}.

\begin{figure*}[!htb]
\centering
\subfigure[\textbf{MICRO}]{ \label{fig:dist:micro}
\includegraphics[width=0.68\columnwidth]{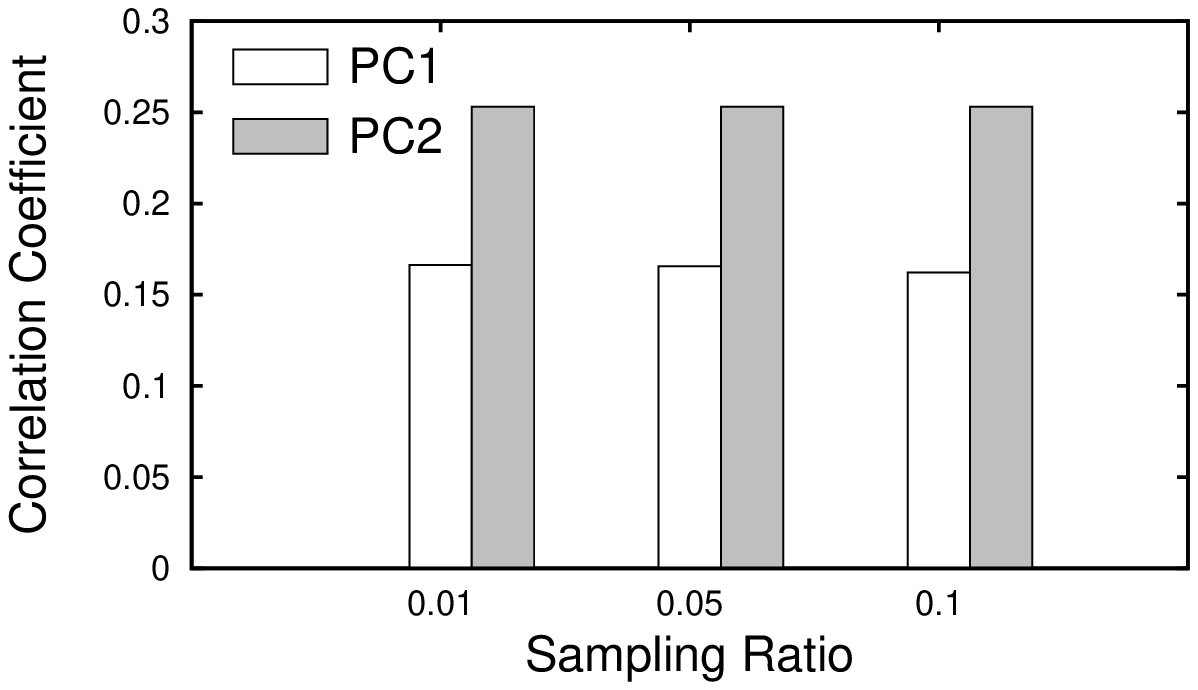}}
\subfigure[\textbf{SELJOIN}]{ \label{fig:dist:sj}
\includegraphics[width=0.68\columnwidth]{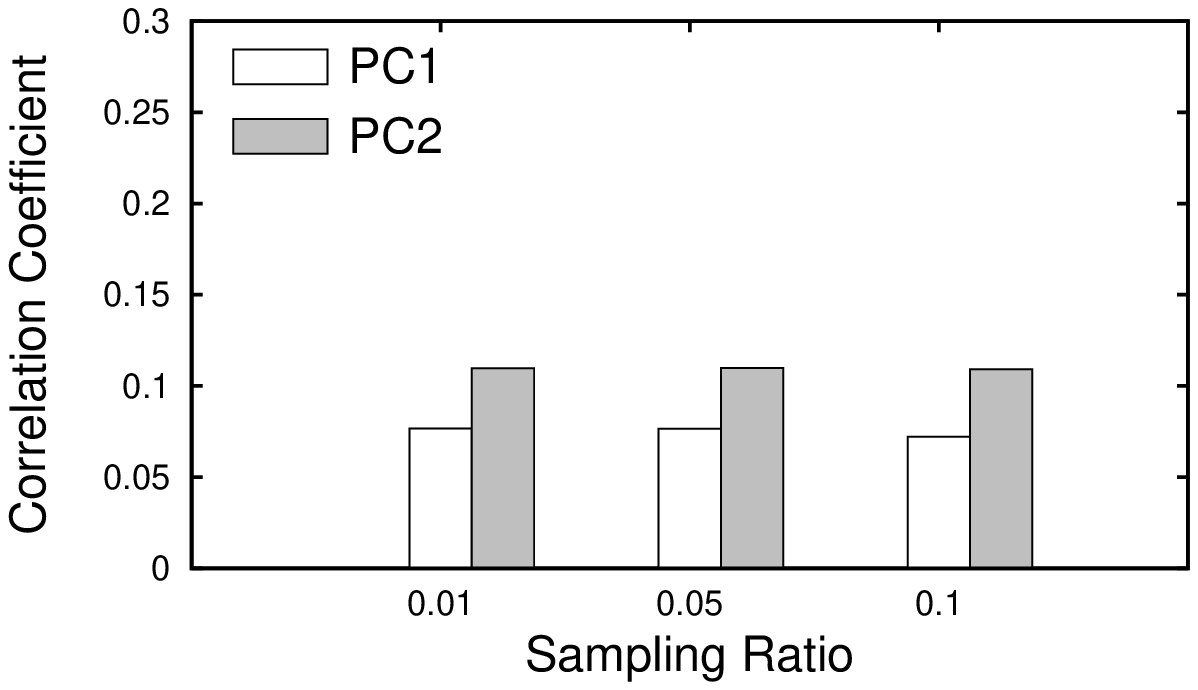}}
\subfigure[\textbf{TPCH}]{ \label{fig:dist:tpch}
\includegraphics[width=0.68\columnwidth]{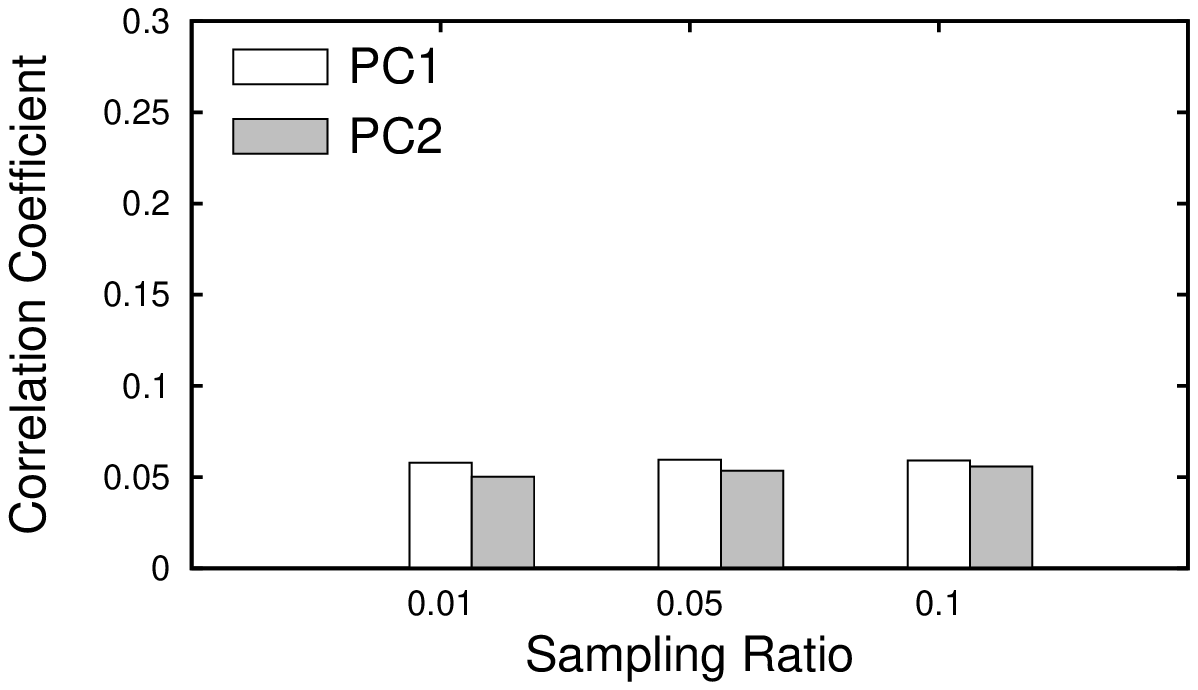}}
\caption{$\overline{D}_n$ of the benchmark queries over uniform TPC-H 10GB databases}
\label{fig:dist}
\shrink
\end{figure*}

\begin{figure*}
\centering
\subfigure[Case (1), $\overline{D}_n$ = 0.2532]{ \label{fig:ks:case1}
\includegraphics[width=0.68\columnwidth]{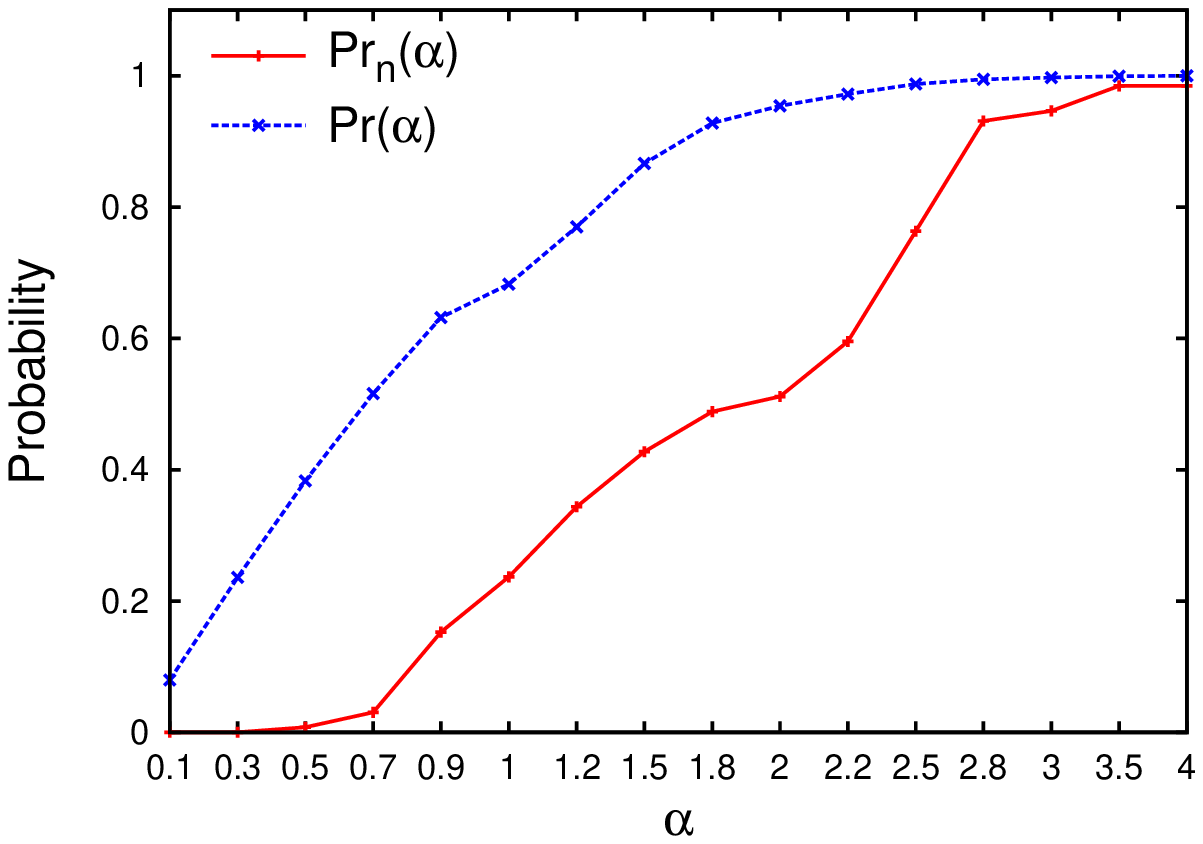}}
\subfigure[Case (2), $\overline{D}_n$ = 0.1098]{ \label{fig:ks:case2}
\includegraphics[width=0.68\columnwidth]{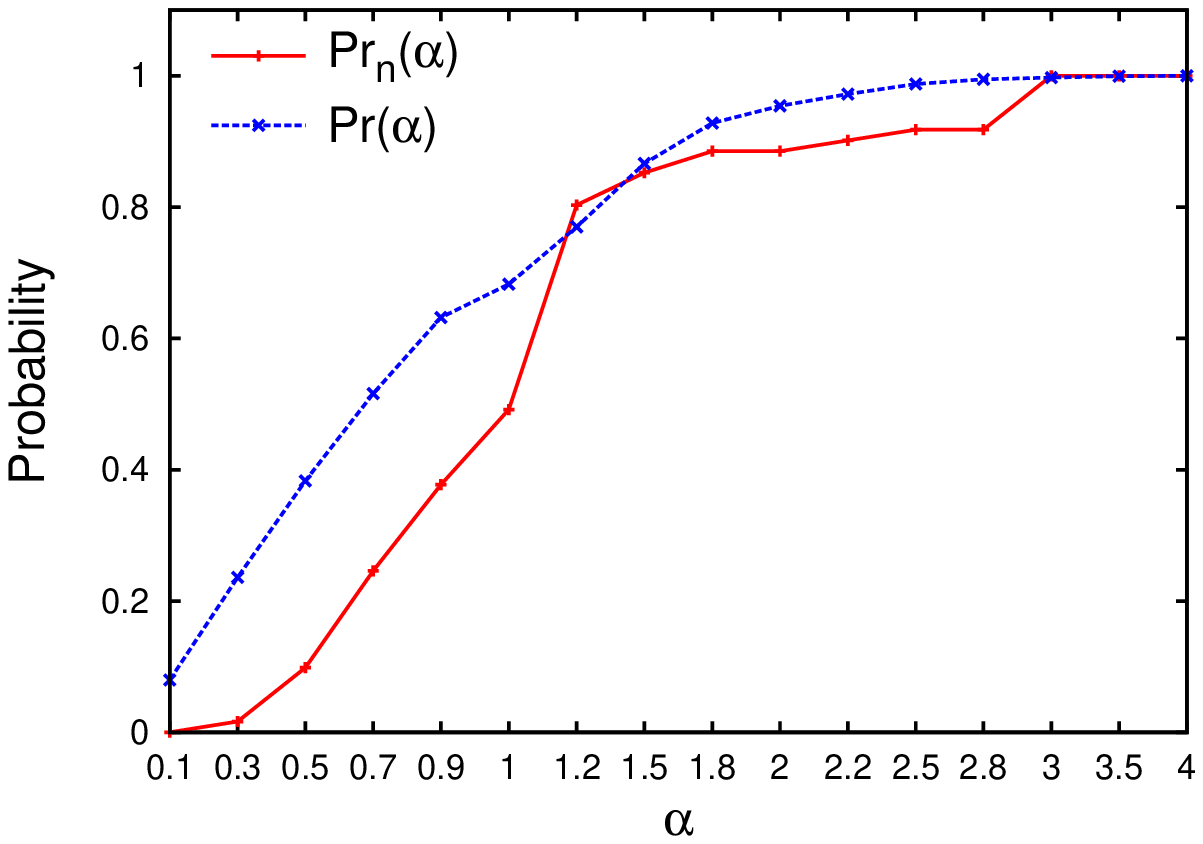}}
\subfigure[Case (3), $\overline{D}_n$ = 0.0535]{ \label{fig:ks:case3}
\includegraphics[width=0.68\columnwidth]{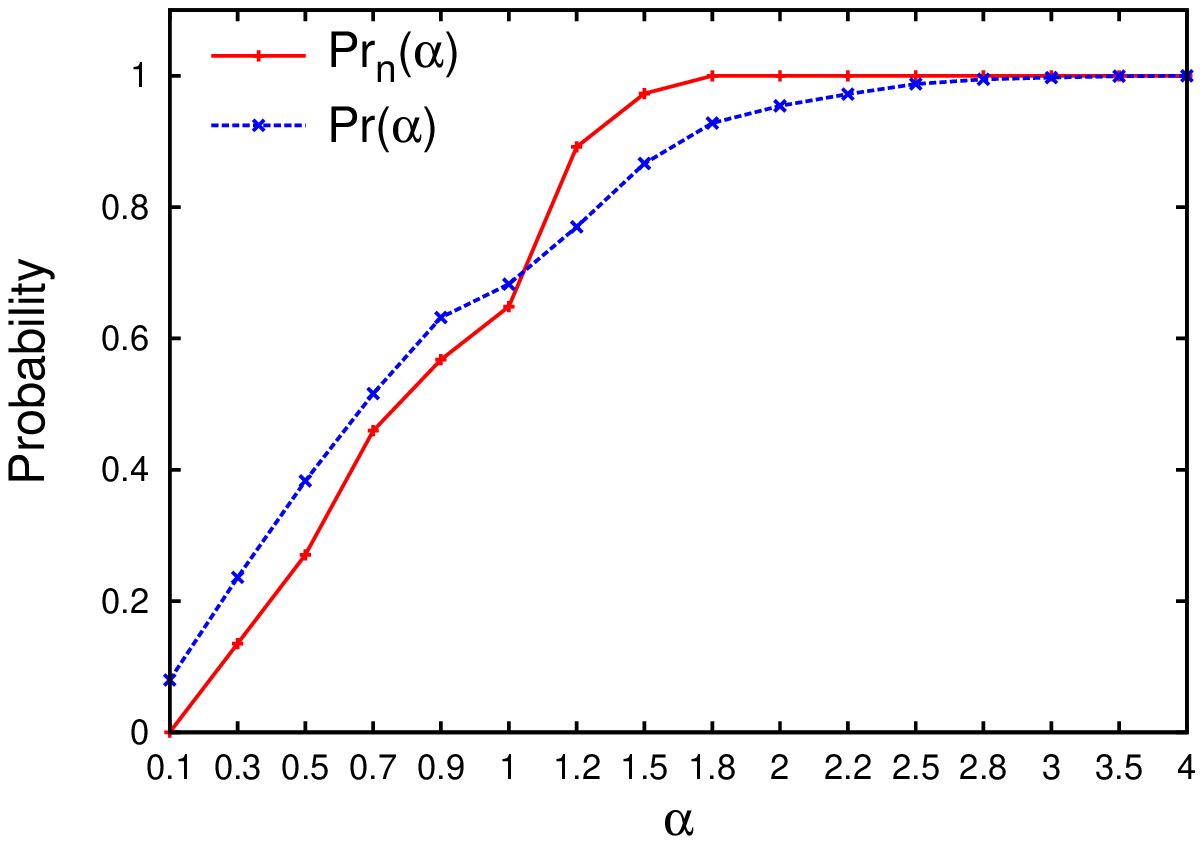}}
\caption{The proximity of $\Pr\nolimits_n(\alpha)$ and $\Pr(\alpha)$ with respect to different $\overline{D}_n$'s}
\label{fig:ks}
\shrink
\end{figure*}

Nonetheless, the strong positive correlations between the estimated standard deviations and the actual prediction errors may not be sufficient to conclude that the distributions of likely running times are useful. For our purpose of informing the consumer of the running time estimates of the potential prediction errors, it might be worth to further consider what information regarding the errors the predicted distributions really carry. Formally, consider the $n$ queries $q_1$, ..., $q_n$ as before. Since the estimated distributions are normal, with the previous notation the distribution for the likely running times $T_i$ of $q_i$ is $T_i\sim\mathcal{N}(\mu_i, \sigma_i^2)$. As a result, assuming $\alpha > 0$, without loss of generality the estimated prediction error $E_i=|T_i-\mu_i|$ follows the distribution
$$\Pr(E_i\leq \alpha\sigma_i)=\Pr(-\alpha\leq \frac{T_i-\mu_i}{\sigma_i}\leq\alpha)=2\Phi(\alpha)-1,$$
where $\Phi$ is the cumulative distribution function of the standard normal distribution $\mathcal{N}(0,1)$. Therefore, if we define the statistic $E'_i=\frac{E_i}{\sigma_i}=|\frac{T_i-\mu_i}{\sigma_i}|$, then $\Pr(E'_i\leq\alpha)=\Pr(E_i\leq \alpha\sigma_i)$. Note that $\Pr(E'_i\leq \alpha)$ is determined by $\alpha$ but not $i$. We thus simply use $\Pr(\alpha)$ to denote $\Pr(E'_i\leq\alpha)$. On the other hand, we can estimate the actual likelihood of $E'_i\leq\alpha$ by using
$$\Pr\nolimits_n(\alpha)=\frac{1}{n}\sum_{i=1}^n I(e'_i\leq \alpha),\textrm{ where }e'_i=\frac{e_i}{\sigma_i}=|\frac{t_i-\mu_i}{\sigma_i}|.$$
Here $I$ is the indicator function. To measure the proximity of $\Pr\nolimits_n(\alpha)$ and $\Pr(\alpha)$, we define
$$D_n(\alpha)=|\Pr\nolimits_n(\alpha)-\Pr(\alpha)|.$$
Clearly, a smaller $D_n(\alpha)$ means $\Pr(\alpha)$ is closer to $\Pr\nolimits_n(\alpha)$, which implies better quality of the distributions. We further generated $\alpha$'s from the interval $(0,6)$ which is sufficiently wide for normal distributions and computed the average of the $D_n(\alpha)$'s (denoted as $\overline{D}_n$). Figure~\ref{fig:dist} reports the results for the benchmark queries over uniform TPC-H 10GB databases (see Table~\ref{tab:ks} of Appendix~\ref{sec:more-exp-results:distances} for the complete results).

We observe that in most cases the $\overline{D}_n$'s are below 0.3 with the majority below 0.2, which suggests that the estimated $\Pr(\alpha)$'s are reasonably close to the observed $\Pr\nolimits_n(\alpha)$'s. To shed some light on what is going on here, in Figure~\ref{fig:ks} we further plot the $\Pr(\alpha)$ and $\Pr\nolimits_n(\alpha)$ for
the (1) \textbf{MICRO}, (2) \textbf{SELJOIN}, and (3) \textbf{TPCH} queries over the uniform TPC-H 10GB database on PC2 when SR = 0.05,
which give $\overline{D}_n$ = 0.2532, 0.1098, and 0.0535 respectively.
We can see that we overestimated the $\Pr(\alpha)$'s for small $\alpha$'s. In other words, we underestimated the prediction errors by presenting smaller than actual variances in the distributions. Moreover, we find that overestimate is more significant for the \textbf{MICRO} queries (Figure~\ref{fig:ks:case1}). One possible reason is that since these queries are really simple the predictor tends to be over-confident by underestimating the variances even more. When handling \textbf{SELJOIN} and \textbf{TPCH} queries, the confidence of the predictor drops and underestimate tends to be alleviated (Figure~\ref{fig:ks:case2} and~\ref{fig:ks:case3}).

\subsubsection{More Discussion on Correlation}

While using ordinal ranks instead of values can help in smoothing the data reducing the impact of outliers, it is still imperfect. The best way of presenting correlations between two quantities might be a scatter plot (as shown in Figure~\ref{fig:robust}).
There are four possible cases: (1) $r_s$ is better than $r_p$; (2) $r_p$ is better than $r_s$; (3) $r_p$ and $r_s$ are both good; and (4) $r_p$ and $r_s$ are both not so good.
We have presented scatter plots for (1) and (2) in Figure~\ref{fig:robust:case1} and~\ref{fig:robust:case2}, respectively.
To gain more insight, in Figure~\ref{fig:cc-casestudy:3} and~\ref{fig:cc-casestudy:4} we present two typical scatter plots for (3) and (4):

\begin{enumerate}[(1)]
\setcounter{enumi}{2}
\item On PC1, the \textbf{TPCH} queries over the skewed TPC-H 10GB database give $r_s$ = 0.9439 and $r_p$ = 0.9887 when SR = 0.05;
\item On PC1, the \textbf{TPCH} queries over the uniform TPC-H 1GB database give $r_s$ = 0.7209 and $r_p$ = 0.7571 when SR = 0.01.
\end{enumerate}

\begin{figure}[!htb]
\centering
\subfigure[Case (3)]{ \label{fig:cc-casestudy:3}
\includegraphics[width=0.68\columnwidth]{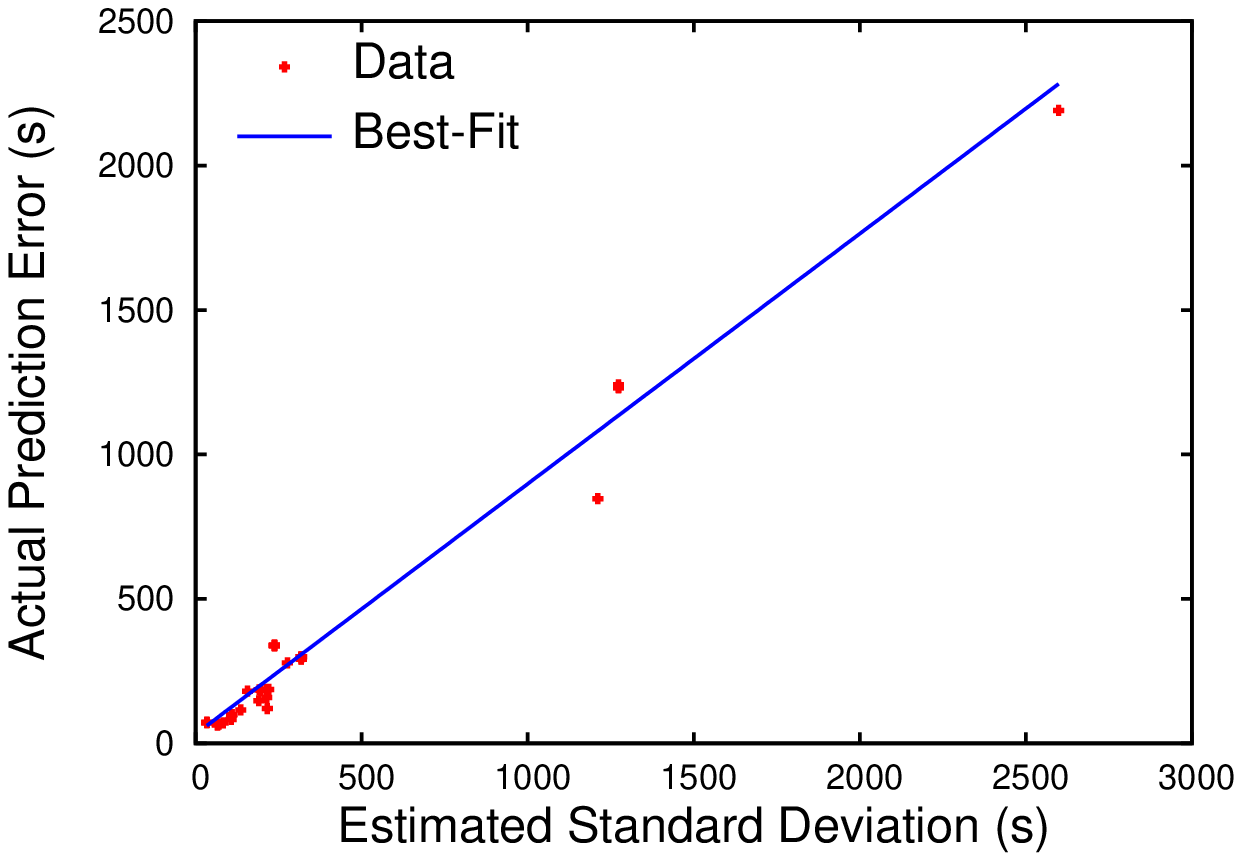}}
\subfigure[Case (4)]{ \label{fig:cc-casestudy:4}
\includegraphics[width=0.68\columnwidth]{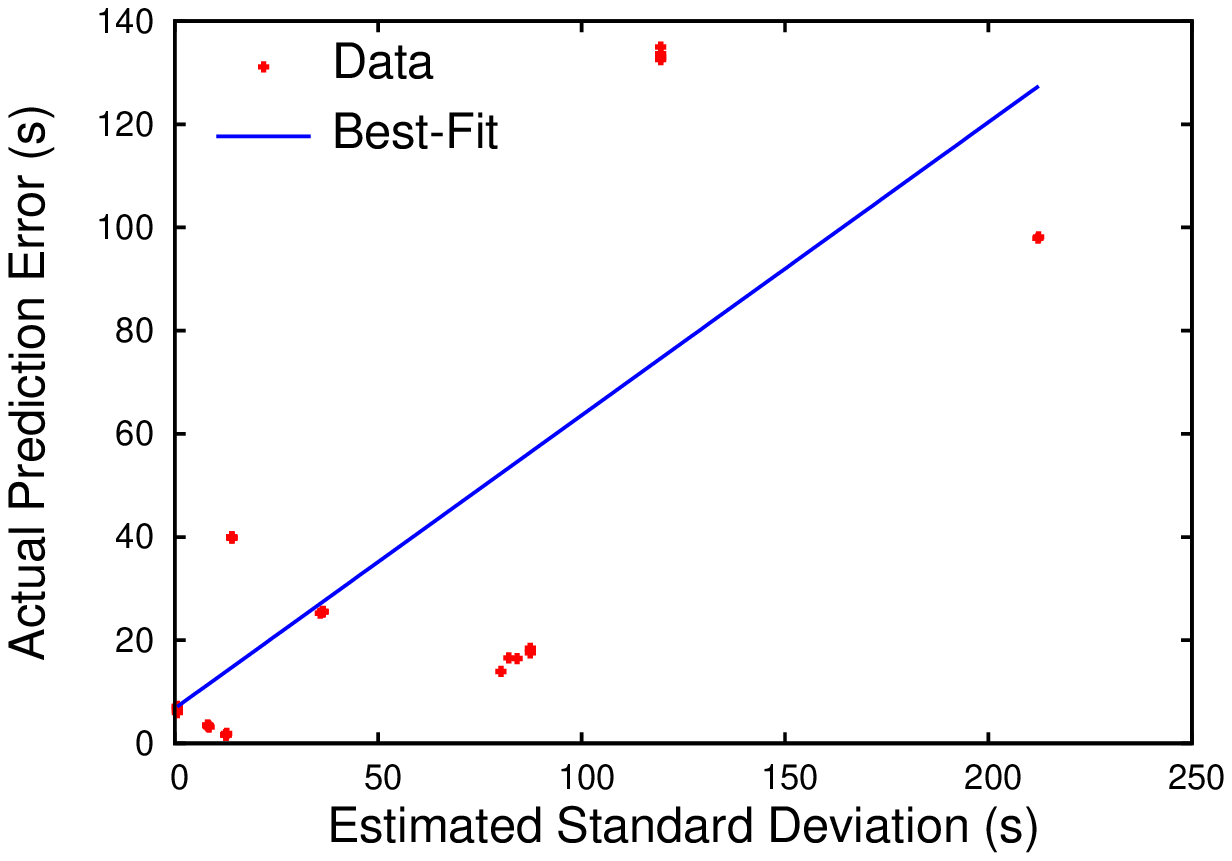}}
\caption{More case studies on correlations}
\label{fig:cc-casestudy}
\shrink
\end{figure}

As we can see, when both $r_s$ and $r_p$ are good, the correlation is close to positive linear. On the other hand, the correlation is not so good when both $r_s$ and $r_p$ are not so good.

\subsubsection{A Note on a Baseline Experiment}

Shrewd readers might have wondered a different, intuitively simpler experiment: fix one query, generate many different samples, and make a prediction based on each sample; then test if and how well the distribution of the predicted estimates matches the distribution computed by using our proposed framework.

The question raised here is if the distribution of running times predicted by using different samples would match the one computed by our model.
But note that ``the distribution in the model'' actually depends on samples, that is, the model will output a different distribution if it uses a different sample.

To put things in context, let us consider a query $q$ where we used two samples $S_1$ and $S_2$ to make predictions for its running time. Suppose that the two point estimates by using $S_1$ and $S_2$ are $\mu_1$ and $\mu_2$, respectively. We would then expect to see a picture as shown in Figure~\ref{fig:distr}(a), where the likelihoods of $\mu_1$ and $\mu_2$ match ``the'' distribution $D$ computed by our model. However, our model would actually compute a distribution $D_1$ describing its uncertainty about $\mu_1$, and compute a different distribution $D_2$ describing its uncertainty about $\mu_2$. We illustrate this in Figure~\ref{fig:distr}(b). Therefore, the expected $D$ is not unique. Rather, using different samples will result in different $D$'s.

\begin{figure}[!htb]
\centering
\includegraphics[width=0.7\columnwidth]{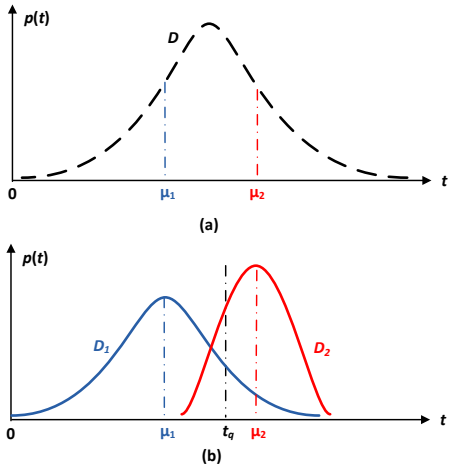}
\caption{Example point estimates and distributions}
\label{fig:distr}
\shrink
\end{figure}

Why should we derive different distributions if different samples are used? This is somehow not surprising. Different samples derived from the base tables may vary tremendously. As a result, the uncertainties in the selectivity estimates based on different samples may differ and hence the uncertainties in the running time estimates may differ as well.

\subsubsection{Comparison with Simplified Versions}

Another interesting question is if we can simplify some steps in our framework.
In the following we consider four alternatives (the complete version as well as three simplified versions):
\begin{enumerate}[(V1)]
\item \emph{All}: the complete version of our proposed framework;
\item \emph{No $\Var[c]$}: ignore the uncertainties in the cost units by setting $\Var[c]=0$ for each $c$;
\item \emph{No $\Var[X]$}: ignore the uncertainties in the selectivity estimates by setting $\Var[X]=0$ for each $X$;
\item \emph{No Cov}: ignore the covariances in the selectivity estimates.
\end{enumerate}

We compared these four alternatives for the \textbf{TPCH} queries.
Figure~\ref{fig:simplified:cc:uniform} presents typical results on uniformed databases in terms of the correlation coefficient $r_s$ (see Appendix~\ref{sec:more-exp-results:simplified} for more results on skewed databases).
We have several observations.
First, ignoring uncertainties in the $c$'s is not a good idea.
For all the cases we tested, this would lead to a drop of at least 0.25 (typically 0.4 to 0.5) in correlation.
Second, the impact of ignoring uncertainties in the $X$'s depends on the sample size.
Intuitively, as we increase the sample size, the uncertainties in the $X$'s diminish due to the strong consistency of the estimator.
When the uncertainties are small enough, ignoring them is safe.
As shown in Appendix~\ref{sec:more-exp-results:selectivity}, a sampling ratio of 1\% is already sufficient for accurate selectivity estimates for most of the queries we tested.
To observe the impact of ignoring the uncertainties in the $X$'s, we therefore used even lower sampling ratios.
As we can observe from Figure~\ref{fig:simplified:cc:uniform}, typically the correlation can drop by 0.2 to 0.3 when the sampling ratios are below 1\%, while it remains almost unaffected when 1\% samples are taken.\footnote{Note that the \emph{absolute} sample size is still not small when the sampling ratio is 1\%. Even for the 1GB TPC-H database, the largest \emph{lineitem} table contains 6,000,000 tuples and hence 60,000 sample tuples, which might be sufficient for most cases we tested.}
Third, while the impact of covariances in the $X$'s is often insignificant, sometimes ignoring the covariances causes problems.
For instance, as shown in Figure~\ref{fig:simplified:cc:uniform:b}, the correlations drop by 0.35 and 0.17 when the sampling ratios are 0.05\% and 0.1\%.
Although we cannot directly compute the covariances, our theoretic study in Appendix~\ref{sec:proofs} suggests that the upper bounds for the covariances become smaller as we increase the sample size.
Nevertheless, in general we have no idea how large the sample size needs to be so that we can safely ignore the uncertainties in the selectivity estimates and their covariances.
It depends on several factors such as the skewness of the data and the complexity of the queries in the workload.
Finally, the complete version is the most robust and effective one among the four alternatives: $r_s$ is consistently above 0.7 (most of the time above 0.8) for all the cases we tested.

\begin{figure}[!htb]
\centering
\subfigure[Uniform 1GB database, PC2]{ \label{fig:simplified:cc:uniform:a}
\includegraphics[width=\columnwidth]{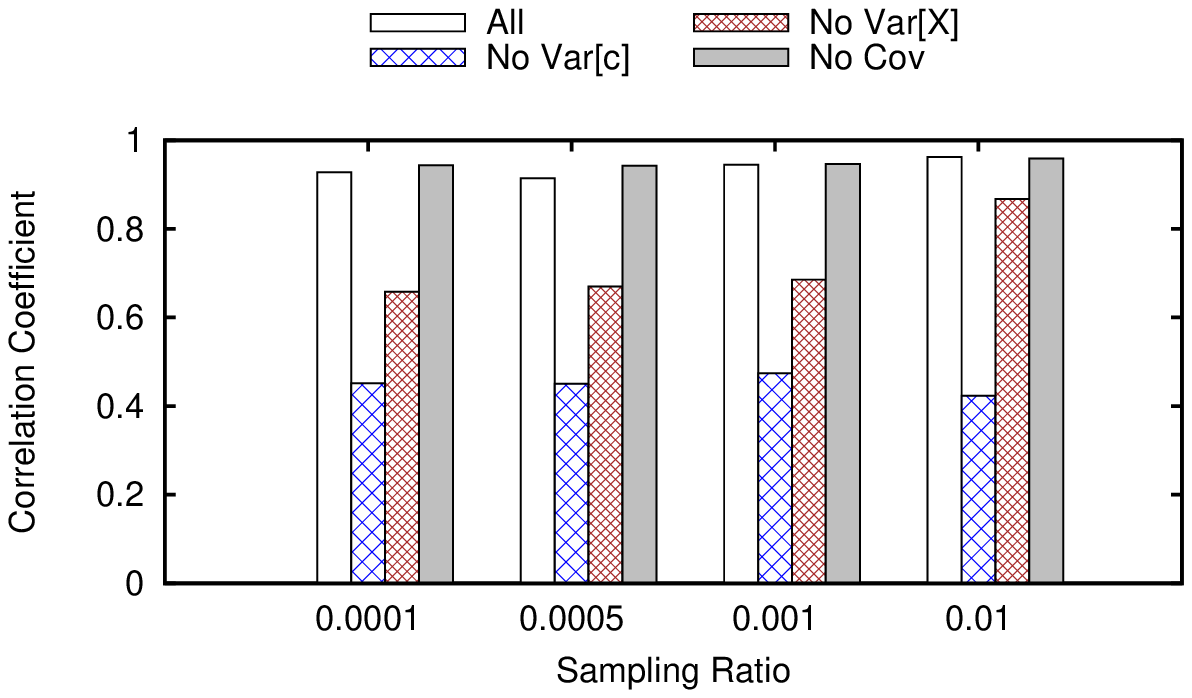}}
\subfigure[Uniform 10GB database, PC1]{ \label{fig:simplified:cc:uniform:b}
\includegraphics[width=\columnwidth]{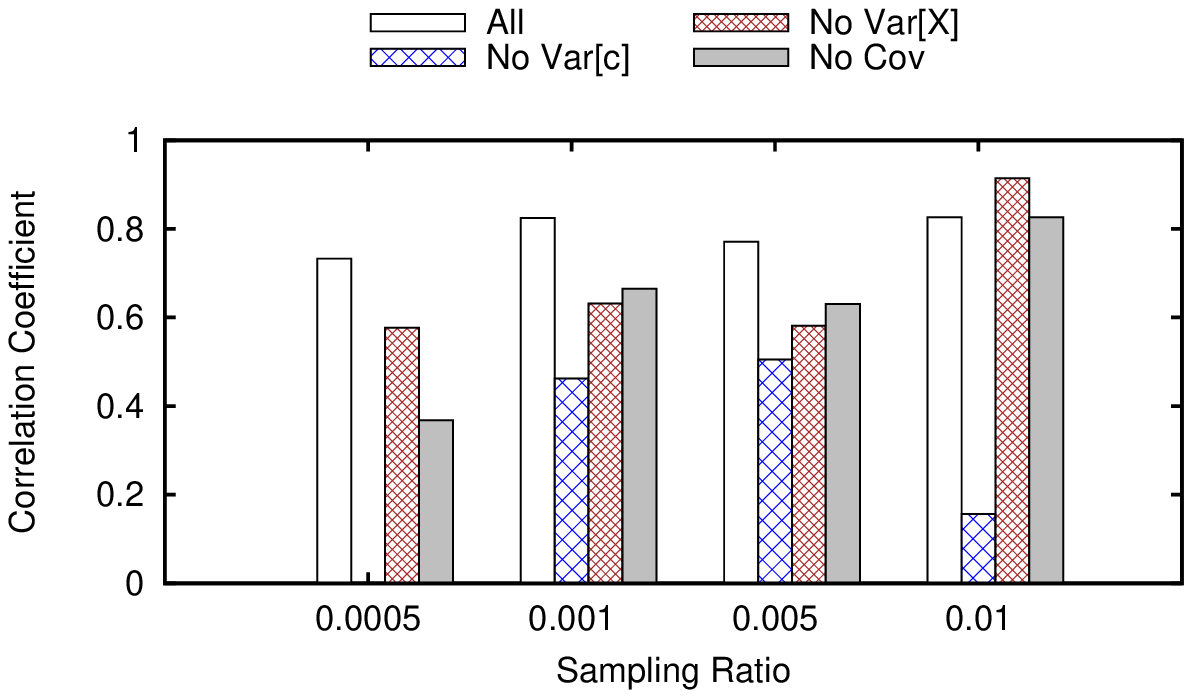}}
\caption{Comparison of four alternatives in terms of $r_s$}
\label{fig:simplified:cc:uniform}
\shrink
\end{figure}

\subsection{Runtime Overhead of Sampling}\label{sec:exp:sample-overhead}

We also measured the relative overhead of running the queries over the sample tables compared with that of running them over the original tables. Figure~\ref{fig:overhead:tpch:pc1} presents the results of the \textbf{TPCH} queries on PC1. Since the other results are very similar, the readers are referred to Appendix~\ref{sec:more-exp-results:overhead} for the complete details. We observe that the relative overhead is comparable to that reported in~\cite{WuCZTHN13}. For instance, for the TPC-H 10GB database, the relative overhead is around 0.04 to 0.06 when the sampling ratio is 0.05. Note that, here we computed the estimated selectivities as well as their variances by only increasing the relative overhead a little. Also note that, here we measured the relative overhead based on disk-resident samples. The relative overhead can be dramatically reduced by using the common practice of caching the samples in memory~\cite{Ramamurthy-bufferQO}.

\begin{figure}
\centering
\includegraphics[width=1.0\columnwidth]{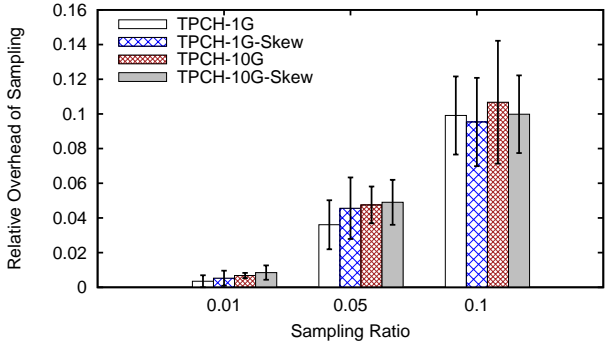}
\caption{Relative overhead of \textbf{TPCH} queries on PC1}
\label{fig:overhead:tpch:pc1}
\shrink
\end{figure}

On the other hand, though the Central Limit Theorem guarantees that the selectivity estimates are Gaussian-distributed for sufficiently large samples, it does not say anything about how large the samples should be.
In fact, there is no exact answer to this question.
As a rule of thumb, statisticians have agreed that the sample size should be larger than or equal to 30 in general, and the larger the better~\cite{Brase09}.

\subsection{Applications}
We discuss some potential applications that could take advantage of the distributional information of query running times. The list of applications here is by no means exhaustive, and it is our hope that our study in this paper could stimulate further research in this direction and more applications could emerge in the future.

\subsubsection{Query Optimization}
Although significant progress has been made in the past several decades, query optimization remains challenging for many queries due to the difficulty in accurately estimating query running times. Rather than betting on the optimality of the plan generated based on (perhaps erroneous) point estimates for parameters such as selectivities and cost units, it makes sense to also consider the uncertainties of these parameters. In fact, there has been some theoretical work investigating optimization based on least {\em expected} cost (LEC) based upon distributions of the parameters of the cost model~\cite{ChuHS99}. However, that work did not address the problem of how to obtain the distributions. It would be interesting to see the effectiveness of LEC plans by incorporating our techniques into query optimizers.

\subsubsection{Query Progress Monitoring}
State-of-the-art query progress indicators~\cite{ChaudhuriNR04,Christian-progIndVLDB12,Li-jessie,LuoNEW04} provide estimates of the percentage of the work that has been completed by a query at regular intervals during the query's execution. However, it has been shown that in the worst case no progress indictor can outperform a naive indicator simply saying the progress is between 0\% and 100\%~\cite{ChaudhuriKR05}. Hence, information about uncertainty in the estimate of progress is desirable. Our work provides a natural building block that could be used to develop an uncertainty-aware query progress indicator: the progress indicator could call our predictor to make a prediction for the remaining query running time as well as its uncertainty.

\subsubsection{Database as a Service}
The problem of predicting query running time is revitalized by the recent move towards providing database as a service (DaaS). Many important decision-making procedures, including admission control~\cite{Tozer-QCop,Xiong-ActiveSLA}, query scheduling~\cite{ChiHHN13}, and system sizing~\cite{Wasserman-dbSizing}, rely on estimation of query running time. Distributional information enables more robust decision procedures in contrast to point estimates. Recent work~\cite{ChiHHN13} has shown the benefits in query scheduling by leveraging distributional information. Similar ideas have also been raised in~\cite{Xiong-ActiveSLA} for admission control. Again, these work did not address the fundamental issue of obtaining the distributions without running the queries. It would be interesting to see the effectiveness of our proposed techniques in these DaaS applications.

\section{Related Work}\label{sec:relatedwork}

The problem of predicting query execution time has been extensively studied quite recently~\cite{AhmadDAB11-edbt,AkdereCRUZ12-brown-icde,DugganCPU11,Ganapathi-berkeley09,Li12Robust,WuCHN13,WuCZTHN13}. Ganapathi et al.~\cite{Ganapathi-berkeley09} first raised the question of predicting the \emph{actual} running time of a query rather than a \emph{rough} estimate of runtime overhead that is usually provided by most, if not all, query optimizers. They further proposed a predictive approach based on Kernel Canonical Correlation Analysis (KCCA). Follow-up approaches improved the prediction accuracy by using different machine learning models such as Support Vector Machines (SVM)~\cite{AkdereCRUZ12-brown-icde} or Multiple Additive Regression-Trees (MART)~\cite{Li12Robust}. Unlike these data-driven machine learning approaches that treated the underlying database system as a black box, we proposed a predictor based on calibrating the query optimizer's cost models and showed that it could often outperform the machine learning based approaches in terms of prediction accuracy~\cite{WuCZTHN13}. While most of this line of work focused on the single-query prediction problem, some of them have considered the more general prediction problem when multiple queries are concurrently running. Ahmad et al. addressed the problem by using Gaussian processes~\cite{AhmadDAB11-edbt}, while Duggan et al. adopted similar ideas but instead used multivariate linear regression~\cite{DugganCPU11}. Both of them, however, assumed \emph{static} database workloads, in the sense that all queries running in the system should be known beforehand. To overcome this limitation, we proposed a conceptually different approach by extending our single-query predictor~\cite{WuCHN13}. We first used query optimizer's cost models to estimate the CPU and I/O requirements for each query, and then used a combination queueing model and buffer pool model to merge these quantities from concurrent queries to predict running times. Nonetheless, none of these work ever considered the problem of measuring the degree of uncertainty in the prediction.
We have reused some techniques developed in~\cite{WuCZTHN13} for computing the means of selectivities and cost units when viewed as random variables.
Nonetheless,~\cite{WuCZTHN13} focused on point estimates rather than distributional information, and hence these techniques were insufficient.
We have substantially extended~\cite{WuCZTHN13} by developing new techniques for computing variances (and hence distributions) of selectivity and cost-unit estimates (Section~\ref{sec:distribution}), cost functions (Section~\ref{sec:costfunc}), and, based on that, distributions of likely running times (Section~\ref{sec:uncertainty}).

The idea of using samples to estimate selectivity goes back more than two decades ago (e.g.,~\cite{BabcockC05,Charikar-sample00,Haas-distinct95,Haas-sample96,Haas-sample92,Hou91,HouOT88,Lipton-sample90}). While we focused on estimators for selection and join queries~\cite{Haas-sample96}, some estimators that estimate the number of distinct values might be further used to refine selectivity estimates of aggregate queries~\cite{Charikar-sample00,Haas-distinct95}. However, not only do we need an estimate of selectivity, we need an estimated distribution as well. So far, we are not aware of any previous study towards this direction for aggregate queries. Regarding the problem of estimating selectivity distributions for selection and join queries, there are options other than the one used in this paper. For example, Babcock and Chaudhuri~\cite{BabcockC05} proposed a framework to learn the posterior distributions of the selectivities based on \emph{join synopses}~\cite{AcharyaGPR99}. Unfortunately, this solution is restricted to SPJ expressions with foreign-key joins, due to the overhead of computing and maintaining join synopses over a large database.

The framework proposed in this paper also relies on accurate approximation of the cost models used by the optimizer. Du et al.~\cite{Du-calib92} first proposed the idea of using logical cost functions in the context of heterogenous database systems. Similar ideas were later on used in developing generic cost models for main memory based database systems~\cite{ManegoldBK02} and identifying robust plans in the plan diagram generated by the optimizer~\cite{DDH08}. Our idea of further using optimization techniques to find the best coefficients in the logical cost functions is motivated by the approach used in~\cite{DDH08}. 

\section{Conclusion} \label{sec:conclusion}

In this paper, we take a first step towards the problem of measuring the uncertainty within query execution time prediction. We quantify prediction uncertainty using the distribution of likely running times. Our experimental results show that the standard deviations of the distributions estimated by our proposed approaches are strongly correlated with the actual prediction errors.

The idea of leveraging cost models to quantify prediction uncertainty need not be restricted to single standalone queries. As shown in~\cite{WuCHN13}, Equation~(\ref{eq:tquery}) can also be used to provide point estimates for multiple concurrently-running queries. The key observation is that the selectivities of the operators in a query are independent of whether or not it is running with other queries. Hence it is promising to consider applying the techniques proposed in this paper to multi-query workloads by viewing the interference between queries as changing the distribution of the $c$'s. We regard this as a compelling area for future work. 

\normalsize

{\renewcommand{\baselinestretch}{1.05}
\small

\bibliographystyle{abbrv}
\bibliography{querytime}
}

\newpage
\appendix

\section{Theoretic Results}\label{sec:proofs}

This section presents the proofs of the lemmas and theorems as well as other related results mentioned in the paper.

\subsection{Variance of The Estimator $\rho_n$}\label{sec:proofs:theorem:variance} 

The variance of the selectivity estimator $\rho_n$, unfortunately, is nontrivial when writing it mathematically:
\begin{theorem}\label{theorem:variance}
The variance of $\rho_n$ is~\cite{Haas-sample96}:
\begin{eqnarray}\label{eq:var-estimator-cp}
\Var[\rho_n]&=&\sum_{r=1}^K\frac{(n-1)^{K-r}}{n^K}\\\nonumber
&\times&\sum_{S\in\mathcal{S}_r}\bigg(\frac{1}{|\Lambda(S)|}\sum_{\mathbf{l}\in\Lambda(S)}(\rho_S(\mathbf{l})-\rho)^2\bigg).
\end{eqnarray}
\end{theorem}
Here $S=\{k_1,...,k_r\}\subseteq\{1,2,...,K\}$ such that $k_1<k_2<\cdots<k_r$, $\mathcal{S}_r$ is the collection of all subsets of $\{1,2,...,K\}$ with size $r$ (for $1\leq r\leq K$), and $\Lambda(S)$ is defined to be
$$\Lambda(S)=\{1,2,...,m_{k_1}\}\times\cdots\times \{1,2,...,m_{k_r}\}.$$
Furthermore, for $\mathbf{l}=(l_1,...,l_r)\in\Lambda(S)$, $\rho_S(\mathbf{l})$ is the average selectivity over $\mathbf{B}(L_1, ..., L_K)$ such that $L_{k_j}=l_j$ ($1\leq j\leq r$). For example, if $K=4$, $S=\{2,3\}$ and $\mathbf{l}=(8,9)$, then $$\rho_S(\mathbf{l})=(m_1m_4)^{-1}\sum_{L_1=1}^{m_1}\sum_{L_4=1}^{m_4}\rho_{\mathbf{B}(L_1,8,9,L_4)}.$$

We next prove Theorem~\ref{theorem:variance}. Roughly speaking, the idea of the proof is to first partition the samples into groups based on how many blocks they share, then compute the variance of each group, and finally sum them up. We start with the following standard result from probability theory:
\begin{lemma}\label{lemma:variance}
Let $X_1$, ..., $X_n$ be $n$ random variables, then
$$\Var[\sum_{i=1}^n X_i]=\sum_{i=1}^n\sum_{j=1}^n \Cov(X_i, X_j),$$
where $\Cov(X_i, X_j)$ is the covariance of $X_i$ and $X_j$:
$$\Cov(X_i, X_j)=\E[(X_i-\E[X_i])(X_j-\E[X_j])].$$
\end{lemma}

\noindent Now define $\mathcal{X}$ to be the set of all sample blocks, namely,
$$\mathcal{X}=\{\rho_{\mathbf{B}(L_{1,i_1},\cdots,L_{K,i_K})}|1\leq i_k\leq n, 1\leq k\leq K\}.$$
Based on Lemma~\ref{lemma:variance} and Equation~(\ref{eq:estimator-cp}), we have
\begin{equation}\label{eq:var-estimator-cp2}
\Var[\rho_n]=\frac{1}{(n^K)^2}\sum_{X\in\mathcal{X}}\sum_{X'\in\mathcal{X}}\Cov(X, X').
\end{equation}
Consider any $X=\rho_{\mathbf{B}(j_1,...,j_K)}$ and $X'=\rho_{\mathbf{B}(j'_1,...,j'_K)}$ in the summands of Equation~(\ref{eq:var-estimator-cp2}). If $j_k\neq j'_k$ for $1\leq k\leq K$, then $X$ and $X'$ are independent and $\Cov(X,X')=0$. Hence, only $X$ and $X'$ that share at least one common coordinate will contribute a non-zero summand to Equation~(\ref{eq:var-estimator-cp2}). We thereby partition the pairs $(X,X')$ according to the number of coordinates they share. Specifically, for $S=\{k_1,...,k_r\}\subseteq\{1,2,...,K\}$, we denote $X\sim_S X'$ if $j_{k_m}=j'_{k_m}$ for $1\leq m\leq r$. This gives us the following equivalent expression for $\Var[\rho_n]$:
\begin{equation}\label{eq:var-estimator-cp3}
\Var[\rho_n]=\frac{1}{(n^K)^2}\sum_{r=1}^{K}\sum_{S\in\mathcal{S}_r}\sum_{X\sim_S X'}\Cov(X, X').
\end{equation}

\begin{lemma}\label{lemma:npairs}
For a fixed $S\in\mathcal{S}_r$, the number of pairs $(X,X')$ such that $X\sim_S X'$ is $(n-1)^{K-r}n^K$. As a result, $\Var[\rho_n]$ can be further expressed as:
\begin{equation}\label{eq:var-estimator-cp4}
\Var[\rho_n]=\sum_{r=1}^{K}\frac{(n-1)^{K-r}}{n^K}\times\sum_{S\in\mathcal{S}_r}\Cov(X, X').
\end{equation}
\end{lemma}

Our next goal is to give an expression for $\Cov(X,X')$ when $X\sim_S X'$, as shown in Lemma~\ref{lemma:covariance}. Equation~(\ref{eq:var-estimator-cp}) in Theorem~\ref{theorem:variance} then follows by combining Lemma~\ref{lemma:npairs} and~\ref{lemma:covariance}.
\begin{lemma}\label{lemma:covariance}
If $X\sim_S X'$, then
$$\Cov(X,X')=\frac{1}{|\Lambda(S)|}\sum_{\mathbf{l}\in\Lambda(S)}(\rho_S(\mathbf{l})-\rho)^2.$$
\end{lemma}
\begin{proof}
We have
$$\Cov(X,X')=\E_{X\sim_S X'}[(X-\E[X])(X'-\E[X'])].$$
Since $\E[X]=\E[X']=\rho$, it follows that
$$\Cov(X,X')=\E_{X\sim_S X'}[(X-\rho)(X'-\rho)].$$
We further denote $X\sim_{S(\mathbf{l})} X'$ for $\mathbf{l}=(l_1,...,l_r)$, if $X\sim_S X'$ and $j_{k_m}=l_m$ for $1\leq m\leq r$. We then have
$$\Cov(X,X')=\frac{1}{|\Lambda(S)|}\sum_{\mathbf{l}\in\Lambda(S)}\E_{X\sim_{S(\mathbf{l})} X'}[(X-\rho)(X'-\rho)].$$
Now consider $\E_{X\sim_{S(\mathbf{l})} X'}[(X-\rho)(X'-\rho)]$. By definition, it is the average of the following quantities
$$\gamma=(\rho_{\mathbf{B}(j_1, ..., j_K)}-\rho)(\rho_{\mathbf{B}(j'_1, ..., j'_K)}-\rho)$$
by setting $j_{k_m}=j'_{k_m}=l_m$ for $1\leq m\leq r$. Let $S^c=\{1,...,K\}-S$ be the complement of $S$. We then have
$$E=\E_{X\sim_{S(\mathbf{l})} X'}[(X-\rho)(X'-\rho)]=\big(\frac{1}{|\Lambda(S^c)|}\big)^2\sum_{\Lambda(S^c)}\sum_{\Lambda(S^c)}\gamma.$$
After some rearrangement of the summands, we can have
\begin{eqnarray*}
E&=&\big(\frac{1}{|\Lambda(S^c)|}\big)^2\big(\sum_{\Lambda(S^c)}(\rho_{\mathbf{B}(j_1, ..., j_K)}-\rho)\big)^2\\
&=&\big(\frac{1}{|\Lambda(S^c)|}\sum_{\Lambda(S^c)}(\rho_{\mathbf{B}(j_1, ..., j_K)}-\rho)\big)^2\\
&=&\big((\frac{1}{|\Lambda(S^c)|}\sum_{\Lambda(S^c)}\rho_{\mathbf{B}(j_1, ..., j_K)})-\rho\big)^2\\
&=&(\rho_S(\mathbf{l})-\rho)^2.
\end{eqnarray*}
This completes the proof of the lemma.
\end{proof}

\subsection{Proof of Lemma~\ref{lemma:c4var}} \label{sec:proofs:lemma:c4var}

\begin{proof}
Table~\ref{tab:non-central-moments} presents the non-central moments of a normal variable $X\sim\mathcal{N}(\mu,\sigma^2)$. By Table~\ref{tab:non-central-moments}, $\Var[X_l^2]=2\sigma_l^2(2\mu_l^2+\sigma_l^2)$, and $\Cov(X_l^2, X_l)=2\mu_l\sigma_l^2$. Thus
\begin{eqnarray*}
\Var[f]&=&b_0^2 \Var[X_l^4] + b_1^2 \Var[X_l] + 2b_0b_1\Cov(X_l^2, X_l)\\
&=&\sigma_l^2 [(b_1 + 2 b_0 \mu_l)^2 + 2 b_0^2 \sigma_l^2].
\end{eqnarray*}
This completes the proof of the lemma.
\end{proof}

\begin{table}
\centering
\begin{tabular}{|r|l|}
\hline
$k$ & Non-central moment $E(X^k)$\\
\hline
$1$ & $\mu$ \\
$2$ & $\mu^2+\sigma^2$ \\
$3$ & $\mu^3+3\mu\sigma^2$ \\
$4$ & $\mu^4+6\mu^2\sigma^2+3\sigma^4$\\
\hline
\end{tabular}
\caption{Non-central moments of $X\sim N(\mu, \sigma^2)$}
\label{tab:non-central-moments}
\shrink
\end{table}

\subsection{Proof of Theorem~\ref{theorem:c4approx}} \label{sec:proofs:theorem:c4approx}

To prove the theorem, we need the following result:

\begin{theorem}\label{theorem:varbound}
$\Var[\rho_n]$ in Equation~(\ref{eq:var-estimator-cp}) can be bounded as:
\begin{equation*}
\Var[\rho_n]\leq\big(1-(1-\frac{1}{n})^K\big)\rho(1-\rho).
\end{equation*}
\end{theorem}

\begin{proof}
The proof is straightforward since this is a special case of Theorem~\ref{theorem:dn-bound} (see Appendix~\ref{sec:more-cov-bounds}). Specifically,
we have $\Var[\rho_n]=\Cov(\rho_n,\rho_n)$. By letting $m=K$ in Theorem~\ref{theorem:dn-bound}, we obtain
$$\Var[\rho_n]\leq f(n, K)g(\rho)^2,$$
where $f(n,K)=1-(1-\frac{1}{n})^K$ and $g(\rho)=\sqrt{\rho(1-\rho)}$. This completes the proof.
\end{proof}

According to Theorem~\ref{theorem:varbound}, $\Var[\rho_n]\to 0$ as $n\to\infty$. We are now ready to prove Theorem~\ref{theorem:c4approx}:

\begin{proof} (of Theorem~\ref{theorem:c4approx})
Let $\mu_l=\rho_n$, and $\E[\rho_n]=\rho$. Define $g(X)=b_0X^2+b_1X+b_2$. Since $\rho_n$ is strongly consistent, $\rho_n\as\rho$. Moreover, since $g$ is continuous, we have $f=g(\rho_n)\as g(\rho)$ by the continuous mapping theorem. Note that $g(\rho)$ is a constant. On the other hand, by Lemma~\ref{lemma:c4var} and Theorem~\ref{theorem:varbound}, $\Var[f]\to 0$ as $n\to\infty$. As a result, $$f^\mathcal{N}\indistr \E[f]=g(\rho).$$
Since $f\as g(\rho)$ implies $f\inprob g(\rho)$,
$$f^\mathcal{N}-f\indistr g(\rho)-g(\rho)=0$$
by Slutsky's theorem. Since $0$ is a constant, $f^\mathcal{N}-f\inprob 0$ as well. As a result, we have $f^\mathcal{N}\inprob f$.
\end{proof}

\subsection{Similar Results for (C6')}\label{sec:proofs:theory:c6}

\begin{lemma}\label{lemma:c6var}
If $X_l\sim \mathcal{N}(\mu_l, \sigma_l^2)$, $X_r\sim \mathcal{N}(\mu_r, \sigma_r^2)$, and $f=b_0X_lX_r + b_1X_l + b_2X_r + b_3$, then
$$\Var[f]=\sigma_l^2\big(b_0\mu_r+b_1\big)^2+\sigma_r^2\big(b_0\mu_l+b_2\big)^2+b_0^2\sigma_l^2\sigma_r^2.$$
\end{lemma}

\begin{proof}
Since $X_l\bot X_r$, $\Cov(X_l, X_r)=0$. So
\begin{eqnarray*}
\Var[f]&=&b_0^2\cdot \Var[X_lX_r]+b_1^2\sigma_l^2+b_2^2\sigma_r^2\\
&+&2b_0b_1\cdot \Cov(X_lX_r, X_l)\\
&+&2b_0b_2\cdot \Cov(X_lX_r,X_r).
\end{eqnarray*}
Since $\Var[X_lX_r]=\mu_l^2\sigma_r^2+\mu_r^2\sigma_l^2+\sigma_l^2\sigma_r^2$,
$\Cov(X_lX_r,X_l)=\mu_r\sigma_l^2$, and similarly
$\Cov(X_lX_r,X_r)=\mu_l\sigma_r^2$,
we can have the desired expression for $\Var[f]$ by substituting these quantities.
\end{proof}

\begin{theorem}\label{theorem:c6approx}
Suppose that $X_l\sim \mathcal{N}(\mu_l, \sigma_l^2)$, $X_r\sim \mathcal{N}(\mu_r, \sigma_r^2)$, and $f=b_0X_lX_r + b_1X_l + b_2X_r + b_3$. Let $f^\mathcal{N}\sim \mathcal{N}(\E[f], \Var[f])$, where $\Var[f]$ is shown in Lemma~\ref{lemma:c6var}. Then $f^\mathcal{N}\inprob f$.
\end{theorem}

\begin{proof}
Let $\mu_l=\rho_n$ and $\mu_r=\rho'_n$. Suppose that $\E[\rho_n]=\rho$ and $\E[\rho'_n]=\rho'$. Define
$$g(X,Y)=b_0XY+b_1X+b_2Y+b_3.$$
Since $\mu_l$ and $\mu_r$ are both strongly consistent, $\rho_n\as\rho$ and $\rho'_n\as\rho'$. Moreover, since $g$ is continuous, by the continuous mapping theorem we have
$$f=g(\rho_n,\rho'_n)\as g(\rho, \rho').$$
Note that $g(\rho,\rho')$ is a constant. On the other hand, by Lemma~\ref{lemma:c6var} and Theorem~\ref{theorem:varbound}, $\Var[f]\to 0$ as $n\to\infty$. As a result, since $X_l\bot X_r$ by Lemma~\ref{lemma:ind-bin}, it follows that
$$f^\mathcal{N}\indistr \E[f]=g(\rho,\rho')$$
Since $f\as g(\rho,\rho')$ implies $f\inprob g(\rho,\rho')$,
$$f^\mathcal{N}-f\indistr g(\rho,\rho')-g(\rho,\rho')=0$$
by Slutsky's theorem. Since $0$ is a constant, $f^\mathcal{N}-f\inprob 0$ as well. As a result, we have $f^\mathcal{N}\inprob f$.
\end{proof}

\subsection{Proof of Theorem~\ref{theorem:tc}}\label{sec:proofs:theorem:tc}

\begin{proof}
Since $f_{kc}$ and $c$ are independent, we have
$$\E[t_{kc}]=\E[f_{kc}^\mathcal{N}c]=\E[f_{kc}]\E[c]$$
and
$$\Var[t_{kc}]=E^2[f_{kc}]\Var[c]+ E^2[c]\Var[f_{kc}]+\Var[c]\Var[f_{kc}].$$
Since $\Var[f_{kc}]\to 0$ as $n\to\infty$,
$$\Pr(t_{kc}^\mathcal{N})\to \mathcal{N}(\E[f_{kc}]\E[c], E^2[f_{kc}]\Var[c]).$$ In other words, $t_{kc}^\mathcal{N}\indistr \E[f_{kc}]c$.

On the other hand, $f_{kc}^\mathcal{N}\inprob \E[f_{kc}]$ and $c\inprob c$. As a result, we have $(f_{kc}^\mathcal{N}, c)\inprob (\E[f_{kc}], c)$. By the continuous mapping theorem, $f_{kc}^\mathcal{N}c \inprob \E[f_{kc}]c$. That is, $t_{kc} \inprob \E[f_{kc}]c$, which implies $t_{kc} \indistr \E[f_{kc}]c$. This completes the proof of the theorem.
\end{proof}

\subsection{Convergence of $g_c^\mathcal{N}$}\label{sec:proofs:theorem:gc}

\begin{theorem}\label{theorem:gc}
Let $g_c=\sum_{k=1}^m f_{kc}^\mathcal{N}$ and
$$g_c^\mathcal{N}\sim\mathcal{N}(\E[g_c], \Var[g_c]).$$
Then $g_c^\mathcal{N}\inprob g_c$.
\end{theorem}

\begin{proof}
Since by definition $f_{kc}^\mathcal{N}\sim\mathcal{N}(\E[f_{kc}],\Var[f_{kc}])$ and $\Var[f_{kc}]\to 0$, $f_{kc}^\mathcal{N}\indistr \E[f_{kc}]$. Since $\E[f_{kc}]$ is a constant, it implies that $f_{kc}^\mathcal{N}\inprob \E[f_{kc}]$. By the continuous mapping theorem, $g_c\inprob \sum_{k=1}^m \E[f_{kc}]$. On the other hand, since $\Var[g_c]\to 0$, $g_c^\mathcal{N}\indistr \E[g_c]$. Since $\E[g_c]$ is again a constant, it follows that
$$g_c^\mathcal{N}\inprob \E[g_c]=\sum_{k=1}^m \E[f_{kc}].$$
As a result, by applying the continuous mapping theorem again, we have $g_c^\mathcal{N}-g_c\inprob 0$ and hence $g_c^\mathcal{N}\inprob g_c$.
\end{proof}

\subsection{A Tighter Upper Bound for Covariance}\label{sec:proofs:theorem:tighter-bound}

Consider two operators $O$ and $O'$ where $O\in Desc(O')$. Suppose that $|\mathcal{R}|=K$,
$|\mathcal{R}'|=K'$, and $|\mathcal{R}\cap\mathcal{R}'|=m$ ($m\geq 1$). Let the estimators for $O$ and $O'$
be $\rho_n$ and $\rho'_n$ where $n$ is the number of sample steps, and define $\rho=\E[\rho_n]$ and $\rho'=\E[\rho'_n]$.

\begin{theorem}\label{theorem:tighter-bound}
Let $\mathcal{S}_r$, $\Lambda(S)$, and $\rho_S(\mathbf{l})$ be the same as that defined in Theorem~\ref{theorem:variance}. Define
$$\sigma_S^2=\frac{1}{|\Lambda(S)|}\sum_{\mathbf{l}\in\Lambda(S)}(\rho_S(\mathbf{l})-\rho)^2,$$
and
$$S_{\rho}^2(m,n)=\sum_{r=1}^m\big(1-\frac{1}{n}\big)^{m-r}\big(\frac{1}{n}\big)^r\sum_{S\in\mathcal{S}_r}\sigma_S^2.$$
We then have
$$|\Cov(\rho_n,\rho'_n)|\leq\sqrt{S_{\rho}^2(m,n)S_{\rho'}^2(m,n)}\leq \sqrt{\Var[\rho_n]\Var[\rho'_n]}.$$
\end{theorem}

We next prove Theorem~\ref{theorem:tighter-bound}. To establish the first inequality in the theorem, namely,
$$|\Cov(\rho_n,\rho'_n)|\leq\sqrt{S_{\rho}^2(m,n)S_{\rho'}^2(m,n)},$$
we need two lemmas. The first one gives an explicit expression of the covariance $\Cov(\rho_n,\rho'_n)$, which is quite similar to the expression of $\Var[\rho_n]$ shown in Theorem~\ref{theorem:variance}.

\begin{lemma}\label{lemma:dn}
Let $\mathcal{S}_r$ be the collection of all subsets of $\mathcal{R}$ with size $r$ (for $1\leq r\leq m$). Define
$$\Cov_S(\rho,\rho')=\frac{1}{|\Lambda(S)|}\sum_{\mathbf{l}\in\Lambda(S)}(\rho_S(\mathbf{l})-\rho)(\rho'_S(\mathbf{l})-\rho').$$
Then
\begin{eqnarray*}
\Cov(\rho_n,\rho'_n)=\sum_{r=1}^m\frac{(n-1)^{m-r}}{n^m}\sum_{S\in \mathcal{S}_r}\Cov_S(\rho,\rho').
\end{eqnarray*}
Here $\rho_S(\mathbf{l})$ and $\rho'_S(\mathbf{l})$ are the same as that in Theorem~\ref{theorem:variance}, defined over $\mathcal{R}$ and $\mathcal{R}'$ respectively.
\end{lemma}

\begin{proof}
The idea is similar to our proof of Theorem~\ref{theorem:variance}. Let $K=|\mathcal{R}|$ and $K'=|\mathcal{R}'|$. We have
$$\rho_n=\frac{1}{n^K}\sum_{i_1=1}^{n}\cdots\sum_{i_K=1}^{n}\rho_{\mathbf{B}(L_{1,i_1},\cdots,L_{K,i_K})},$$
and
$$\rho'_n=\frac{1}{n^{K'}}\sum_{i_1=1}^{n}\cdots\sum_{i_{K'}=1}^{n}\rho'_{\mathbf{B}(L_{1,i_1},\cdots,L_{K,i_{K'}})}.$$
Therefore,
$$\E[\rho_n\rho'_n]=\frac{1}{n^{K+K'}}\sum_{k=1}^{n^K}\sum_{k'=1}^{n^{K'}}\E[\rho_{\mathbf{B}}\rho'_{\mathbf{B}}],$$
and
$$\E[\rho_n]\E[\rho'_n]=\frac{1}{n^{K+K'}}\sum_{k=1}^{n^K}\sum_{k'=1}^{n^{K'}}\E[\rho_{\mathbf{B}}]\E[\rho'_{\mathbf{B}}].$$
Hence, by letting $d_n=\Cov(\rho_n,\rho'_n)=\E[\rho_n\rho'_n]-\E[\rho_n]\E[\rho'_n]$,
\begin{eqnarray*}
d_n&=&\frac{1}{n^{K+K'}}\sum_{k=1}^{n^K}\sum_{k'=1}^{n^{K'}}\big(\E[\rho_{\mathbf{B}}\rho'_{\mathbf{B}}]-\E[\rho_{\mathbf{B}}]\E[\rho'_{\mathbf{B}}]\big)\\
&=&\frac{1}{n^{K+K'}}\sum_{k=1}^{n^K}\sum_{k'=1}^{n^{K'}}\Cov(\rho_{\mathbf{B}},\rho'_{\mathbf{B}}).
\end{eqnarray*}
If $\mathbf{B}$ and $\mathbf{B}'$ share no blocks, then $\rho_{\mathbf{B}}$ and $\rho'_{\mathbf{B}}$ are independent and thus $\Cov(\rho_{\mathbf{B}},\rho'_{\mathbf{B}})=0$. Thus we only need to consider the case that $\mathbf{B}$ and $\mathbf{B}'$ share at least one block. Similarly as before, we partition the pairs $(\rho_{\mathbf{B}},\rho'_{\mathbf{B}})$ based on the number of blocks $\mathbf{B}$ and $\mathbf{B}'$ share. According to Lemma~\ref{lemma:npairs}, for a fixed $S\in\mathcal{S}_r$, the number of pairs $(\rho_{\mathbf{B}},\rho'_{\mathbf{B}})$ such that $\rho_{\mathbf{B}}\sim_S\rho'_{\mathbf{B}}$ is
$$n^K(n-1)^{m-r}n^{K'-m}=n^{K+K'-m}(n-1)^{m-r}.$$
We hence have
$$\Cov(\rho_n,\rho'_n)=\sum_{r=1}^{m}\frac{(n-1)^{m-r}}{n^m}\times\sum_{S\in\mathcal{S}_r}\Cov(\rho_{\mathbf{B}},\rho'_{\mathbf{B}}).$$
Similarly as in Lemma~\ref{lemma:covariance}, we have
$$\Cov(\rho_{\mathbf{B}},\rho'_{\mathbf{B}})=\E_{\rho_{\mathbf{B}}\sim_S\rho'_{\mathbf{B}}}[(\rho_{\mathbf{B}}-\rho)(\rho'_{\mathbf{B}}-\rho')],$$
and hence
$$\Cov(\rho_{\mathbf{B}},\rho'_{\mathbf{B}})=\frac{1}{|\Lambda(S)|}\sum_{\mathbf{l}\in\Lambda(S)}\E_{\rho_{\mathbf{B}}\sim_{S(\mathbf{l})}\rho'_{\mathbf{B}}}
[(\rho_{\mathbf{B}}-\rho)(\rho'_{\mathbf{B}}-\rho')].$$
Now let $\mathcal{K}$ and $\mathcal{K}'$ be the indexes of the relations in $\mathcal{R}$ and $\mathcal{R}'$ respectively. Denote $S^c_K=\mathcal{K}-S$ and $S^c_{K'}=\mathcal{K}'-S$. Let $$E=\E_{\rho_{\mathbf{B}}\sim_{S(\mathbf{l})}\rho'_{\mathbf{B}}}[(\rho_{\mathbf{B}}-\rho)(\rho'_{\mathbf{B}}-\rho')].$$
We have
\begin{eqnarray*}
E&=&\frac{1}{|\Lambda(S^c_K)|\cdot|\Lambda(S^c_{K'})|}\sum_{\Lambda(S^c_K)}\sum_{\Lambda(S^c_{K'})}
\big((\rho_{\mathbf{B}}-\rho)(\rho'_{\mathbf{B}}-\rho')\big)\\
&=&\big(\frac{1}{|\Lambda(S^c_K)|}\sum_{\Lambda(S^c_K)}(\rho_{\mathbf{B}}-\rho)\big)
\big(\frac{1}{|\Lambda(S^c_{K'})|}\sum_{\Lambda(S^c_{K'})}(\rho'_{\mathbf{B}}-\rho')\big)\\
&=&\big((\frac{1}{|\Lambda(S^c_K)|}\sum_{\Lambda(S^c_K)}\rho_{\mathbf{B}})-\rho\big)
\big((\frac{1}{|\Lambda(S^c_{K'})|}\sum_{\Lambda(S^c_{K'})}\rho'_{\mathbf{B}})-\rho'\big)\\
&=&(\rho_S(\mathbf{l})-\rho)(\rho'_S(\mathbf{l})-\rho').
\end{eqnarray*}
This completes the proof of the lemma.
\end{proof}

\noindent Our second lemma further provides an upper bound for $\Cov_S(\rho,\rho')$:
\begin{lemma}\label{lemma:ds}
Let $S\in\mathcal{S}_r$. Then we have
$$|\Cov_S(\rho,\rho')|\leq\sqrt{\sigma_S^2\cdot(\sigma'_S)^2}.$$
\end{lemma}

\begin{proof}
Let $d_S=\Cov_S(\rho,\rho')$ and $d_{\rho}^2=(\rho_S(\mathbf{l})-\rho)^2$. By the Cauchy-Schwarz inequality, we have
\begin{eqnarray*}
d_S^2&=&\frac{1}{|\Lambda(S)|^2}\big(\sum_{\mathbf{l}\in\Lambda(S)}(\rho_S(\mathbf{l})-\rho)(\rho'_S(\mathbf{l})-\rho')\big)^2\\
&\leq&\frac{1}{|\Lambda(S)|^2}\big(\sum_{\mathbf{l}\in\Lambda(S)}d_{\rho}^2\big)\big(\sum_{\mathbf{l}\in\Lambda(S)}d_{\rho'}^2\big)\\
&=&\big(\frac{1}{|\Lambda(S)|}\sum_{\mathbf{l}\in\Lambda(S)}d_{\rho}^2\big)\big(\frac{1}{|\Lambda(S)|}\sum_{\mathbf{l}\in\Lambda(S)}d_{\rho'}^2\big)\\
&=&\sigma_S^2\cdot(\sigma'_S)^2,
\end{eqnarray*}
The lemma then follows immediately.
\end{proof}

\noindent We can now prove the first inequality in Theorem~\ref{theorem:tighter-bound}:
\begin{proof}
Let $d_n=\Cov(\rho_n,\rho'_n)$. By Lemma~\ref{lemma:dn} and~\ref{lemma:ds} we have
\begin{eqnarray*}
|d_n|&=&|\sum_{r=1}^m\big(1-\frac{1}{n}\big)^{m-r}\big(\frac{1}{n}\big)^r\sum_{S\in\mathcal{S}_r}\Cov_S(\rho,\rho')|\\
&\leq&\sum_{r=1}^m\big(1-\frac{1}{n}\big)^{m-r}\big(\frac{1}{n}\big)^r\sum_{S\in\mathcal{S}_r}\sqrt{\sigma_S^2(\sigma'_S)^2}.
\end{eqnarray*}
By the Cauchy-Schwarz inequality, we have
$$\sum_{S\in\mathcal{S}_r}\sqrt{\sigma_S^2(\sigma'_S)^2}\leq\sqrt{\sum_{S\in\mathcal{S}_r}\sigma_S^2\sum_{S\in\mathcal{S}_r}(\sigma'_S)^2}.$$
Combining these two inequalities, we obtain
$$|d_n|\leq\sum_{r=1}^m\big(1-\frac{1}{n}\big)^{m-r}\big(\frac{1}{n}\big)^r\sqrt{\sum_{S\in\mathcal{S}_r}\sigma_S^2\sum_{S\in\mathcal{S}_r}(\sigma'_S)^2}.$$
Now define
$$A_r=\sqrt{\big(1-\frac{1}{n}\big)^{m-r}\big(\frac{1}{n}\big)^r\sum_{S\in\mathcal{S}_r}\sigma_S^2}.$$
Then $|d_n|\leq \sum_{r=1}^m A_rA'_r$. Applying the Cauchy-Schwarz inequality again,
$$|d_n|\leq \sqrt{\big(\sum_{r=1}^m A_r^2\big)\big(\sum_{r=1}^m (A')_r^2\big)}=\sqrt{S_{\rho}^2(m,n)S_{\rho'}^2(m,n)},$$
which completes the proof of the inequality.
\end{proof}

To establish the second inequality in the theorem, namely,
$$\sqrt{S_{\rho}^2(m,n)S_{\rho'}^2(m,n)}\leq \sqrt{\Var[\rho_n]\Var[\rho'_n]},$$
we need two more lemmas. The first one states that the $\sigma_S^2$ has some nice \emph{monotonicity} property:

\begin{lemma}\label{lemma:sigma}
Let $S\in\mathcal{S}_r$ and $S'\in\mathcal{S}_{r+1}$ such that $S\subset S'$, for $1\leq r\leq K-1$. Then $\sigma_S^2\leq\sigma_{S'}^2$.
\end{lemma}

\begin{proof}
Without loss of generality, let $S=\{1,...,r\}$ and $S'=\{1,...,r+1\}$. For a given $\mathbf{l}=(j_1,...,j_r)\in\Lambda(S)$, let $\mathbf{l}'_j=(j_1,...,j_r,j)$, for $1\leq j\leq m_{r+1}$. Since $\Lambda(S')=\Lambda(S)\times\{1,...,m_{r+1}\}$, we have $\Lambda(S^c)=\Lambda((S')^c)\times\{1,...,m_{r+1}\}$ and thus $|\Lambda(S^c)|=m_{r+1}|\Lambda((S')^c)|$. Therefore, by letting $d_{\rho}^2=(\rho_S(\mathbf{l})-\rho)^2$, it follows that
\begin{eqnarray*}
d_{\rho}^2&=&\big((\frac{1}{|\Lambda(S^c)|}\sum_{\Lambda(S^c)}\rho_{\mathbf{B}})-\rho\big)^2\\
&=&\big(\frac{1}{|\Lambda(S^c)|}\sum_{\Lambda(S^c)}(\rho_{\mathbf{B}}-\rho)\big)^2\\
&=&\big(\frac{1}{m_{r+1}|\Lambda((S')^c)|}\sum_{j=1}^{m_{r+1}}\sum_{\Lambda((S')^c)}(\rho_{\mathbf{B}}-\rho)\big)^2\\
&=&\frac{1}{m_{r+1}^2}\big(\sum_{j=1}^{m_{r+1}}\frac{1}{|\Lambda((S')^c)|}\sum_{\Lambda((S')^c)}(\rho_{\mathbf{B}}-\rho)\big)^2.
\end{eqnarray*}
By the Cauchy-Schwarz inequality, we have
\begin{eqnarray*}
d_{\rho}^2&=&\frac{1}{m_{r+1}}\sum_{j=1}^{m_{r+1}}\big(\frac{1}{|\Lambda((S')^c)|}\sum_{\Lambda((S')^c)}(\rho_{\mathbf{B}}-\rho)\big)^2\\
&=&\frac{1}{m_{r+1}}\sum_{j=1}^{m_{r+1}}\big((\frac{1}{|\Lambda((S')^c)|}\sum_{\Lambda((S')^c)}\rho_{\mathbf{B}})-\rho\big)^2\\
&=&\frac{1}{m_{r+1}}\sum_{j=1}^{m_{r+1}}(\rho_{S'}(\mathbf{l}'_j)-\rho)^2.
\end{eqnarray*}
Therefore,
\begin{eqnarray*}
\sigma_S^2&=&\frac{1}{|\Lambda(S)|}\sum_{\Lambda(S)}(\rho_S(\mathbf{l})-\rho)^2\\
&\leq &\frac{1}{|\Lambda(S)|m_{r+1}}\sum_{\Lambda(S)}\sum_{j=1}^{m_{r+1}}(\rho_{S'}(\mathbf{l}'_j)-\rho)^2\\
&=&\frac{1}{|\Lambda(S')|}\sum_{\Lambda(S')}(\rho_{S'}(\mathbf{l}')-\rho)^2\\
&=&\sigma_{S'}^2.
\end{eqnarray*}
This completes the proof of the lemma.
\end{proof}

\noindent Our next lemma further shows that the $S_{\rho}^2(m,n)$ also has some similar monotonicity property:

\begin{lemma}\label{lemma:monotone}
For $m\geq 1$, we have
$$S_{\rho}^2(m,n)\leq S_{\rho}^2(m+1, n).$$
\end{lemma}
\begin{proof}
We should be careful now since $\mathcal{S}_r$ is actually related to $m$. Specifically, $\mathcal{S}_r$ is all the $r$-subsets of $\{1,...,m\}$.\footnote{More generally, the indexes could be represented as $\mathcal{J}_m=\{j_1,...,j_m\}$ and $\mathcal{J}_{m+1}=\{j_1,...,j_m,j_{m+1}\}$ such that $\mathcal{J}_m\subset\mathcal{J}_{m+1}$. We used $\mathcal{J}_m=\{1,...,m\}$ and $\mathcal{J}_{m+1}=\{1,...,m,m+1\}$ in our proof without loss of generality.} To make this more explicit, we further use $\mathcal{S}_r^{(m)}$ to indicate this relationship. Moreover, to simplify notation, we define
$$A_r^{(m)}=\sum_{S\in\mathcal{S}_r^{(m)}}\sigma_S^2.$$
Furthermore, if $r=m$, then $\mathcal{S}_m^{(m)}$ contains only one single element $\{1,...,m\}$. We thus simply use $\sigma_m^2$ to represent $A_m^{(m)}$, i.e.,
$$\sigma_m^2=\sum_{S\in\mathcal{S}_m^{(m)}}\sigma_S^2.$$
Now consider $S_{m+1}=S_{\rho}^2(m+1,n)$. We have
\begin{eqnarray*}
S_{m+1}&=&\sum_{r=1}^{m+1}\big(1-\frac{1}{n}\big)^{m+1-r}\big(\frac{1}{n}\big)^r A_r^{(m+1)}\\
&=&\sum_{r=1}^{m}\big(1-\frac{1}{n}\big)^{m+1-r}\big(\frac{1}{n}\big)^r A_r^{(m+1)}+\big(\frac{1}{n}\big)^{m+1}\sigma_{m+1}^2.
\end{eqnarray*}
Define $\Delta_r^{(m+1)}=\sum_{S\in\mathcal{S}_r^{(m+1)}\setminus\mathcal{S}_r^{(m)}}\sigma_S^2$. Then
$$\Delta_r^{(m+1)}=A_r^{(m+1)}-A_r^{(m)}.$$
We therefore have
$S_{m+1}=\big(1-\frac{1}{n}\big)S_m + B_m$, where
$$B_m=\sum_{r=1}^{m}\big(1-\frac{1}{n}\big)^{m+1-r}\big(\frac{1}{n}\big)^r\Delta_r^{(m+1)}+\big(\frac{1}{n}\big)^{m+1}\sigma_{m+1}^2.$$
Let us further define $\mathcal{S}_r^{(m)}=\emptyset$ if $r > m$. Then $\Delta_{m+1}^{(m+1)}=\sigma_{m+1}^2$, and therefore
$$B_m=\sum_{r=1}^{m+1}\big(1-\frac{1}{n}\big)^{m+1-r}\big(\frac{1}{n}\big)^r\Delta_r^{(m+1)}.$$

Next, consider some $S\in\mathcal{S}_r^{(m+1)}\setminus\mathcal{S}_r^{(m)}$ where $r\geq 2$. Note that $S$ must contain $m+1$ since otherwise $S\in\mathcal{S}_r^{(m)}$. What's more, if we remove $m+1$ from $S$, then $S$ must be now in $\mathcal{S}_{r-1}^{(m)}$, that is, $S\setminus\{m+1\}\in\mathcal{S}_{r-1}^{(m)}$. On the other hand, for any $S'\in \mathcal{S}_{r-1}^{(m)}$, we can obtain an element in $\mathcal{S}_r^{(m+1)}\setminus\mathcal{S}_r^{(m)}$ by simply adding $m+1$, that is, $S'\cup\{m+1\}\in\mathcal{S}_r^{(m+1)}\setminus\mathcal{S}_r^{(m)}$. We therefore have established a 1-1 mapping $\varphi$ between $\mathcal{S}_r^{(m+1)}\setminus\mathcal{S}_r^{(m)}$ and $\mathcal{S}_{r-1}^{(m)}$.

Furthermore, note that for any $S\in\mathcal{S}_r^{(m+1)}\setminus\mathcal{S}_r^{(m)}$, we have $\varphi(S)\subset S$. Hence by Lemma~\ref{lemma:sigma}, $\sigma_{\varphi(S)}^2\leq\sigma_S^2$. Therefore, we have
$$\Delta_r^{(m+1)}=\sum_{S\in\mathcal{S}_r^{(m+1)}\setminus\mathcal{S}_r^{(m)}}\sigma_S^2\geq\sum_{\varphi(S)\in\mathcal{S}_{r-1}^{(m)}}\sigma_{\varphi(S)}^2=A_{r-1}^{(m)}.$$
As a result, we have
\begin{eqnarray*}
B_m&\geq&C_m+\sum_{r=2}^{m+1}\big(1-\frac{1}{n}\big)^{m+1-r}\big(\frac{1}{n}\big)^r A_{r-1}^{(m)}\\
&=&C_m+\sum_{r'=1}^{m}\big(1-\frac{1}{n}\big)^{m+1-(r'+1)}\big(\frac{1}{n}\big)^{r'+1} A_{r'}^{(m)}\\
&=&C_m+\frac{1}{n}\sum_{r'=1}^{m}\big(1-\frac{1}{n}\big)^{m-r'}\big(\frac{1}{n}\big)^{r'} A_{r'}^{(m)}\\
&=&C_m+\frac{1}{n}S_m,
\end{eqnarray*}
where
$$C_m=\big(1-\frac{1}{n}\big)^{m}\frac{1}{n}\Delta_1^{(m+1)}=\frac{1}{n}\big(1-\frac{1}{n}\big)^{m}\sigma_{\{m+1\}}^2\geq 0.$$
Hence, $B_m\geq\frac{1}{n}S_m$. Since $S_{m+1}=\big(1-\frac{1}{n}\big)S_m + B_m$, we conclude that $S_{m+1}\geq S_m$. This completes the proof of the lemma.
\end{proof}

\noindent It is now easy to prove the second inequality in Theorem~\ref{theorem:tighter-bound}:

\begin{proof}
Based on Lemma~\ref{lemma:monotone}, by induction, we can easily prove that $S_{\rho}^2(m,n)\leq \Var[\rho_n]$ and $S_{\rho'}^2(m,n)\leq \Var[\rho'_n]$, since $m\leq\min\{K, K'\}$. The inequality then follows.
\end{proof}

For our special case in this paper where $\Cov(\rho_n,\rho'_n)\neq 0$, we will always have $m=\min\{K,K'\}$. Without loss of generality, let $m=K$. Then $S_{\rho}^2(m,n)= \Var[\rho_n]$, and we only need to approximate $S_{\rho'}^2(m,n)$ with $S_{\rho'}^2(K,n)$, which by Lemma~\ref{lemma:monotone} is guaranteed to be superior to $\Var[\rho'_n]$. Intuitively, the bigger $K'-K$ is, the bigger the gap is between $S_{\rho'}^2(K,n)$ and $\Var[\rho'_n]$. In fact, in the proof of Lemma~\ref{lemma:monotone}, we have actually showed that $S_{m+1}\geq S_m+C_m$. So we can roughly estimate that
$$\Var[\rho'_n]-S_{\rho'}^2(K,n)\geq\frac{1}{n}\big(1-\frac{1}{n}\big)^K\sum_{r=K+1}^{K'}\sigma_{\{r\}}^2.$$

\subsection{More Bounds for Covariances}\label{sec:more-cov-bounds}

We can actually have another upper bound for $\Cov(\rho_n,\rho'_n)$:
\begin{theorem}\label{theorem:dn-bound}
We have
$$|\Cov(\rho_n,\rho'_n)|\leq f(n,m)g(\rho)g(\rho'),$$
where $f(n,m)=1-(1-\frac{1}{n})^m$ and $g(\rho)=\sqrt{\rho(1-\rho)}$.
\end{theorem}

\begin{proof}
As in the proof of Lemma~\ref{lemma:dn}, let $\mathcal{K}$ and $\mathcal{K}'$ be the indexes of the relations in $\mathcal{R}$ and $\mathcal{R}'$ respectively. By Lemma~\ref{lemma:sigma}, we have $\sigma_S^2\leq\sigma_{\mathcal{K}}^2$ and $(\sigma'_S)^2\leq\sigma_{\mathcal{K}'}^2$. Moreover, consider
$$\sigma_{\mathcal{K}}^2=\frac{1}{|\Lambda(\mathcal{K})|}\sum_{\mathbf{l}\in\Lambda(\mathcal{K})}(\rho_{\mathcal{K}}(\mathbf{l})-\rho)^2.$$
Since we use the tuple-level partition scheme, we have $\rho_{\mathcal{K}}(\mathbf{l})=1$ or $\rho_{\mathcal{K}}(\mathbf{l})=0$. Therefore,
\begin{eqnarray*}
\sigma_{\mathcal{K}}^2&=&(1-\rho)^2\cdot\frac{1}{|\Lambda(\mathcal{K})|}\sum_{\mathbf{l}\in\Lambda(\mathcal{K})}I(\rho_{\mathcal{K}}(\mathbf{l})=1)\\
&\quad&+\rho^2\cdot\frac{1}{|\Lambda(\mathcal{K})|}\sum_{\mathbf{l}\in\Lambda(\mathcal{K})}I(\rho_{\mathcal{K}}(\mathbf{l})=0)\\
&=&(1-\rho)^2\cdot\rho + \rho^2\cdot(1-\rho)\\
&=&\rho(1-\rho).
\end{eqnarray*}
Similarly, we have $\sigma_{\mathcal{K}'}^2=\rho'(1-\rho')$. Hence,
$$\Cov_S(\rho,\rho')\leq\sqrt{\sigma_S^2(\sigma'_S)^2}\leq\sqrt{\rho(1-\rho)\cdot\rho'(1-\rho')},$$
and therefore, by letting $g(\rho)=\sqrt{\rho(1-\rho)}$,
\begin{eqnarray*}
|d_n|&=&|\sum_{r=1}^m\frac{(n-1)^{m-r}}{n^m}\times\sum_{S\in\mathcal{S}_r}\Cov_S(\rho,\rho')|\\
&\leq&\sum_{r=1}^m\frac{(n-1)^{m-r}}{n^m}\times\sum_{S\in\mathcal{S}_r}g(\rho)g(\rho')\\
&=&g(\rho)g(\rho')\sum_{r=1}^m{m \choose r}\big(\frac{1}{n}\big)^r\big(1-\frac{1}{n}\big)^{m-r}\\
&=&g(\rho)g(\rho')[1-(1-\frac{1}{n})^m].
\end{eqnarray*}
This completes the proof of the theorem.
\end{proof}

When $n$ is large, $(1-\frac{1}{n})^m\approx 1-\frac{m}{n}$. As a result, $1-(1-\frac{1}{n})^m\approx\frac{m}{n}$. Therefore, when $n\to\infty$, $\Cov(\rho_n,\rho'_n)\to 0$. This is intuitively true considering the strong consistency of $\rho_n$. If we keep taking samples, finally the estimated selectivity should converge to the actual selectivity (a constant). On the other hand, a larger $m$ implies a larger bound since the computations of $\rho_n$ and $\rho'_n$ share more samples. Another interesting observation is that the bound also depends on the actual selectivities $\rho$ and $\rho'$. Note that $g(\rho)$ is minimized at $\rho=0$ or $\rho=1$ (with $g_{\min}=0$), and is maximized at $\rho=\frac{1}{2}$ (with $g_{\max}=\frac{1}{2}$). To shed some light on this, observe that whenever $\rho$ or $\rho'$ is 0 or 1, $\rho_n$ or $\rho'_n$ is always 0 or 1 regardless of the number of samples. Hence $\Cov(\rho_n,\rho'_n)=0$ in such cases.

An natural question is how good this bound is compared with the two bounds in Section~\ref{sec:uncertainty:covar:bounds}. Let us name these three bounds as
\begin{enumerate}[($B_1$)]
\item $\sqrt{S_{\rho}^2(m,n)S_{\rho'}^2(m,n)}$, the first bound in Theorem~\ref{theorem:tighter-bound};
\item $\sqrt{\Var[\rho_n]\Var[\rho'_n]}$, the second bound in Theorem~\ref{theorem:tighter-bound};
\item $f(n,m)g(\rho)g(\rho')$, the bound in Theorem~\ref{theorem:dn-bound}.
\end{enumerate}
By Theorem~\ref{theorem:tighter-bound}, we already know that $B_1\leq B_2$. Next, according to the proof of Theorem~\ref{theorem:dn-bound}, $\sigma_S^2\leq \rho(1-\rho)$ and $(\sigma')_S^2\leq \rho'(1-\rho')$. We then immediately have
$$\sqrt{S_{\rho}^2(m,n)S_{\rho'}^2(m,n)}\leq f(n,m)g(\rho)g(\rho'),$$
by the definition of $S_{\rho}^2(m,n)$. That is, $B_1\leq B_3$. Moreover, by Theorem~\ref{theorem:varbound}, we have
$$|\Cov(\rho_n,\rho'_n)|\leq\sqrt{\Var[\rho_n]\Var[\rho'_n]}\leq f(n)g(\rho)g(\rho'),$$
where $$f(n)=\sqrt{\big(1-(1-\frac{1}{n})^K\big)\big(1-(1-\frac{1}{n})^{K'}\big)}.$$
When $n$ is large, $1-(1-\frac{1}{n})^K\approx\frac{K}{n}$, and $1-(1-\frac{1}{n})^{K'}\approx\frac{K'}{n}$. Therefore, the right hand is close to $\frac{\sqrt{KK'}}{n}g(\rho)g(\rho')$.
Since $m\leq \min\{K, K'\} <\sqrt{KK'}$, we know that $B_3$ is better than the upper bound of $B_2$. However, in general $B2$ and $B_3$ are incomparable.

One more issue of $B_3$ is that it includes the \emph{true} selectivities $\rho$ and $\rho'$ that are not known without running the query. As a result, $B_3$ is not directly computable. Nonetheless, when $n$ is large, we can simply use the observed $\rho_n$ and $\rho'_n$ as approximations due to the strong consistency of $\rho_n$.

Finally, the techniques we used in the proof of Theorem~\ref{theorem:dn-bound} can be further generalized to establish similar bounds for other covariances such as $\Cov(\rho_n^2, (\rho')_n^2)$ and $\Cov(\rho_n^2, \rho'_n$).

\begin{theorem}\label{theorem:dn22-bound}
$$|\Cov(\rho_n^2, (\rho')_n^2)|\leq f(n, m)h(\rho)h(\rho'),$$
where
\begin{eqnarray*}
f(n, m)&=&[1 - (1 - \frac{1}{n})^{K + K' - m}(1 - \frac{2}{n})^m(1 - \frac{3}{n})^m]\\
&\cdot &\sqrt{1 - (1 - \frac{1}{n})^K} \sqrt{1 - (1- \frac{1}{n})^{K'}},
\end{eqnarray*}
and
$$h(\rho) = \sqrt{\rho(1 - \rho)(\rho - \rho^2 + 1)}.$$
\end{theorem}
When $n$ is large, we can approximate $f(n, m)$ as:
\begin{eqnarray*}
f(n,m)&\approx&\big(1-(1-\frac{K+K'-m}{n})(1-\frac{2m}{n})(1-\frac{3m}{n})\big)\\
&\quad&\cdot\big(1-(1-\frac{K}{n})\big)^{\frac{1}{2}}\big(1-(1-\frac{K'}{n})\big)^{\frac{1}{2}}\\
&\approx&\big(1-(1-\frac{K+K'-m}{n}-\frac{2m}{n}-\frac{3m}{n})\big)\frac{\sqrt{KK'}}{n}\\
&=&\frac{(K+K'+4m)\sqrt{KK'}}{n^2}.
\end{eqnarray*}

\begin{proof}
For notational convenience, define $E_{\rho\rho'}=\E[\rho_n^2(\rho'_n)^2]$, and $E_{\rho}E_{\rho'}=\E[\rho_n^2]\E[(\rho'_n)^2]$. We have
\begin{eqnarray*}
E_{\rho\rho'}&=&\E[\big(\frac{1}{n^K}\sum_{k=1}^{n^K}\rho_{\mathbf{B}}\big)^2\big(\frac{1}{n^{K'}}\sum_{k'=1}^{ n^{K'}}\rho'_{\mathbf{B}}\big)^2]\\
&=&\frac{1}{n^{2(K+K')}}\sum_{k_1,k_2=1}^{n^K}\sum_{k'_1,k'_2=1}^{n^{K'}}\E[\rho_{\mathbf{B}_1}
\rho_{\mathbf{B}_2}\rho'_{\mathbf{B}_1}\rho'_{\mathbf{B}_2}],
\end{eqnarray*}
and
\begin{eqnarray*}
E_{\rho}E_{\rho'}&=&\E[\big(\frac{1}{n^K}\sum_{k=1}^{n^K}\rho_{\mathbf{B}}\big)^2]
\E[\big(\frac{1}{n^{K'}}\sum_{k'=1}^{n^{K'}}\rho'_{\mathbf{B}}\big)^2]\\
&=&\frac{1}{n^{2(K+K')}}\sum_{k_1,k_2=1}^{n^K}\sum_{k'_1,k'_2=1}^{n^{K'}}\E[\rho_{\mathbf{B}_1}
\rho_{\mathbf{B}_2}]\E[\rho'_{\mathbf{B}_1}\rho'_{\mathbf{B}_2}].
\end{eqnarray*}
Therefore, by letting $d_n^2=\Cov(\rho_n^2,(\rho')_n^2)$, we have
\begin{eqnarray*}
d_n^2&=&\E[\rho_n^2(\rho'_n)^2]-\E[\rho_n^2]\E[(\rho'_n)^2]\\
&=&\frac{1}{n^{2(K+K')}}\sum_{k_1,k_2=1}^{n^K}\sum_{k'_1,k'_2=1}^{n^{K'}}\Cov(\rho_{\mathbf{B}_1}
\rho_{\mathbf{B}_2},\rho'_{\mathbf{B}_1}\rho'_{\mathbf{B}_2}).
\end{eqnarray*}
Note that, among the $n^{2(K+K')}$ summands in $d_n^2$, $n^K(n-1)^K(n-2)^mn^{K'-m}(n-3)^m(n-1)^{K'-m}=[n(n-1)]^{K+K'-m}[(n-2)(n-3)]^m$ of them involve samples that do not share any blocks and are therefore independent. As a result, for these terms,
$$\Cov(\rho_{\mathbf{B}_1}\rho_{\mathbf{B}_2},\rho'_{\mathbf{B}_1}\rho'_{\mathbf{B}_2})=0.$$
Since
$$\Cov^2(\rho_{\mathbf{B}_1}\rho_{\mathbf{B}_2},\rho'_{\mathbf{B}_1}\rho'_{\mathbf{B}_2})\leq
\Var[\rho_{\mathbf{B}_1}\rho_{\mathbf{B}_2}]\Var[\rho'_{\mathbf{B}_1}\rho'_{\mathbf{B}_2}],$$
we have
\begin{eqnarray*}
|d_n^2|&\leq& \big(1-\frac{[n(n-1)]^{K+K'-m}[(n-2)(n-3)]^m}{n^{2(K+K')}}\big)\\
&\quad&\Var[\rho_{\mathbf{B}_1}\rho_{\mathbf{B}_2}]^{\frac{1}{2}}\Var[\rho'_{\mathbf{B}_1}\rho'_{\mathbf{B}_2}]^{\frac{1}{2}}\\
&=&\big(1-(1-\frac{1}{n})^{K+K'-m}(1-\frac{2}{n})^m(1-\frac{3}{n})^m\big)\\
&\quad&\Var[\rho_{\mathbf{B}_1}\rho_{\mathbf{B}_2}]^{\frac{1}{2}}\Var[\rho'_{\mathbf{B}_1}\rho'_{\mathbf{B}_2}]^{\frac{1}{2}}.
\end{eqnarray*}

Next, let us consider
$$V=\Var[\rho_{\mathbf{B}_1}\rho_{\mathbf{B}_2}]=\E[(\rho_{\mathbf{B}_1}\rho_{\mathbf{B}_2}
-\E[\rho_{\mathbf{B}_1}\rho_{\mathbf{B}_2}])^2].$$
As before, we partition the pairs $(\mathbf{B}_1, \mathbf{B}_2)$ based on the blocks they share, which gives us
$$V=\sum_{r=1}^{K}\frac{(n-1)^{K-r}}{n^K}\times\sum_{S\in\mathcal{S}_r}
\big(\frac{1}{|\Lambda(S)|}\sum_{\mathbf{l}\in\Lambda(S)}\E_{S(\mathbf{l})}(\mathbf{B}_1,\mathbf{B}_2)\big),$$
where
$$\E_{S(\mathbf{l})}(\mathbf{B}_1,\mathbf{B}_2)=\E_{\rho_{\mathbf{B}_1}\sim_{S(\mathbf{l})}\rho_{\mathbf{B}_2}}
[(\rho_{\mathbf{B}_1}\rho_{\mathbf{B}_2}-\E[\rho_{\mathbf{B}_1}\rho_{\mathbf{B}_2}])^2].$$
Now define
$$\sigma_S^2=\frac{1}{|\Lambda(S)|}\sum_{\mathbf{l}\in\Lambda(S)}\E_{S(\mathbf{l})}(\mathbf{B}_1,\mathbf{B}_2).$$
Similarly as before, we are able to show that, if $S\in\mathcal{S}_r$, $S'\in\mathcal{S}_{r+1}$, and $S\subset S'$, then $\sigma_S^2\leq\sigma_{S'}^2$. To see this, without loss of generality, let $S=\{1,...,r\}$ and $S'=\{1,...,r+1\}$. For a given $\mathbf{l}=(j_1,...,j_r)\in\Lambda(S)$, let $\mathbf{l}'_j=(j_1,...,j_r,j)$, for $1\leq j\leq m_{r+1}$. Since $\Lambda(S')=\Lambda(S)\times\{1,...,m_{r+1}\}$, $\Lambda(S^c)=\Lambda((S')^c)\times\{1,...,m_{r+1}\}$ and thus $|\Lambda(S^c)|=m_{r+1}|\Lambda((S')^c)|$.

We have
$$E=\frac{1}{|\Lambda(S^c)|^2}\sum_{\Lambda(S^c)}\sum_{\Lambda(S^c)}
(\rho_{\mathbf{B}_1}\rho_{\mathbf{B}_2}-\E[\rho_{\mathbf{B}_1}\rho_{\mathbf{B}_2}])^2.$$
Let $\E[\rho_{\mathbf{B}_1}\rho_{\mathbf{B}_2}]=A$. We then have
$$E=\frac{1}{|\Lambda(S^c)|^2}\sum_{\Lambda(S^c)}\sum_{\Lambda(S^c)}(\rho_{\mathbf{B}_1}^2\rho_{\mathbf{B}_2}^2
-2A\rho_{\mathbf{B}_1}\rho_{\mathbf{B}_2}+A^2).$$
Consider the terms $\rho_{\mathbf{B}_1}^2\rho_{\mathbf{B}_2}^2$ and $\rho_{\mathbf{B}_1}\rho_{\mathbf{B}_2}$. We have
\begin{eqnarray*}
E_1&=&\frac{1}{|\Lambda(S^c)|^2}\sum_{\Lambda(S^c)}\sum_{\Lambda(S^c)}\rho_{\mathbf{B}_1}^2\rho_{\mathbf{B}_2}^2\\
&=&\big(\frac{1}{|\Lambda(S^c)|}\sum_{\Lambda(S^c)}\rho_{\mathbf{B}_1}^2\big)
\cdot\big(\frac{1}{|\Lambda(S^c)|}\sum_{\Lambda(S^c)}\rho_{\mathbf{B}_2}^2\big)\\
&=&\big(\frac{1}{|\Lambda(S^c)|}\sum_{\Lambda(S^c)}\rho_{\mathbf{B}}^2\big)^2\\
&=&\big(\frac{1}{|\Lambda((S')^c)|m_{r+1}}\sum_{\Lambda((S')^c)}\sum_{j=1}^{m_{r+1}}\rho_{\mathbf{B}}^2\big)^2,
\end{eqnarray*}
and
\begin{eqnarray*}
E_2&=&\frac{1}{|\Lambda(S^c)|^2}\sum_{\Lambda(S^c)}\sum_{\Lambda(S^c)}\rho_{\mathbf{B}_1}\rho_{\mathbf{B}_2}\\
&=&\big(\frac{1}{|\Lambda(S^c)|}\sum_{\Lambda(S^c)}\rho_{\mathbf{B}_1}\big)
\cdot\big(\frac{1}{|\Lambda(S^c)|}\sum_{\Lambda(S^c)}\rho_{\mathbf{B}_2}\big)\\
&=&\big(\frac{1}{|\Lambda(S^c)|}\sum_{\Lambda(S^c)}\rho_{\mathbf{B}}\big)^2\\
&=&\big(\frac{1}{|\Lambda((S')^c)|m_{r+1}}\sum_{\Lambda((S')^c)}\sum_{j=1}^{m_{r+1}}\rho_{\mathbf{B}}\big)^2.
\end{eqnarray*}
Since $\rho_{\mathbf{B}}=0$ or $\rho_{\mathbf{B}}=1$, we have $\rho_{\mathbf{B}}^2=\rho_{\mathbf{B}}$ and thus $E_1=E_2$. Therefore,
$$E=(1-2A)E_2+A^2.$$
Furthermore, define
$$\overline{\rho}_{S'}=\frac{1}{|\Lambda((S')^c)|}\sum_{\Lambda((S')^c)}\rho_{\mathbf{B}},\quad \overline{\rho^2}_{S'}=\frac{1}{|\Lambda((S')^c)|}\sum_{\Lambda((S')^c)}\rho_{\mathbf{B}}^2,$$
$$\overline{\rho_1\rho_2}_{S'}=\frac{1}{|\Lambda((S')^c)|^2}\sum_{\Lambda((S')^c)}\sum_{\Lambda((S')^c)}\rho_{\mathbf{B}_1}\rho_{\mathbf{B}_2},$$
$$\overline{\rho_1^2\rho_2^2}_{S'}=\frac{1}{|\Lambda((S')^c)|^2}\sum_{\Lambda((S')^c)}
\sum_{\Lambda((S')^c)}\rho_{\mathbf{B}_1}^2\rho_{\mathbf{B}_2}^2,$$
$$\overline{(\rho_1\rho_2-A)^2}_{S'}=\frac{1}{|\Lambda((S')^c)|^2}\sum_{\Lambda((S')^c)}
\sum_{\Lambda((S')^c)}\big(\rho_{\mathbf{B}_1}\rho_{\mathbf{B}_2}-A\big)^2.$$
By the Cauchy-Schwarz inequality, we have
$$E_2=\frac{1}{m_{r+1}^2}\big(\sum_{j=1}^{m_{r+1}}\overline{\rho}_{S'}\big)^2\leq\frac{1}{m_{r+1}}\sum_{j=1}^{m_{r+1}}\overline{\rho}_{S'}^2.$$
Hence,
\begin{eqnarray*}
E&\leq&(1-2A)\frac{1}{m_{r+1}}\sum_{j=1}^{m_{r+1}}\overline{\rho}_{S'}^2+A^2\\
&=&\frac{1}{m_{r+1}}\sum_{j=1}^{m_{r+1}}\big((1-2A)\overline{\rho}_{S'}^2+A^2\big)\\
&=&\frac{1}{m_{r+1}}\sum_{j=1}^{m_{r+1}}[\overline{\rho}_{S'}^2-2A\overline{\rho}_{S'}^2+A^2]\\
&=&\frac{1}{m_{r+1}}\sum_{j=1}^{m_{r+1}}[(\overline{\rho^2}_{S'})^2-2A\overline{\rho}_{S'}^2+A^2]\\
&=&\frac{1}{m_{r+1}}\sum_{j=1}^{m_{r+1}}[\overline{\rho_1^2\rho_2^2}_{S'}-2A\overline{\rho_1\rho_2}_{S'}+A^2]\\
&=&\frac{1}{m_{r+1}}\sum_{j=1}^{m_{r+1}}\overline{(\rho_1\rho_2-A)^2}_{S'}.
\end{eqnarray*}
Therefore,
\begin{eqnarray*}
\sigma_S^2&=&\frac{1}{|\Lambda(S)|}\sum_{\mathbf{l}\in\Lambda(S)}\E_{S(\mathbf{l})}(\mathbf{B}_1,\mathbf{B}_2)\\
&\leq&\frac{1}{|\Lambda(S)|}\sum_{\mathbf{l}\in\Lambda(S)}\frac{1}{m_{r+1}}\sum_{j=1}^{m_{r+1}}\overline{(\rho_1\rho_2-A)^2}_{S'}\\
&=&\frac{1}{|\Lambda(S)|m_{r+1}}\sum_{\mathbf{l}\in\Lambda(S)}\sum_{j=1}^{m_{r+1}}\overline{(\rho_1\rho_2-A)^2}_{S'}\\
&=&\frac{1}{|\Lambda(S')|}\sum_{\mathbf{l}'\in\Lambda(S')}\overline{(\rho_1\rho_2-A)^2}_{S'}\\
&=&\sigma_{S'}^2.
\end{eqnarray*}
As a result, we have $\sigma_S^2\leq\sigma_{\mathcal{K}}^2$. Since
\begin{eqnarray*}
\sigma_{\mathcal{K}}^2&=&\frac{1}{|\Lambda(\mathcal{K})|}\sum_{\mathbf{l}\in\Lambda(\mathcal{K})}\E_{\mathcal{K}(\mathbf{l})}(\mathbf{B}_1,\mathbf{B}_2)\\
&=&\frac{1}{|\Lambda(\mathcal{K})|}\sum_{\mathbf{l}\in\Lambda(\mathcal{K})}(\rho_{\mathbf{B}_1}\rho_{\mathbf{B}_2}-A)^2\\
&=&\frac{1}{|\Lambda(\mathcal{K})|}\sum_{\mathbf{l}\in\Lambda(\mathcal{K})}(\rho_{\mathbf{B}}^2-A)^2\\
&=&\frac{1}{|\Lambda(\mathcal{K})|}\sum_{\mathbf{l}\in\Lambda(\mathcal{K})}(\rho_{\mathbf{B}}-A)^2\\
&=&(1-A)^2\cdot\frac{1}{|\Lambda(\mathcal{K})|}\sum_{\mathbf{l}\in\Lambda(\mathcal{K})}I(\rho_{\mathbf{B}}=1)\\
&\quad&+A^2\cdot\frac{1}{|\Lambda(\mathcal{K})|}\sum_{\mathbf{l}\in\Lambda(\mathcal{K})}I(\rho_{\mathbf{B}}=0)\\
&=&(1-A)^2\cdot\rho + A^2\cdot(1-\rho)\\
&=&A^2-2A\rho+\rho\\
&=&(A-\rho)^2+(\rho-\rho^2).
\end{eqnarray*}
Now consider $A=\E[\rho_{\mathbf{B}_1}\rho_{\mathbf{B}_2}]$. We have
\begin{eqnarray*}
A&=&\Pr(\rho_{\mathbf{B}_1}=1,\rho_{\mathbf{B}_2}=1)\\
&\geq&\Pr(\rho_{\mathbf{B}_1}=1,\rho_{\mathbf{B}_2}=1, \rho_{\mathbf{B}_1}\sim_{\emptyset}\rho_{\mathbf{B}_2})\\
&=&\Pr(\rho_{\mathbf{B}_1}=1)\cdot \Pr(\rho_{\mathbf{B}_2}=1)\\
&=&\rho^2.
\end{eqnarray*}
On the other hand,
$$A=\Pr(\rho_{\mathbf{B}_1}=1,\rho_{\mathbf{B}_2}=1)\leq \Pr(\rho_{\mathbf{B}_1}=1)\leq\rho.$$
Thus, $\rho^2\leq A\leq\rho$. So we have
$$\rho-\rho^2\leq\sigma_{\mathcal{K}}^2\leq(\rho^2-\rho)^2+(\rho-\rho^2).$$
Therefore,
\begin{eqnarray*}
V&=&\Var[\rho_{\mathbf{B}_1}\rho_{\mathbf{B}_2}]\\
&=&\sum_{r=1}^{K}\frac{(n-1)^{K-r}}{n^K}\times\sum_{S\in\mathcal{S}_r}\sigma_S^2\\
&\leq&\sum_{r=1}^{K}\frac{(n-1)^{K-r}}{n^K}\times\sum_{S\in\mathcal{S}_r}\sigma_{\mathcal{K}}^2\\
&\leq&[(\rho^2-\rho)^2+(\rho-\rho^2)]\cdot\big(1-(1-\frac{1}{n})^K\big).
\end{eqnarray*}
Similarly, we have
\begin{eqnarray*}
V'&=&\Var[\rho'_{\mathbf{B}_1}\rho'_{\mathbf{B}_2}]\\
&\leq&[((\rho')^2-\rho')^2+(\rho'-(\rho')^2)]\cdot\big(1-(1-\frac{1}{n})^{K'}\big).
\end{eqnarray*}
As a result, since
\begin{eqnarray*}
|d_n^2|&\leq&\big(1-(1-\frac{1}{n})^{K+K'-m}(1-\frac{2}{n})^m(1-\frac{3}{n})^m\big)\\
&\quad&\cdot \Var[\rho_{\mathbf{B}_1}\rho_{\mathbf{B}_2}]^{\frac{1}{2}}\Var[\rho'_{\mathbf{B}_1}\rho'_{\mathbf{B}_2}]^{\frac{1}{2}},
\end{eqnarray*}
we have
\begin{eqnarray*}
|\Cov(\rho_n^2,(\rho')_n^2)|=|d_n^2|\leq f(n,m)h(\rho)h(\rho'),
\end{eqnarray*}
where
\begin{eqnarray*}
f(n,m)&=&\big(1-(1-\frac{1}{n})^{K+K'-m}(1-\frac{2}{n})^m(1-\frac{3}{n})^m\big)\\
&\quad&\cdot\big(1-(1-\frac{1}{n})^K\big)^{\frac{1}{2}}\big(1-(1-\frac{1}{n})^{K'}\big)^{\frac{1}{2}},
\end{eqnarray*}
and $h(\rho) = \sqrt{\rho(1 - \rho)(\rho - \rho^2 + 1)}$. This completes the proof of the theorem.
\end{proof}

With very similar arguments, we are able to show that
\begin{theorem}\label{theorem:dn21-bound}
$$|\Cov(\rho_n^2, \rho'_n)| \leq f(n, m)h(\rho)g(\rho'),$$
where
\begin{eqnarray*}
f(n, m)&=&[1 - (1 - \frac{1}{n})^K(1 - \frac{2}{n})^m]\\
&\cdot &\sqrt{1 - (1 - \frac{1}{n})^K}\sqrt{1 - (1 - \frac{1}{n})^{K'}},
\end{eqnarray*}
$g(\rho)=\sqrt{\rho(1 - \rho)}$, and $h(\rho)=\sqrt{\rho(1 - \rho)(\rho - \rho^2 + 1)}$.
\end{theorem}
The approximate version is:
\begin{eqnarray*}
f(n, m)&\approx & [1 - (1 - \frac{K}{n})(1 - \frac{2m}{n})]\frac{\sqrt{KK'}}{n}\\
&\approx &\frac{(K + 2m)\sqrt{KK'}}{n^2}.
\end{eqnarray*}

\section{The Complete Framework}\label{sec:framework}

We present the complete framework of estimating the distribution of $t_q$ in Algorithm~\ref{alg:dist-tq}. Note that, the framework is a two-stage one: we first obtain the marginal distributions of the selectivities via sampling, and then obtain the distribution of $t_q$.

It is worth to point out that a more straightforward, one-stage alternative can also solve the problem: we keep running the query plan over different sample tables and observe the joint distribution of the selectivities. It will then directly give the distribution of the estimated running times: we simply plug in each observed selectivity vector $\mathbf{X}$ to the cost formulas and compute the running times. However, the overhead of this approach might be prohibitive in practice: we need the same number of sample runs as the observations we need to build the histogram of the running times.

Nonetheless, this conceptually simpler framework is of some theoretic interest. Note that in our current framework, we view the execution time as a function of the selectivities over \emph{all} operators. Can we instead view the time as a function of the selectivities over just the leaf operators? The answer is no, because the selectivities of the internal operators cannot be simply determined by the selectivities of the leaf nodes. However, we can indeed view the time as a function of the leaf tables. That is, as long as we fix the input table of each leaf (i.e., scan) operator, the selectivity of each internal operator is also fixed and hence the running time can be determined. Since different leaf tables may lead to the same selectivity on a leaf operator, this explains why simply fixing the selectivities of the leaf operators may not be sufficient for characterizing the running time of the query plan. However, a function of \emph{tables} is not feasible for mathematical analysis. The only way is to leave it as a black box and repeatedly feed it with different sample tables, which is costly and infeasible for our purpose of query execution time prediction. By instead representing the running time as a function of selectivities over all operators, we obtain something mathematically manipulatable and practically efficient, though we now need to address the new challenge of estimating the covariances between the selectivities.

\begin{algorithm}
  \SetAlgoLined
  \KwIn{$q$, the input query; $\mathcal{C}$, cost units calibrated offline}
  \KwOut{$t_q^{\mathcal{N}}\sim \mathcal{N}(\E[t_q], \Var[t_q])$}
  \SetAlgoLined
  $Agg\leftarrow false$\;

  \quad

  \textbf{EstSelDistr}$(O)$:

  \If{$O$ has left child $O_l$} {
    $EstSelDistr(O_l)$\;
  }
  \If{$O$ has right child $O_r$} {
    $EstSelDistr(O_r)$\;
  }
  \If{$O$ is aggregate} {
    $Agg\leftarrow true$\;
  }

  Compute $\rho_n$ and $S_n^2$ for $O$ using Algorithm~\ref{alg:sel-est}\;

  \quad

  \textbf{GetCostFunc}$(\mathcal{P}_q)$:

  \ForEach{$O\in\mathcal{P}_q$} {
    \ForEach{$f\in CostFunctions(O)$} {
        \uIf{$O$ is unary} {
            Collect $(X_l, f)$'s in $\mathcal{I}_l$\;
        }\Else {
            Collect $(X_l, X_r, f)$'s in $\mathcal{I}_l\times\mathcal{I}_r$\;
        }
        Compute $f$ by solving the optimization problem\;
    }
  }

  \quad

  \textbf{Main}:

  $\mathcal{P}_q\leftarrow GetQueryPlan(q)$\;

  Run $\mathcal{P}_q$ over the sample tables\;

  $O_{root}\leftarrow GetRootOperator(\mathcal{P}_q)$\;

  $EstSelDistr(O_{root})$\;

  $GetCostFunc(\mathcal{P}_q)$\;

  Estimate $t_q^\mathcal{N}\sim \mathcal{N}(\E[t_q],\Var[t_q])$ using Algorithm~\ref{alg:var-tq}\;

  \caption{Estimation of the distribution of $t_q$}
\label{alg:dist-tq}
\end{algorithm}

We further summarize the procedure of computing $\Var[t_q]$ in Algorithm~\ref{alg:var-tq}. We first compute the variances for the $t_k$'s, since we have shown that they are directly computable. We then collect all the paths from the leaf operators to the root and compute or bound the covariances between the operators along each path. Based on Lemma~\ref{lemma:correlated}, these are all the pairs of operators that we need to check the covariances. Based on if $\Cov(Z,Z')$ is computable, we directly compute or provide some upper bound for it.

\begin{algorithm}[!htb]
  \SetAlgoLined
  \KwIn{$\mathcal{P}_q$, the query plan of $q$}
  \KwOut{$\Var[t_q]$, the estimated variance of $t_q$}
  \SetAlgoLined
  $VarOps\leftarrow 0$, $CovOpsUb\leftarrow 0$\;

  \ForEach{$O_k\in\mathcal{P}_q$} {
    $VarOps\leftarrow VarOps + \Var[t_k]$\;
  }

  $\mathcal{L}\leftarrow GetLeafOps(\mathcal{P}_q)$\;

  \ForEach{$L\in\mathcal{L}$} {
    $\mathcal{P}\leftarrow GetPath(L)$\;
    \ForEach{$O,O'\in\mathcal{P}$ s.t. $O\neq O'$} {
        \ForEach{$\Cov(Z, Z')$} {
            \uIf{$\Cov(Z, Z')$ is computable} {
                $VarOps\leftarrow VarOps + ComputeCov(O, O', Z, Z')$\;
            }\Else {
                $CovOpsUb\leftarrow CovOpsUb + UpperBoundCov(O, O', Z, Z')$\;
            }
        }
    }
  }

  $\Var[t_q]\leftarrow VarOps + CovOpsUb$\;

  \Return{$\Var[t_q]$\;}
  \caption{Estimation of the variance of $t_q$}
\label{alg:var-tq}
\end{algorithm}

\section{More Experimental Results}\label{sec:more-exp-results}

In this section we present additional experimental results.

\subsection{Correlations}\label{sec:more-exp-results:correlations}

Table~\ref{tab:cc:pred-time} reports the $r_s$'s (and the corresponding $r_p$'s) for the benchmark queries over different hardware and database settings. Here, SR stands for the sampling ratio (see Section~\ref{sec:experiment:accuracy}). Values below 0.7 are highlighted.
As we mentioned, these are cases where the correlations are not strong.

\begin{table*}[!htb]
\centering
\begin{tabular}{|r||r|r||r|r||r|r|}
\hline
\multicolumn{1}{|c||}{$\quad$} & \multicolumn{2}{c||}{\textbf{MICRO}} & \multicolumn{2}{c||}{\textbf{SELJOIN}} & \multicolumn{2}{c|}{\textbf{TPCH}}\\
\hline
\multicolumn{1}{|c||}{SR} & \multicolumn{1}{c|}{PC1} & \multicolumn{1}{c||}{PC2} & \multicolumn{1}{c|}{PC1} & \multicolumn{1}{c||}{PC2} & \multicolumn{1}{c|}{PC1} & \multicolumn{1}{c|}{PC2}\\
\hline
\hline
\multicolumn{7}{|c|}{Uniform TPC-H 1GB Database}\\
\hline
0.01 & 0.9321 (0.9830) & 0.9400 (\textbf{0.5691}) & 0.7554 (0.8989) & 0.8551 (0.9724) & 0.7209 (0.7571) & 0.9457 (0.9688)\\
0.05 & 0.9381 (0.9875) & 0.9813 (0.7904) & \textbf{0.6958} (0.8414) & 0.9170 (0.9865) & 0.7171 (0.7738) & 0.9583 (0.9768)\\
0.1 & 0.9415 (0.9862) & 0.9740 (0.8252) & 0.7160 (0.8204) & 0.9265 (0.9883) & 0.7498 (0.7700) & 0.9607 (0.9778)\\
\hline
\hline
\multicolumn{7}{|c|}{Skewed TPC-H 1GB Database}\\
\hline
0.01 & 0.9418 (0.9753) & 0.9827 (0.9236) & 0.8545 (\textbf{0.5575}) & 0.9495 (0.8656) & 0.7829 (0.8768) & 0.9614 (0.9189)\\
0.05 & 0.9435 (0.9762) & 0.9838 (0.9130) & 0.8374 (\textbf{0.6502}) & 0.9621 (0.9543) & 0.9248 (0.9266) & 0.9729 (0.9897)\\
0.1 & 0.9431 (0.9765) & 0.9840 (0.9168) & 0.8546 (\textbf{0.6644}) & 0.9622 (0.9574) & 0.9248 (0.9285) & 0.9639 (0.9901)\\
\hline
\hline
\multicolumn{7}{|c|}{Uniform TPC-H 10GB Database}\\
\hline
0.01 & 0.9397 (0.9518) & 0.9853 (0.9549) & 0.9660 (0.8263) & 0.9054 (0.9288) & 0.8265 (0.9344) & 0.7926 (0.9614)\\
0.05 & 0.9379 (0.9675) & 0.9855 (0.9536) & 0.9774 (0.8786) & 0.9094 (0.9617) & 0.8749 (0.9592) & 0.8504 (0.9699)\\
0.1 & 0.9383 (0.9760) & 0.9853 (0.9539) & 0.9708 (0.8574) & 0.9095 (0.9649) & 0.8026 (0.9498) & 0.8559 (0.9706)\\
\hline
\hline
\multicolumn{7}{|c|}{Skewed TPC-H 10GB Database}\\
\hline
0.01 & 0.9674 (0.9665) & 0.9819 (0.9830) & 0.9636 (0.8986) & 0.9728 (0.9532) & 0.9480 (0.9696) & 0.8894 (0.9884)\\
0.05 & 0.9669 (0.9812) & 0.9841 (0.9831) & 0.9650 (0.9519) & 0.9784 (0.9761) & 0.9439 (0.9887) & 0.9127 (0.9936)\\
0.1 & 0.9675 (0.9905) & 0.9840 (0.9830) & 0.9663 (0.9580) & 0.9781 (0.9781) & 0.9354 (0.9910) & 0.9198 (0.9944)\\
\hline
\end{tabular}
\caption{$r_s$ ($r_p$) of the benchmark queries over different hardware and database settings (values below 0.7 are highlighted)}
\label{tab:cc:pred-time}
\shrink
\end{table*}

\subsection{Distributional Distances}\label{sec:more-exp-results:distances}

Table~\ref{tab:ks} reports the complete results of distributional distances for the benchmark queries (see Section~\ref{sec:experiment:accuracy}).
Values above 0.3 are highlighted. The closer a value is to 0, the better the proximity of two distributions is.

\begin{table*}[!htb]
\centering
\begin{tabular}{|r||r|r|r|r|r|r||r|r|r|r|r|r|}
\hline
\multicolumn{1}{|c||}{$\quad$} & \multicolumn{2}{c|}{\textbf{MICRO}} & \multicolumn{2}{c|}{\textbf{SELJOIN}} & \multicolumn{2}{c||}{\textbf{TPCH}}
& \multicolumn{2}{c|}{\textbf{MICRO}} & \multicolumn{2}{c|}{\textbf{SELJOIN}} & \multicolumn{2}{c|}{\textbf{TPCH}}\\
\hline
\multicolumn{1}{|c||}{SR} & \multicolumn{1}{c|}{PC1} & \multicolumn{1}{c|}{PC2} & \multicolumn{1}{c|}{PC1} & \multicolumn{1}{c|}{PC2} & \multicolumn{1}{c|}{PC1} & \multicolumn{1}{c||}{PC2} & \multicolumn{1}{c|}{PC1} & \multicolumn{1}{c|}{PC2} & \multicolumn{1}{c|}{PC1} & \multicolumn{1}{c|}{PC2} & \multicolumn{1}{c|}{PC1} & \multicolumn{1}{c|}{PC2}\\
\hline
\hline
\multicolumn{7}{|c||}{Uniform TPC-H 1GB Database} & \multicolumn{6}{c|}{Skewed TPC-H 1GB Database}\\
\hline
0.01 & \textbf{0.5573} & \textbf{0.6235} & 0.2228 & 0.2833 & 0.1175 & 0.0872 & 0.2276 & \textbf{0.3150} & 0.1882 & 0.1395 & 0.1717 & 0.0850\\
0.05 & 0.2728 & \textbf{0.3885} & 0.1563 & 0.1787 & 0.0610 & 0.0570 & 0.2286 & \textbf{0.3180} & 0.1686 & 0.1334 & 0.1691 & 0.1068\\
0.1 & 0.2312 & \textbf{0.3236} & 0.1170 & 0.1441 & 0.0520 & 0.0664 & 0.2286 & \textbf{0.3183} & 0.1695 & 0.1341 & 0.1691 & 0.1068\\
\hline
\hline
\multicolumn{7}{|c||}{Uniform TPC-H 10GB Database} & \multicolumn{6}{c|}{Skewed TPC-H 10GB Database}\\
\hline
0.01 & 0.1663 & 0.2532 & 0.0766 & 0.1097 & 0.0579 & 0.0502 & 0.1170 & 0.2512 & 0.1052 & 0.1316 & 0.1388 & 0.0713\\
0.05 & 0.1657 & 0.2532 & 0.0765 & 0.1098 & 0.0595 & 0.0535 & 0.1158 & 0.2524 & 0.1022 & 0.1275 & 0.1296 & 0.0769\\
0.1 & 0.1622 & 0.2532 & 0.0722 & 0.1091 & 0.0591 & 0.0558 & 0.1136 & 0.2518 & 0.1040 & 0.1282 & 0.1296 & 0.0814\\
\hline
\end{tabular}
\caption{$\overline{D}_n$ of the benchmark queries over different hardware and database settings (values above 0.3 are highlighted)}
\label{tab:ks}
\shrink
\end{table*}

\subsection{Comparison with Simplified Versions}\label{sec:more-exp-results:simplified}

Figure~\ref{fig:simplified:cc:skewed} presents more results on comparison of the four alternatives discussed in Section~\ref{sec:experiment:accuracy} over skewed databases for the \textbf{TPCH} queries.
The observations are similar to that over uniform databases as presented in Section~\ref{sec:experiment:accuracy}.

\begin{figure}[!htb]
\centering
\subfigure[Skewed 1GB database, PC1]{ \label{fig:simplified:cc:skewed:a}
\includegraphics[width=\columnwidth]{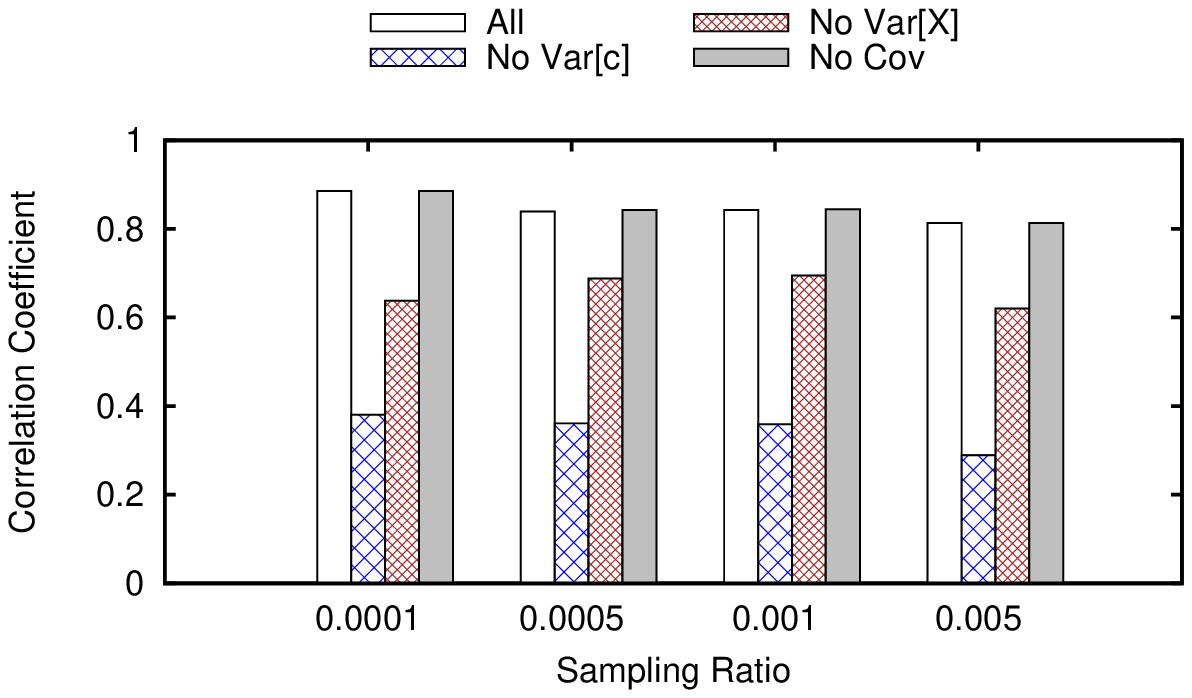}}
\subfigure[Skewed 10GB database, PC2]{ \label{fig:simplified:cc:skewed:b}
\includegraphics[width=\columnwidth]{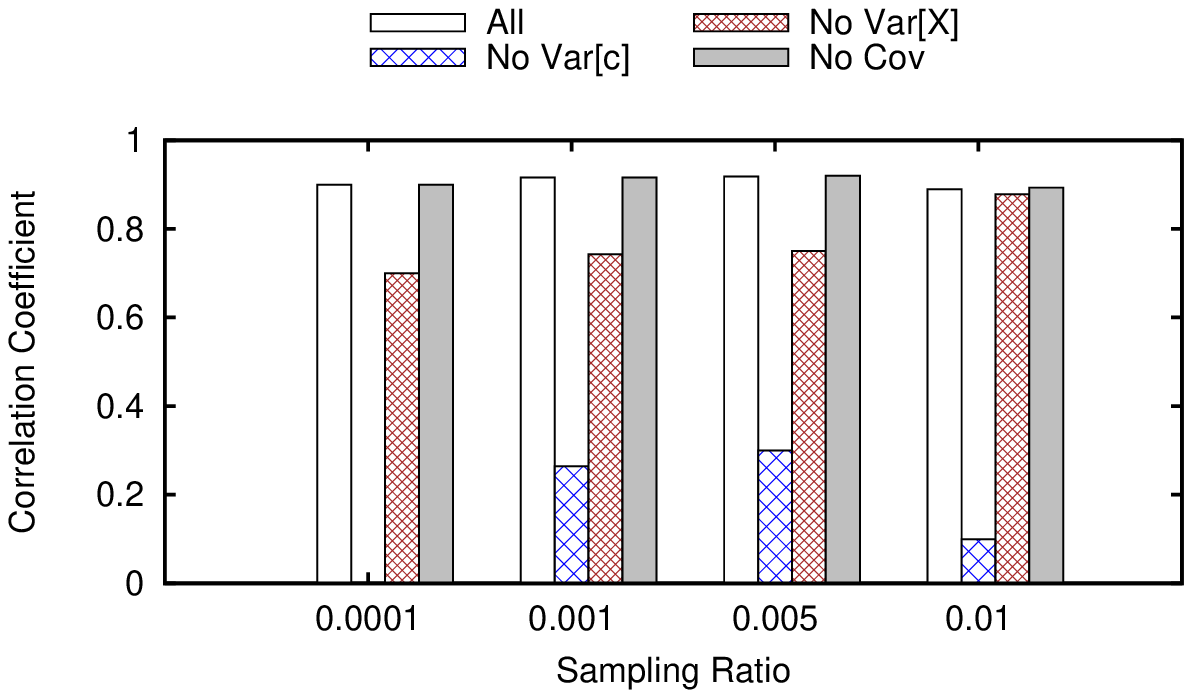}}
\caption{Comparison of four alternatives in terms of $r_s$}
\label{fig:simplified:cc:skewed}
\shrink
\end{figure}

\subsection{Sampling Overhead}\label{sec:more-exp-results:overhead}

Figure~\ref{fig:rel-overhead:all} reports the complete experimental results for the relative overhead of running the queries over the sample tables, which were omitted in Section~\ref{sec:exp:sample-overhead}.

\begin{figure*}[!htb]
\centering
\subfigure[\textbf{MICRO}, PC1]{ \label{fig:rel-overhead:pc1-bogus}
\includegraphics[width=0.66\columnwidth]{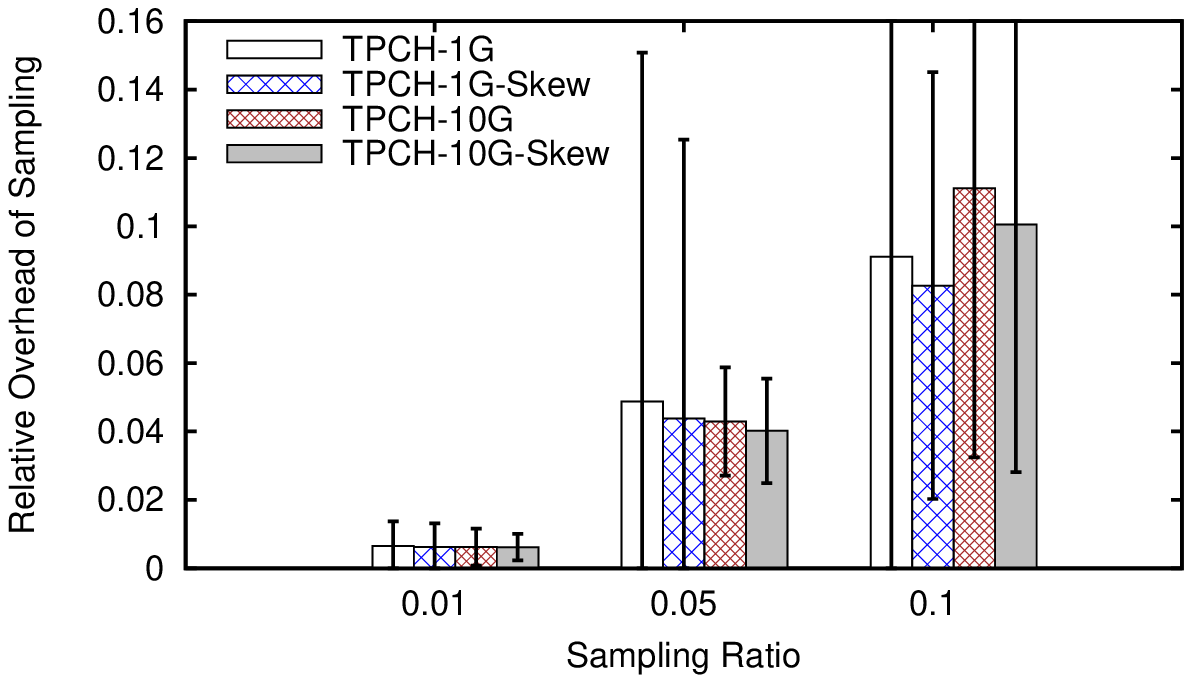}}
\subfigure[\textbf{SELJOIN}, PC1]{ \label{fig:rel-overhead:pc1-sj}
\includegraphics[width=0.66\columnwidth]{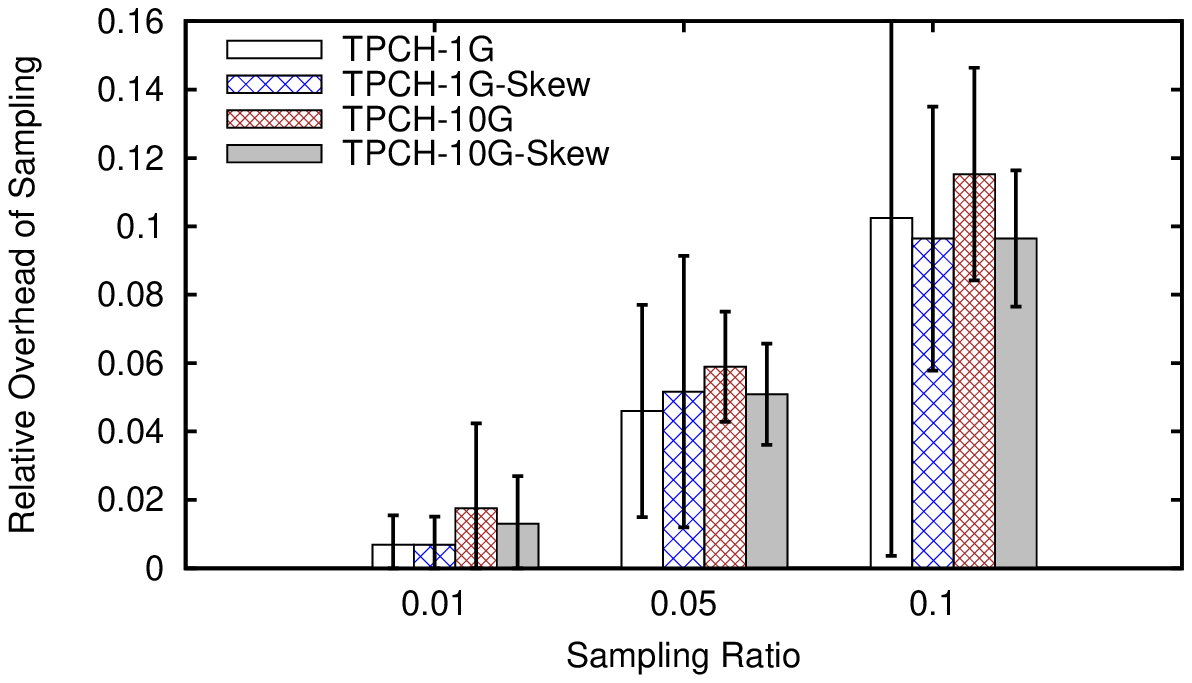}}
\subfigure[\textbf{TPCH}, PC1]{ \label{fig:rel-overhead:pc1-tpch}
\includegraphics[width=0.66\columnwidth]{./figs/rel-overhead-pc1-tpch.eps}}
\subfigure[\textbf{MICRO}, PC2]{ \label{fig:rel-overhead:pc2-bogus}
\includegraphics[width=0.66\columnwidth]{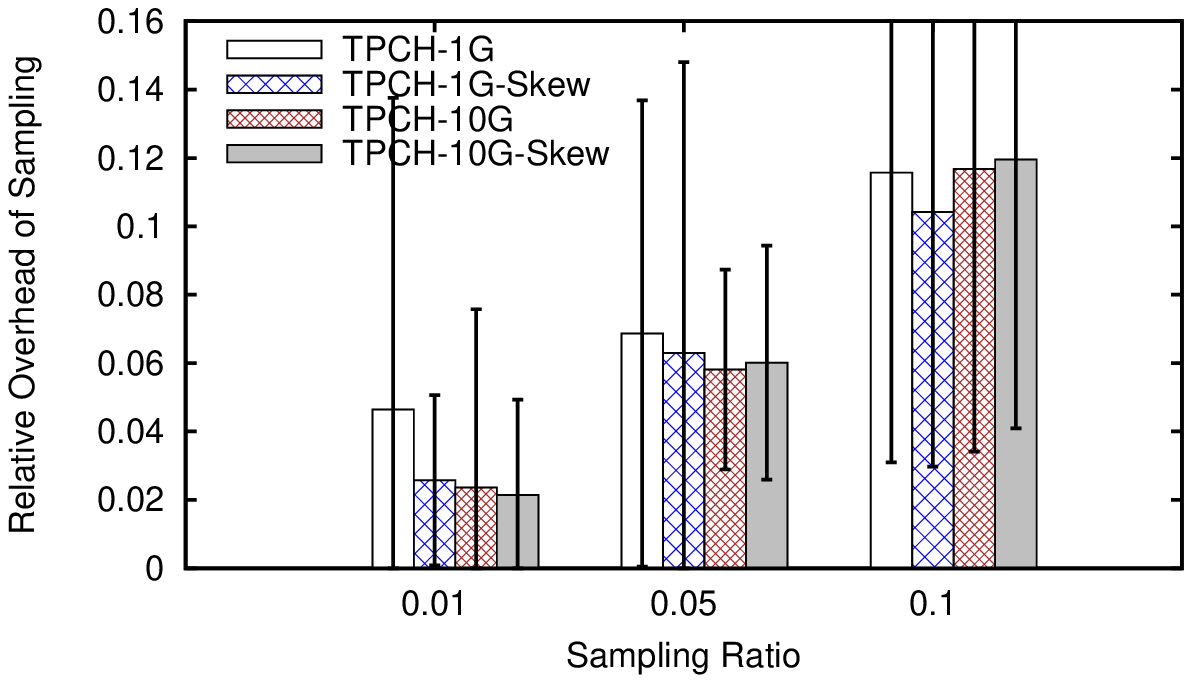}}
\subfigure[\textbf{SELJOIN}, PC2]{ \label{fig:rel-overhead:pc2-sj}
\includegraphics[width=0.66\columnwidth]{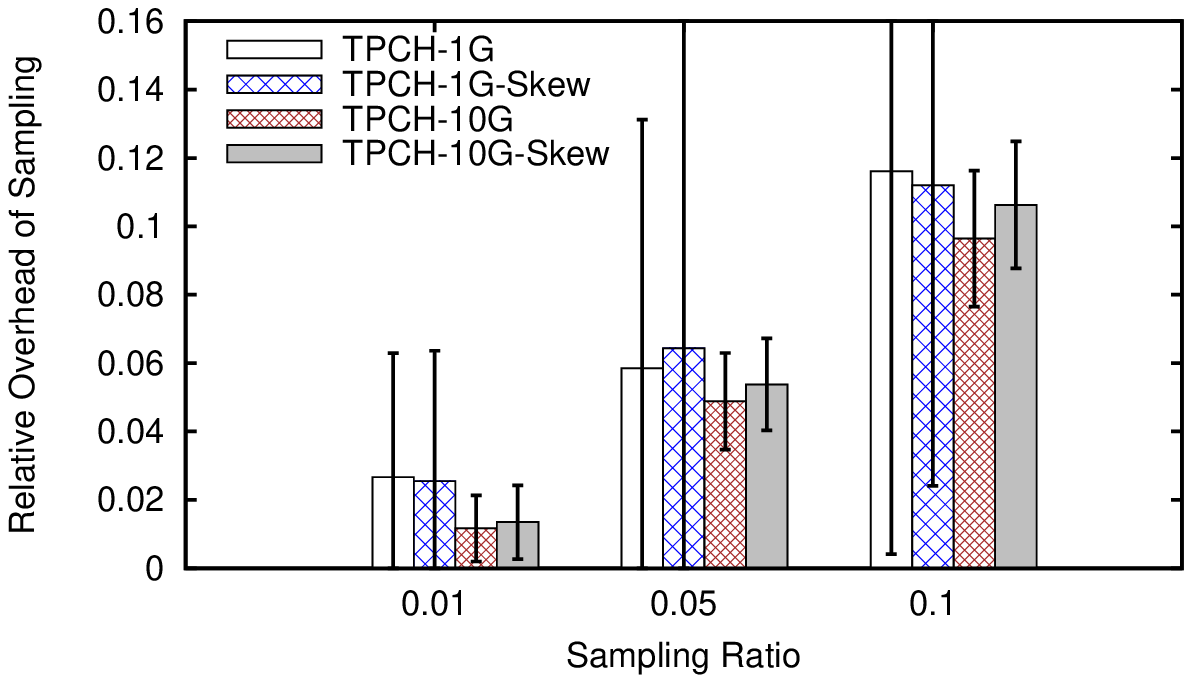}}
\subfigure[\textbf{TPCH}, PC2]{ \label{fig:rel-overhead:pc2-tpch}
\includegraphics[width=0.66\columnwidth]{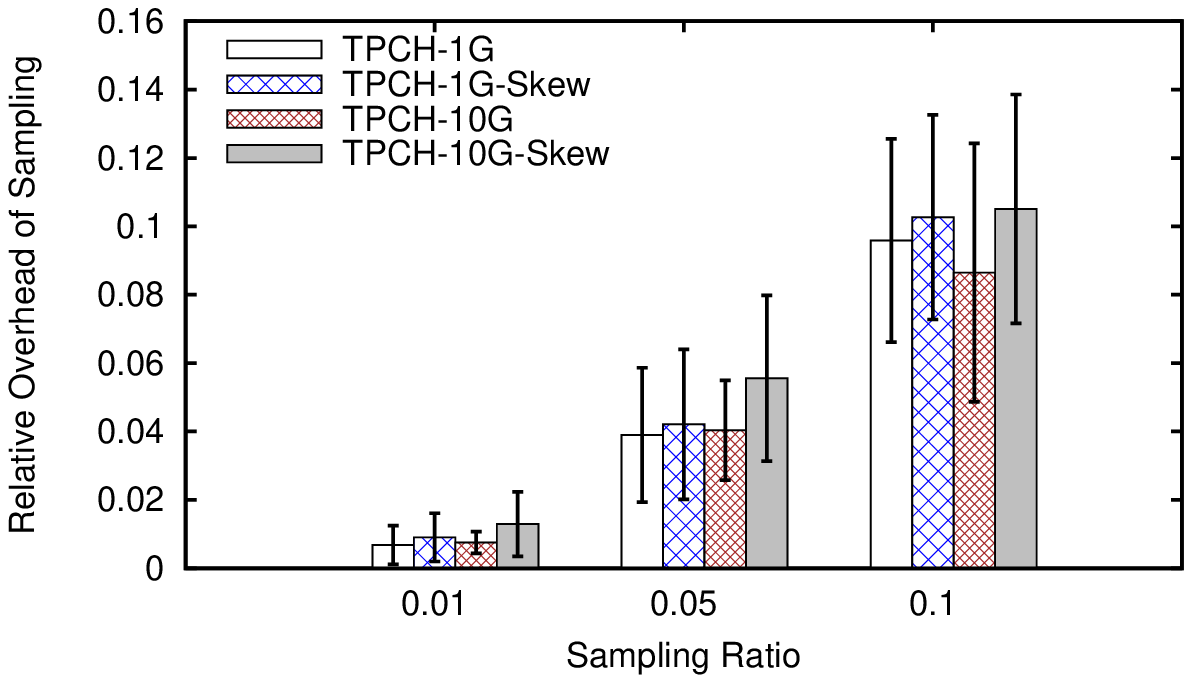}}
\caption{Relative overhead of benchmark queries}
\label{fig:rel-overhead:all}
\shrink
\end{figure*}

\begin{table*}[!htb]
\centering
\begin{tabular}{|r||l|l||l|l||l|l|}
\hline
\multicolumn{1}{|c||}{$\quad$} & \multicolumn{2}{c||}{\textbf{MICRO}} & \multicolumn{2}{c||}{\textbf{SELJOIN}} & \multicolumn{2}{c|}{\textbf{TPCH}}\\
\hline
\multicolumn{1}{|c||}{SR} & \multicolumn{1}{c|}{PC1} & \multicolumn{1}{c||}{PC2} & \multicolumn{1}{c|}{PC1} & \multicolumn{1}{c||}{PC2} & \multicolumn{1}{c|}{PC1} & \multicolumn{1}{c|}{PC2}\\
\hline
\hline
\multicolumn{7}{|c|}{Uniform TPC-H 1GB Database}\\
\hline
0.01 & 0.8127 (0.8730) & 0.9245 (0.9621) & 0.7731 (\textbf{0.2294}) & 0.7465 (\textbf{-0.0015}) & 0.7168 (\textbf{0.0016}) & 0.7612 (\textbf{0.2239})\\
0.05 & 0.8670 (0.7109) & 0.8497 (0.8542) & 0.7818 (\textbf{0.0470}) & 0.7731 (\textbf{0.0247}) & 0.7356 (\textbf{-0.0322}) & 0.7807 (\textbf{0.0676})\\
0.1 & 0.8116 (0.7556) & 0.8246 (0.8216) & 0.7907 (\textbf{0.0276}) & 0.7770 (\textbf{0.0172}) & 0.7304 (\textbf{-0.0339}) & 0.7739 (\textbf{0.0199})\\
0.2 & 0.8102 (\textbf{0.6451}) & 0.8897 (\textbf{0.6541}) & 0.8065 (\textbf{-0.0034}) & 0.7942 (\textbf{-0.0044}) & 0.7506 (\textbf{-0.0606}) & 0.7676 (\textbf{-0.0145})\\
0.3 & 0.9107 (\textbf{0.5566}) & 0.9274 (0.7311) & 0.7940 (\textbf{-0.0056}) & 0.7944 (\textbf{-0.0213}) & 0.7396 (\textbf{-0.0341}) & 0.7930 (\textbf{-0.0346})\\
0.4 & 0.8654 (\textbf{0.4994}) & 0.8988 (0.7034) & 0.7923 (\textbf{-0.0314}) & 0.7902 (\textbf{-0.0239}) & 0.7534 (\textbf{-0.0442}) & 0.7898 (\textbf{-0.0166})\\
\hline
\hline
\multicolumn{7}{|c|}{Skewed TPC-H 1GB Database}\\
\hline
0.01 & 0.8736 (0.7780) & 0.8938 (0.8807) & \textbf{0.6172} (\textbf{-0.0178}) & \textbf{0.6259} (\textbf{-0.0175}) & \textbf{0.5949} (\textbf{0.0145}) & \textbf{0.6556} (\textbf{0.1528})\\
0.05 & 0.9007 (0.7815) & 0.9066 (0.8129) & \textbf{0.6274} (\textbf{0.0265}) & \textbf{0.6293} (\textbf{0.0182}) & \textbf{0.5772} (\textbf{-0.0310}) & \textbf{0.6485} (\textbf{0.0269})\\
0.1 & 0.7748 (\textbf{0.2938}) & 0.9320 (0.8578) & \textbf{0.6347} (\textbf{-0.0306}) & \textbf{0.6286} (\textbf{0.0176}) & \textbf{0.5952} (\textbf{-0.0525}) & \textbf{0.6324} (\textbf{0.0176})\\
0.2 & 0.7566 (\textbf{0.5545}) & 0.8772 (\textbf{0.6146}) & \textbf{0.6360} (\textbf{-0.0202}) & \textbf{0.6211} (\textbf{-0.0174}) & \textbf{0.5970} (\textbf{-0.0352}) & \textbf{0.6071} (\textbf{-0.0127})\\
0.3 & 0.7880 (\textbf{0.4806}) & 0.9137 (\textbf{0.5158}) & \textbf{0.6505} (\textbf{-0.0218}) & 0.7093 (\textbf{-0.0180}) & \textbf{0.5631} (\textbf{-0.0472}) & \textbf{0.6921} (\textbf{-0.0377})\\
0.4 & 0.8063 (\textbf{0.1580}) & 0.8722 (\textbf{0.6483}) & \textbf{0.6808} (\textbf{-0.0520}) & \textbf{0.6180} (\textbf{-0.0198}) & \textbf{0.6553} (\textbf{-0.0601}) & \textbf{0.6161} (\textbf{-0.0175})\\
\hline
\hline
\multicolumn{7}{|c|}{Uniform TPC-H 10GB Database}\\
\hline
0.01 & 0.8407 (\textbf{0.6932}) & 0.9311 (0.7887) & \textbf{0.6594} (\textbf{0.0283}) & \textbf{0.6481} (\textbf{0.0050}) & 0.7395 (\textbf{0.0347}) & 0.8199 (\textbf{0.0048})\\
0.02 & 0.9080 (\textbf{0.6594}) & 0.8781 (0.7153) & \textbf{0.6524} (\textbf{0.0069}) & \textbf{0.6425} (\textbf{-0.0068}) & 0.7362 (\textbf{-0.0122}) & 0.8062 (\textbf{0.0363})\\
0.05 & 0.9004 (\textbf{0.2230}) & 0.9208 (\textbf{0.6030}) & \textbf{0.6366} (\textbf{-0.0132}) & 0.7474 (\textbf{-0.0143}) & 0.7240 (\textbf{-0.0177}) & 0.8313 (\textbf{0.0105})\\
0.1 & 0.8733 (\textbf{0.2993}) & 0.7862 (\textbf{0.3673}) & \textbf{0.6696} (\textbf{-0.0470}) & \textbf{0.6579} (\textbf{-0.0359}) & 0.7461 (\textbf{-0.0514}) & 0.8240 (\textbf{-0.0262})\\
\hline
\hline
\multicolumn{7}{|c|}{Skewed TPC-H 10GB Database}\\
\hline
0.01 & 0.9365 (\textbf{0.6938}) & 0.8742 (0.8389) & \textbf{0.6187} (\textbf{0.0487}) & \textbf{0.6020} (\textbf{-0.0088}) & \textbf{0.6988} (\textbf{-0.0232}) & 0.7820 (\textbf{-0.0170})\\
0.02 & 0.8273 (\textbf{0.5548}) & 0.8929 (0.7476) & \textbf{0.5771} (\textbf{0.0424}) & \textbf{0.6017} (\textbf{0.0029}) & \textbf{0.6812} (\textbf{-0.0291}) & 0.7787 (\textbf{0.0815})\\
0.05 & 0.8008 (\textbf{0.4130}) & 0.8855 (\textbf{0.4758}) & \textbf{0.5537} (\textbf{-0.0133}) & 0.7081 (\textbf{0.0165}) & \textbf{0.6441} (\textbf{-0.0602}) & 0.7274 (\textbf{0.0339})\\
0.1 & 0.7808 (\textbf{0.3152}) & 0.8712 (\textbf{0.4872}) & \textbf{0.5978} (\textbf{-0.0585}) & \textbf{0.6855} (\textbf{-0.0254}) & \textbf{0.6548} (\textbf{-0.0417}) & 0.7366 (\textbf{0.0086})\\
\hline
\end{tabular}
\caption{$r_s$ ($r_p$) between the estimated and actual errors in selectivity estimates (values below 0.7 are highlighted)}
\label{tab:cc:sel-err}
\shrink
\end{table*}

\begin{table*}[!htb]
\centering
\begin{tabular}{|r||l|l||l|l||l|l|}
\hline
\multicolumn{1}{|c||}{$\quad$} & \multicolumn{2}{c||}{\textbf{MICRO}} & \multicolumn{2}{c||}{\textbf{SELJOIN}} & \multicolumn{2}{c|}{\textbf{TPCH}}\\
\hline
\multicolumn{1}{|c||}{SR} & \multicolumn{1}{c|}{PC1} & \multicolumn{1}{c||}{PC2} & \multicolumn{1}{c|}{PC1} & \multicolumn{1}{c||}{PC2} & \multicolumn{1}{c|}{PC1} & \multicolumn{1}{c|}{PC2}\\
\hline
\hline
\multicolumn{7}{|c|}{Uniform TPC-H 1GB Database}\\
\hline
0.01 & 0.9808 (0.9977) & 0.9826 (0.9916) & 0.9934 (1.0000) & 0.9907 (0.9939) & 0.9962 (1.0000) & 0.9967 (1.0000)\\
0.05 & 0.9829 (0.9992) & 0.9923 (0.9970) & 0.9930 (1.0000) & 0.9956 (1.0000) & 0.9973 (1.0000) & 0.9971 (1.0000)\\
0.1 & 0.9920 (0.9993) & 0.9910 (0.9998) & 0.9971 (1.0000) & 0.9986 (1.0000) & 0.9973 (1.0000) & 0.9978 (1.0000)\\
0.2 & 0.9873 (0.9997) & 0.9925 (0.9999) & 0.9997 (1.0000) & 0.9997 (1.0000) & 0.9982 (1.0000) & 0.9972 (1.0000)\\
0.3 & 0.9878 (0.9996) & 0.9961 (0.9998) & 0.9982 (1.0000) & 0.9982 (1.0000) & 0.9986 (1.0000) & 0.9985 (1.0000)\\
0.4 & 0.9958 (0.9996) & 0.9949 (0.9997) & 0.9984 (1.0000) & 0.9973 (1.0000) & 0.9993 (1.0000) & 0.9992 (1.0000)\\
\hline
\hline
\multicolumn{7}{|c|}{Skewed TPC-H 1GB Database}\\
\hline
0.01 & 0.9896 (0.9986) & 0.9964 (0.9973) & 0.9741 (0.9904) & 0.9850 (0.9826) & 0.9942 (1.0000) & 0.9957 (1.0000)\\
0.05 & 0.9994 (0.9997) & 0.9983 (0.9996) & 0.9934 (1.0000) & 0.9930 (1.0000) & 0.9938 (1.0000) & 0.9962 (1.0000)\\
0.1 & 0.9987 (0.9999) & 0.9985 (0.9998) & 0.9947 (1.0000) & 0.9955 (1.0000) & 0.9947 (1.0000) & 0.9955 (1.0000)\\
0.2 & 0.9994 (0.9998) & 0.9997 (1.0000) & 0.9968 (1.0000) & 0.9964 (1.0000) & 0.9965 (1.0000) & 0.9952 (1.0000)\\
0.3 & 0.9993 (0.9999) & 0.9999 (1.0000) & 0.9977 (1.0000) & 0.9996 (1.0000) & 0.9970 (1.0000) & 0.9980 (1.0000)\\
0.4 & 0.9997 (1.0000) & 0.9994 (0.9998) & 0.9988 (1.0000) & 0.9967 (1.0000) & 0.9992 (1.0000) & 0.9962 (1.0000)\\
\hline
\hline
\multicolumn{7}{|c|}{Uniform TPC-H 10GB Database}\\
\hline
0.01 & 0.9964 (0.9996) & 0.9979 (0.9995) & 0.9885 (1.0000) & 0.9883 (1.0000) & 0.9823 (1.0000) & 0.9932 (1.0000)\\
0.02 & 0.9866 (0.9997) & 0.9921 (0.9998) & 0.9896 (1.0000) & 0.9890 (1.0000) & 0.9827 (1.0000) & 0.9940 (1.0000)\\
0.05 & 0.9959 (1.0000) & 0.9938 (1.0000) & 0.9945 (1.0000) & 0.9938 (1.0000) & 0.9894 (1.0000) & 0.9957 (1.0000)\\
0.1 & 0.9964 (1.0000) & 0.9968 (1.0000) & 0.9974 (1.0000) & 0.9969 (1.0000) & 0.9977 (1.0000) & 0.9990 (1.0000)\\
\hline
\hline
\multicolumn{7}{|c|}{Skewed TPC-H 10GB Database}\\
\hline
0.01 & 0.9986 (0.9994) & 0.9957 (0.9994) & 0.9881 (0.9868) & 0.9904 (0.9942) & 0.9925 (1.0000) & 0.9884 (0.9838)\\
0.02 & 0.9992 (0.9998) & 0.9999 (1.0000) & 0.9934 (0.9996) & 0.9900 (0.9936) & 0.9925 (1.0000) & 0.9946 (1.0000)\\
0.05 & 0.9992 (0.9999) & 0.9993 (0.9999) & 0.9893 (1.0000) & 0.9966 (0.9997) & 0.9912 (1.0000) & 0.9935 (1.0000)\\
0.1 & 0.9999 (1.0000) & 0.9997 (1.0000) & 0.9963 (1.0000) & 0.9978 (1.0000) & 0.9939 (1.0000) & 0.9944 (1.0000)\\
\hline
\end{tabular}
\caption{$r_s$ ($r_p$) between the estimated and actual selectivities (values below 0.7 are highlighted)}
\label{tab:cc:sel}
\shrink
\end{table*}

\begin{table*}[!htb]
\centering
\begin{tabular}{|r||r|r|r|r|r|r||r|r|r|r|r|r|}
\hline
\multicolumn{1}{|c||}{$\quad$} & \multicolumn{2}{c|}{\textbf{MICRO}} & \multicolumn{2}{c|}{\textbf{SELJOIN}} & \multicolumn{2}{c||}{\textbf{TPCH}}
& \multicolumn{2}{c|}{\textbf{MICRO}} & \multicolumn{2}{c|}{\textbf{SELJOIN}} & \multicolumn{2}{c|}{\textbf{TPCH}}\\
\hline
\multicolumn{1}{|c||}{SR} & \multicolumn{1}{c|}{PC1} & \multicolumn{1}{c|}{PC2} & \multicolumn{1}{c|}{PC1} & \multicolumn{1}{c|}{PC2} & \multicolumn{1}{c|}{PC1} & \multicolumn{1}{c||}{PC2} & \multicolumn{1}{c|}{PC1} & \multicolumn{1}{c|}{PC2} & \multicolumn{1}{c|}{PC1} & \multicolumn{1}{c|}{PC2} & \multicolumn{1}{c|}{PC1} & \multicolumn{1}{c|}{PC2}\\
\hline
\hline
\multicolumn{7}{|c||}{Uniform TPC-H 1GB Database} & \multicolumn{6}{c|}{Skewed TPC-H 1GB Database}\\
\hline
0.01 & 0.1328 & \textbf{0.2299} & \textbf{0.4678} & \textbf{0.4216} & \textbf{0.5349} & \textbf{0.4738} & 0.1454 & 0.1134 & \textbf{0.6256} & \textbf{0.4279} & \textbf{0.6689} & \textbf{0.5745}\\
0.05 & 0.0340 & 0.0551 & 0.1824 & \textbf{0.2946} & \textbf{0.3121} & \textbf{0.2212} & 0.0508 & 0.0445 & \textbf{0.3765} & \textbf{0.2505} & \textbf{0.7402} & \textbf{0.2580}\\
0.1 & 0.0306 & 0.0318 & 0.1586 & 0.1484 & 0.1988 & 0.1967 & 0.0393 & 0.0255 & 0.1593 & 0.1567 & 0.1815 & 0.1847\\
0.2 & 0.0197 & 0.0122 & 0.1134 & 0.0836 & 0.1017 & 0.1590 & 0.0210 & 0.0250 & 0.1397 & 0.1117 & 0.1127 & 0.1294\\
0.3 & 0.0144 & 0.0132 & 0.0577 & 0.0422 & 0.0583 & 0.0734 & 0.0197 & 0.0161 & 0.0910 & 0.0858 & 0.0890 & 0.1121\\
0.4 & 0.0146 & 0.0200 & 0.0389 & 0.0371 & 0.0585 & 0.0534 & 0.0115 & 0.0203 & 0.0770 & 0.0581 & 0.0878 & 0.0651\\
\hline
\hline
\multicolumn{7}{|c||}{Uniform TPC-H 10GB Database} & \multicolumn{6}{c|}{Skewed TPC-H 10GB Database}\\
\hline
0.01 & 0.0381 & 0.0492 & \textbf{0.3162} & \textbf{0.3396} & \textbf{0.4475} & \textbf{0.5894} & 0.0591 & 0.0542 & \textbf{0.3433} & \textbf{0.3478} & \textbf{0.4723} & \textbf{0.5955}\\
0.02 & 0.0342 & 0.0241 & \textbf{0.2419} & \textbf{0.2344} & \textbf{0.3439} & \textbf{0.3533} & 0.0514 & 0.0336 & \textbf{0.3471} & \textbf{0.2671} & \textbf{0.3819} & \textbf{0.4354}\\
0.05 & 0.0147 & 0.0101 & 0.1491 & 0.1713 & \textbf{0.2018} & \textbf{0.2130} & 0.0254 & 0.0203 & 0.1809 & 0.1914 & \textbf{0.2046} & \textbf{0.2732}\\
0.1 & 0.0068 & 0.0096 & 0.1047 & 0.0800 & 0.1291 & 0.1499 & 0.0116 & 0.0127 & 0.1280 & 0.1389 & 0.1313 & 0.1754\\
\hline
\end{tabular}
\caption{Relative errors in the selectivity estimates of the benchmark queries (values above 0.2 are highlighted)}
\label{tab:sel-rel-err}
\shrink
\end{table*}

\subsection{Selectivity Estimates}\label{sec:more-exp-results:selectivity}

Given a query $q$, the quality of the estimated distribution of $t_q$ depends on a number of factors such as the accuracy of the distributions of the $c$'s and the $X$'s, the quality of the approximated cost functions, and the closeness of the upper bounds of the covariances to the actual values. Note that how well we could estimate the potential errors in selectivity estimates plays a crucial role here, for it directly impacts the accuracy of the estimated distributions of the $X$'s, which further impacts the accuracy of the approximated cost functions as well as the estimated covariances. Therefore, we further studied the correlations between the estimated and actual errors in selectivity estimates. Since the selectivities are modeled as Gaussian variables, we again focus on measuring the correlations between the standard deviations of the Gaussian distributions and the actual errors in the selectivity estimates, as what we did in Section~\ref{sec:experiment:accuracy} for the distributions of the $t_q$'s. For this sake, we examined the correlations for the selective operators (i.e., selections and joins) of each query in the benchmarks under different hardware and database settings. Table~\ref{tab:cc:sel-err} presents the results.

\begin{figure*}
\centering
\subfigure[\textbf{MICRO}]{ \label{fig:cc:sel:bogus}
\includegraphics[width=0.68\columnwidth]{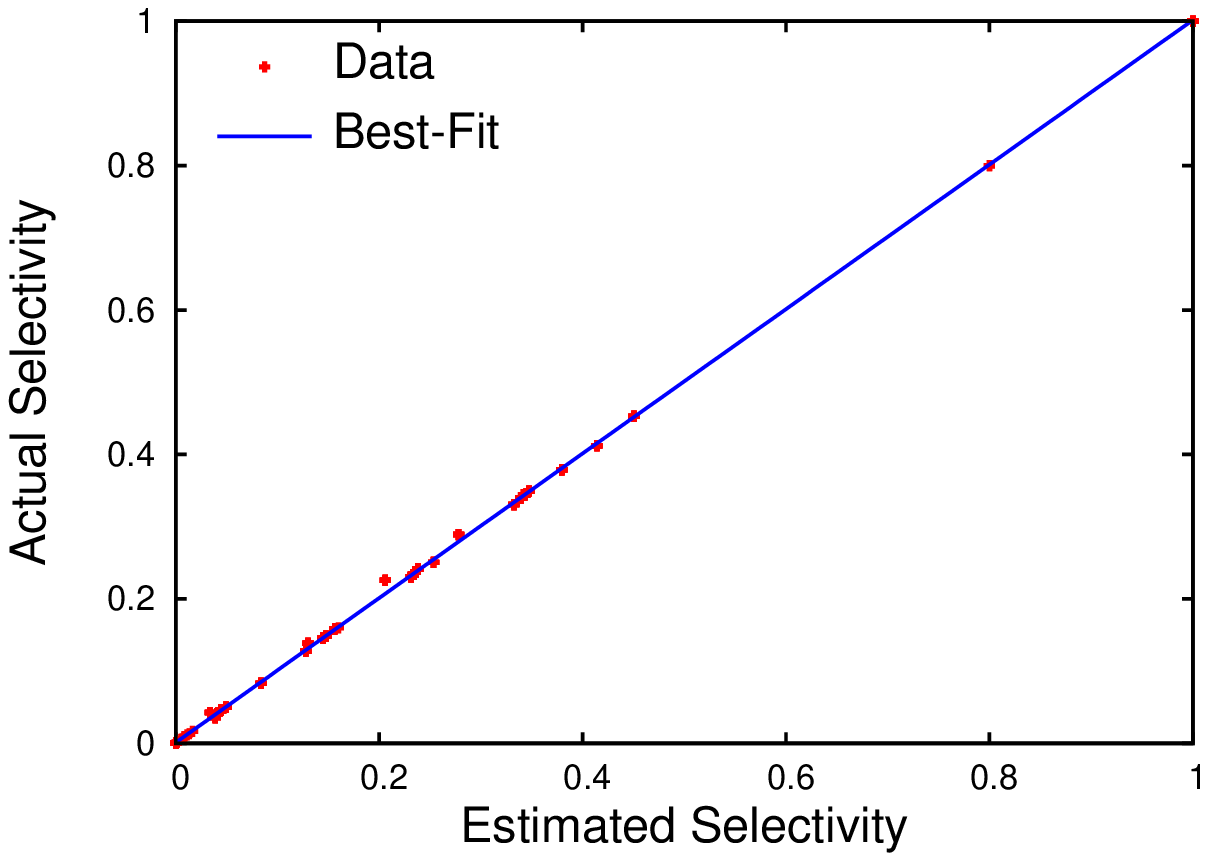}}
\subfigure[\textbf{SELJOIN}]{ \label{fig:cc:sel:sj}
\includegraphics[width=0.68\columnwidth]{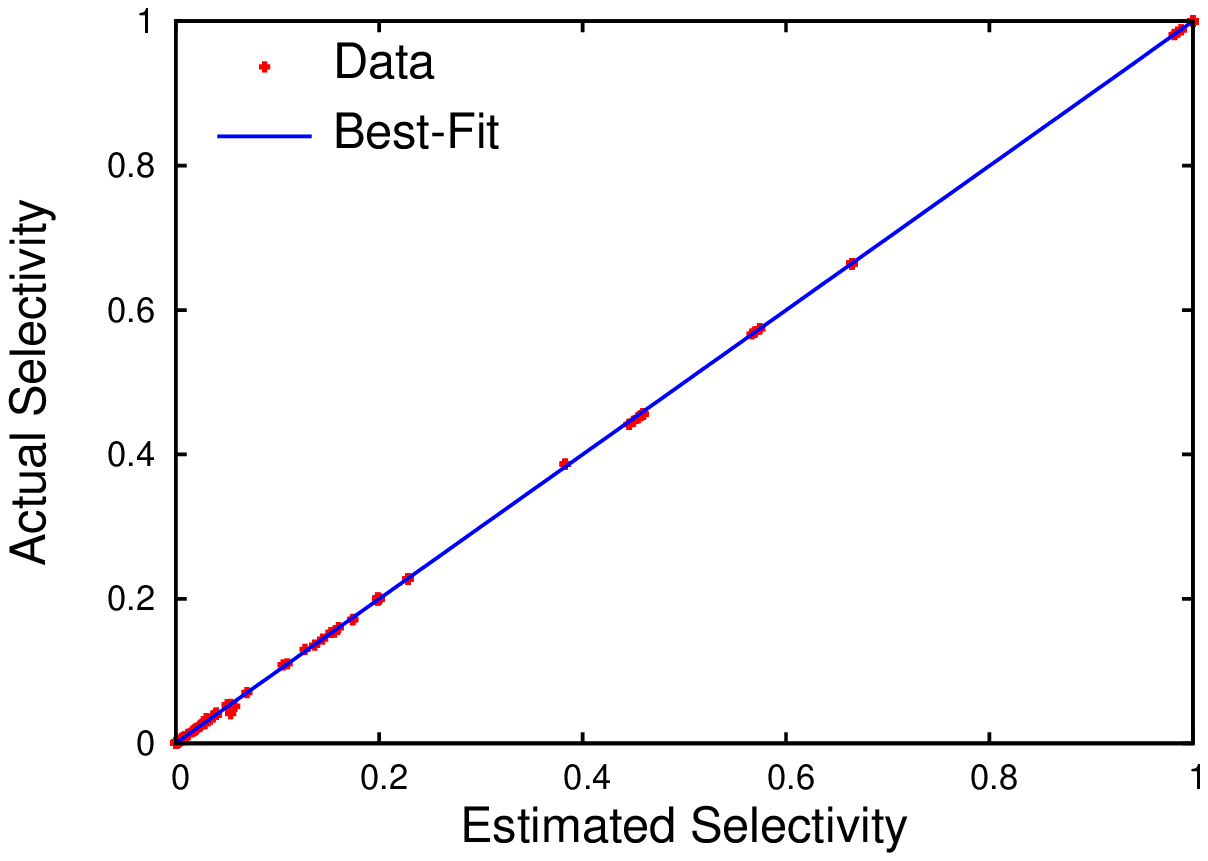}}
\subfigure[\textbf{TPCH}]{ \label{fig:cc:sel:tpch}
\includegraphics[width=0.68\columnwidth]{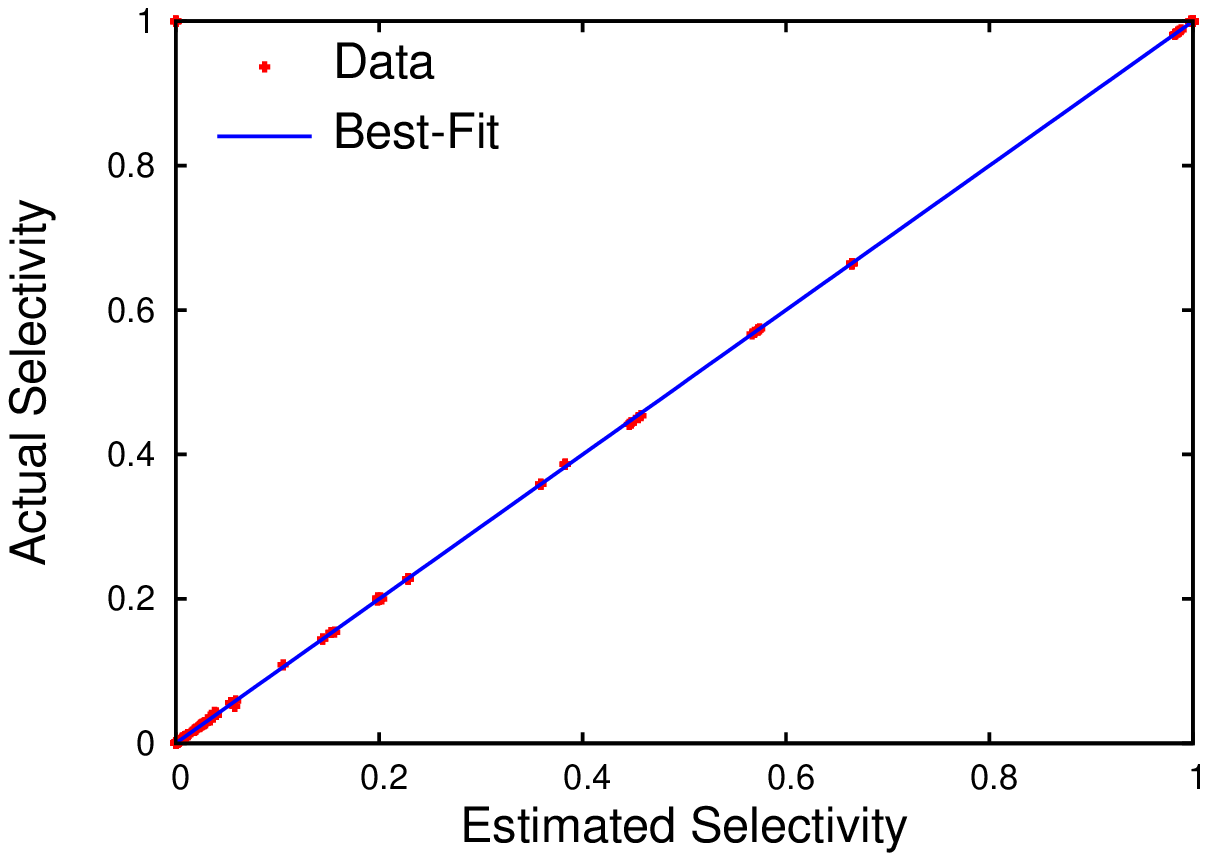}}
\caption{A case study of the correlations between the estimated and actual selectivities}
\label{fig:cc:sel}
\shrink
\end{figure*}

We observe that the correlations are not as good as that of the $t_q$'s in Table~\ref{tab:cc:pred-time}. In particular, there seems to be no linear correlations for the \textbf{SELJOIN} and \textbf{TPCH} queries by examining the corresponding $r_p$'s. One possible reason for this phenomenon is that the actual errors are usually too small. To verify this, in Table~\ref{tab:cc:sel} we present the correlations between the estimated and the actual selectivities, and in Table~\ref{tab:sel-rel-err} we compute the mean relative errors in the selectivity estimates, where the relative error of an estimated selectivity $\rho_n$ is defined as $\frac{|\rho_n-\rho|}{\rho}$. We find that the $r_p$'s between the estimated and actual selectivities are almost 1 for almost all the cases we tested, which suggests a very strong linear dependency. As a case study, in Figure~\ref{fig:cc:sel} we show the scatter plots of the \textbf{MICRO}, \textbf{SELJOIN}, and \textbf{TPCH} queries over skewed 1GB database on PC1 with SR = 0.05. We can see that the estimated selectivities are almost the same as the ground truths. On the other hand, the average relative errors are usually below 20\% according to Table~\ref{tab:sel-rel-err}. Note that our sampling based method cannot be very effective when the errors are too small unless we can have very large number of samples. This is because we estimate the variances of the distributions based on the observed variances in the samples. Since the samples are taken randomly, different batches of samples may present different sample variances although the variation should be small due to the strong consistency property. However, if the errors in selectivity estimates are already very small, then the small variation in sample variances now cannot be negligible. The impact on the correlations due to this variation is unpredictable since this variation is totally random. But the correlations here seem to be not very important, since based on the small variances we can still infer that the errors are small although we have no idea of the correlations. In Table~\ref{tab:cc:sel-big-err} we further present the correlations when the relative errors are above 0.2. We find that now the correlations are much better. This implies that the estimated errors are strongly correlated with the actual errors in selectivity estimates when the errors are relatively large.

\begin{table*}[t]
\centering
\begin{tabular}{|r||l|l||l|l||l|l|}
\hline
\multicolumn{1}{|c||}{$\quad$} & \multicolumn{2}{c||}{\textbf{MICRO}} & \multicolumn{2}{c||}{\textbf{SELJOIN}} & \multicolumn{2}{c|}{\textbf{TPCH}}\\
\hline
\multicolumn{1}{|c||}{SR} & \multicolumn{1}{c|}{PC1} & \multicolumn{1}{c||}{PC2} & \multicolumn{1}{c|}{PC1} & \multicolumn{1}{c||}{PC2} & \multicolumn{1}{c|}{PC1} & \multicolumn{1}{c|}{PC2}\\
\hline
\hline
\multicolumn{7}{|c|}{Uniform TPC-H 1GB Database}\\
\hline
0.01 & \textbf{0.6980} (1.0000) & 0.8223 (0.9952) & 0.9101 (0.8918) & 0.9038 (\textbf{0.0480}) & 0.9355 (0.9994) & 0.9293 (0.9689)\\
0.05 & 0.9920 (1.0000) & 0.9983 (1.0000) & 0.8423 (0.9536) & 0.8320 (0.9989) & 0.8732 (0.9987) & 0.8768 (1.0000)\\
0.1 & 1.0000 (0.9998) & 1.0000 (0.9942) & 0.8696 (0.9995) & 0.9133 (0.9999) & 0.8917 (0.9999) & 0.8784 (0.9999)\\
0.2 & 1.0000 (1.0000) & N/A (N/A) & 0.9747 (0.9706) & 0.9688 (1.0000) & 0.9795 (0.9999) & 0.9684 (0.9999)\\
0.3 & N/A (N/A) & N/A (N/A) & 0.9850 (0.9751) & 0.9884 (0.9754) & 0.9696 (1.0000) & 0.9841 (1.0000)\\
0.4 & N/A (N/A) & 0.9000 (0.9994) & 0.9693 (0.9847) & 0.9688 (1.0000) & 0.9708 (1.0000) & 0.9852 (1.0000)\\
\hline
\hline
\multicolumn{7}{|c|}{Skewed TPC-H 1GB Database}\\
\hline
0.01 & 0.7253 (0.9090) & 0.9067 (1.0000) & 0.7964 (\textbf{0.0291}) & 0.8693 (\textbf{0.0381}) & 0.9163 (0.9970) & 0.9478 (0.9554)\\
0.05 & 0.9903 (1.0000) & 0.9728 (0.9984) & 0.8863 (1.0000) & 0.8872 (0.9944) & 0.9104 (0.9998) & 0.9434 (0.9999)\\
0.1 & 0.9905 (1.0000) & 1.0000 (0.9981) & 0.9081 (0.9541) & 0.9413 (0.9751) & 0.9532 (0.9931) & 0.9602 (0.9999)\\
0.2 & 1.0000 (0.9847) & 1.0000 (0.9803) & 0.9966 (1.0000) & 0.9621 (0.9736) & 0.9326 (0.9679) & 0.9316 (0.9703)\\
0.3 & 1.0000 (0.9624) & 1.0000 (0.9999) & 0.9937 (1.0000) & 0.9951 (0.9000) & 0.9898 (0.9798) & 0.9933 (0.9976)\\
0.4 & 1.0000 (1.0000) & 1.0000 (0.9966) & 0.9954 (1.0000) & 0.9959 (0.9995) & 0.9637 (0.9721) & 0.9924 (0.9641)\\
\hline
\hline
\multicolumn{7}{|c|}{Uniform TPC-H 10GB Database}\\
\hline
0.01 & 0.9441 (0.9929) & 0.9950 (1.0000) & 0.7533 (0.9824) & 0.7147 (0.9697) & \textbf{0.6966} (\textbf{0.0532}) & 0.8128 (\textbf{0.0964})\\
0.02 & 0.9818 (0.9637) & 1.0000 (1.0000) & 0.7937 (0.9996) & 0.7226 (0.9982) & \textbf{0.6818} (\textbf{0.0452}) & 0.8085 (0.9947)\\
0.05 & 1.0000 (1.0000) & N/A (N/A) & 0.8029 (0.8149) & 0.7450 (0.7202) & 0.8093 (0.7995) & 0.8918 (0.9989)\\
0.1 & N/A (N/A) & N/A (N/A) & 0.9839 (0.9831) & 0.9913 (0.9994) & 0.9729 (0.9832) & 0.9818 (0.9996)\\
\hline
\hline
\multicolumn{7}{|c|}{Skewed TPC-H 10GB Database}\\
\hline
0.01 & 0.9805 (1.0000) & 0.8878 (0.9984) & 0.8917 (0.7953) & 0.8718 (\textbf{0.0572}) & 0.8814 (0.9996) & 0.9065 (\textbf{0.0617})\\
0.02 & 0.9759 (0.9280) & 0.9647 (0.9237) & 0.9617 (0.9994) & 0.8844 (0.7031) & 0.9236 (0.9999) & 0.9415 (0.9998)\\
0.05 & 0.8810 (0.9741) & 1.0000 (0.9991) & 0.7759 (0.9970) & 0.8701 (1.0000) & 0.8846 (0.9988) & 0.9834 (0.9999)\\
0.1 & 1.0000 (1.0000) & 1.0000 (1.0000) & 0.9873 (0.9999) & 0.9485 (1.0000) & 0.9919 (0.9999) & 0.9786 (0.9999)\\
\hline
\end{tabular}
\caption{$r_s$ ($r_p$) of selectivity estimates with relative errors above 0.2 (values below 0.7 are highlighted)}
\label{tab:cc:sel-big-err}
\shrink
\end{table*}

\end{document}